\documentclass[english,printmode]{enigtex-article}

\usepackage{dichromatic-settings}

\includecomment{dichroarticle}
\excludecomment{dichrothesis}

\excludecomment{toricyes}
\includecomment{toricno}

\title{Dichromatic state sum models for four-manifolds from pivotal functors}
\author{{\scshape Manuel Bärenz$^1$, John Barrett$^2$}
	\smallskip\\
	1) \begin{minipage}{0.48\textwidth}
		\itshape
		Fakultät für Mathematik\\
		Universität Wien\\
		Oskar-Morgenstern-Platz 1\\
		1090 Wien\\
		Austria
	\end{minipage}
	2) \begin{minipage}{0.48\textwidth}
		\itshape
		School of Mathematical Sciences\\
		University of Nottingham\\
		University Park\\
		NG7 2RD Nottingham\\
		United Kingdom
	\end{minipage}
}

\begin{document}

\maketitle

\begin{abstract}
	A family of invariants of smooth, oriented four-dimensional manifolds is defined via handle decompositions and the Kirby calculus of framed link diagrams.
	The invariants are parametrised by a pivotal functor from a spherical fusion category into a ribbon fusion category.

	A state sum formula for the invariant is constructed via the chain-mail procedure, so a large class of topological state sum models can be expressed as link invariants.
	Most prominently, the Crane-Yetter state sum over an arbitrary ribbon fusion category is recovered, including the nonmodular case.
	It is shown that the Crane-Yetter invariant for nonmodular categories is stronger than signature and Euler invariant.

	A special case is the four-dimensional untwisted Dijkgraaf-Witten model.
	Derivations of state space dimensions of TQFTs arising from the state sum model agree with recent calculations of ground state degeneracies in Walker-Wang models.

	Relations to different approaches to quantum gravity such as Cartan geometry and teleparallel gravity are also discussed.
\end{abstract}

\section{Introduction}

The Crane-Yetter model \cite{CraneYetterKauffman:1997177} is a state sum invariant of four-dimensional manifolds that determines a topological quantum field theory (TQFT).
The purpose of this
\begin{dichroarticle}
paper
\end{dichroarticle}
\begin{dichrothesis}
part
\end{dichrothesis}
is to give a more general construction that puts the Crane-Yetter model in a wider context and allows the exploration of new models, as well as a more thorough understanding of the Crane-Yetter model itself.
There is interest in four-dimensional TQFTs from solid-state physics, where they allow the study of topological insulators, for example in the framework of Walker and Wang \cite{WalkerWang}, which is expected to be the Hamiltonian formulation of the Crane-Yetter TQFT.
The Crane-Yetter model is also the starting point for constructing spin foam models of quantum gravity \cite{BarrettCrane:19983296}.
Therefore the main motivation for this paper is to provide a firmer and more unified basis for a variety of physical models.

A state sum model is a discretised path integral formulation for a lattice theory.
In order to calculate the transition amplitude from one lattice state to another (possibly on a different lattice), a cobordism, or spacetime, from the initial to the final lattice is discretised using a triangulation or a cell complex.
Then the amplitude is the sum of a weight function over states on the discretised cobordism.
A state is typically a labelling of the elements of the discretisation with some algebraic data, for example objects and morphisms in a certain category.

In a topological state sum model, the sum over all states is independent of the particular discretisation chosen, and thus gives rise to a TQFT.
The weight function corresponds to an action functional and is calculated locally, for example per simplex if the discretisation is a triangulation.
This property is motivated by the physical assumption of the action being local,
and is expected to have the far-reaching mathematical consequence that the resulting TQFT is `fully extendable',
which means that it is well-defined on manifolds with corners of all dimensions down to zero.

Topological state sum models are an approach for quantum gravity.
The Turaev-Viro state sum is an excellent model of three-dimensional Euclidean quantum gravity (\cite[Section V.B]{Barrett1995QGTQFT} and \cite{Barrett:200397-Geometrical-measurements}).
As Witten famously remarks \cite[Section 3]{Witten:1989113-2+1gravity},
one would expect any manifestly diffeomorphism-covariant theory to give rise to a topological quantum theory.
So far, no topological state sum has modelled four-dimensional quantum gravity in a satisfactory way.
The most prominent topological state sum model remains the $U_qsl(2)$-Crane-Yetter state sum; however this is not considered a gravity model.
It was shown to reduce to the signature \cite{CraneYetterKauffman:1997177} and the Reshetikhin-Turaev theory on the boundary \cite{Barrett:2007-Observables}.
As a consequence of this, the state spaces attached to the boundary manifolds are only one-dimensional,
whereas in a gravity theory one would expect a large state space containing many graviton modes.
The more general framework developed here suggests some different Crane-Yetter type models that may be related to approaches such as teleparallel gravity \cite{Baezderek:teleparallel}.

\subsection{The Crane-Yetter invariant and its dichromatic generalisation}
\label{Crane Yetter}

In three-dimensional topology, the Turaev-Viro state sum invariant distinguishes even some homotopy-equivalent three-manifolds:
By \cite[Proposition 2]{Sokolov1997TVLensSpaces},
the lens spaces $L(7,1)$ and $L(7,2)$,
which are homotopy equivalent, but not homeomorphic,
have different values for the Turaev-Viro invariant.
However the Crane-Yetter invariant of four-manifolds for \emph{modular} categories,
as it was originally defined,
is just a function of the signature and the Euler characteristic of the manifold
\cite[Proposition 6.2]{CraneYetterKauffman:1997177}.

A closer look at the construction reveals a possible explanation why this is the case.
By the Morse theorem, smooth manifolds admit handle decompositions.
(Additionally, there is a canonical handle decomposition determined by any triangulation, by thickening the dual complex.)
Different handle decompositions of the same manifold can be related by a sequence of handle slides and cancellations.
Thus, one can construct a manifold invariant by assigning numbers to handle decompositions;
if the numbers do not change under the handle moves, they define an invariant.

Handle decompositions can be described by Kirby diagrams.
These are framed links where the components of the link represent the 1- and 2-handles.
For the modular Crane-Yetter invariant, the components of the link are each labelled by the Kirby colour of the ribbon fusion category $\mathcal{C}$ that determines the invariant.
By the universal property of the tangle category \cite{Shum:1994TortileTensorCats}, this can be interpreted as diagrammatic calculus in $\mathcal{C}$.
Evaluating the diagram and multiplying by a normalisation gives the invariant.

Since the 2-handles are treated in the same way as the 1-handles,
there is a redundancy in the construction of the modular Crane-Yetter invariant:
it does not change if all 1-handles are replaced by 2-handles in the link diagram.
But such a replacement radically changes the topology of the manifold and ensures, for example,
that every manifold has the same modular Crane-Yetter invariant as a simply-connected one.
Consequently, the invariant cannot even detect the first homology.

The solution is to define invariants that label the 1- and 2-handles with different objects in the category.
Petit's ``dichromatic invariant'' \cite{Petit:dichromatic} does exactly this:
in addition to the ribbon fusion category,
one also chooses a full fusion subcategory and labels the 2-handles with the Kirby colour of the subcategory.
Whether this change actually improves the invariant remained unstudied at the time.
It will be shown in Section \ref{non-simply-connected} that it does indeed lead to a stronger invariant that is sensitive to the fundamental group and can thus distinguish manifolds with the same signature and Euler characteristic.
Now one can indeed pinpoint the improvement of the invariant as due to the differing labels on 1-handles and 2-handles.
As a bonus, the general Crane-Yetter invariant is recovered as a special case of the dichromatic invariant.
Previously, no description of it in terms of Kirby calculus was known for nonmodular ribbon categories.

A generalisation of the dichromatic invariant is presented here and translated into a state sum model.
Instead of a ribbon fusion subcategory,
the generalisation is to use a pivotal functor from a spherical fusion category to a ribbon fusion category.
The 1-handles are still labelled with the Kirby colour of the target category,
but the 2-handles are labelled with the Kirby colour of the source category,
with the functor applied to it.
\subsection{Outline}
\begin{dichroarticle}
In Section \ref{preliminaries}, the common definitions such as spherical and ribbon fusion categories and their graphical calculus are recalled.
\end{dichroarticle}
\begin{dichrothesis}
In Section \ref{preliminaries}, the graphical calculus of spherical and ribbon fusion categories is recalled.
\end{dichrothesis}
Various notational conventions are established.

In Section \ref{section invariant}, the \emph{sliding lemma} from spherical and ribbon fusion categories is generalised.
The original lemma allows for sliding the identity morphism of any object over an encirclement by the Kirby colour of the category.
The generalised lemma generalises this to an encirclement by the image of a Kirby colour under a pivotal functor.
This generalisation will be a key step in the proof of invariance (Section \ref{proof of invariance}) of the \emph{generalised dichromatic invariant} (Definition \ref{definvariant}) of smooth, oriented, closed four-manifolds.
The section concludes with some general properties of the invariant and a motivating special case, Petit's dichromatic invariant (Example \ref{dichromatic}).

Many functors lead to the same invariant, and a general situation in which this is the case is presented in Section \ref{sec:simplification}.
This often leads to a simplification of the invariant, especially when the functor and both categories are unitary, or when the target category is modularisable.

If the target category of the functor is modularisable, which is often the case, the generalised invariant can also be cast in the form of a state sum.
In Section \ref{SSM}, this state sum formula \eqref{ssm formula} is derived using the chain mail technique.

Section \ref{examples invariant} is a non-exhaustive survey of several different examples of the generalised dichromatic invariant.
The Crane-Yetter state sum is recovered as a special case, both for modular and nonmodular ribbon fusion categories.
For the nonmodular Crane-Yetter invariant, a chain mail construction was not previously known.
A further special case is Dijkgraaf-Witten theory without a cocycle,
implying that the invariant can be sensitive to the fundamental group.
The Dijkgraaf-Witten example is then generalised to group homomorphisms.

There is a discussion in Section \ref{other-models} of how the present framework could connect to Walker-Wang models and state sum models used in the study of quantum gravity such as spin foam models.
Relations to Cartan geometry and teleparallelism are discussed as well.

Finally, a handy overview of the different known special cases of the generalised dichromatic invariant is given as a table in Section \ref{outlook}, together with some comments on the results.

\section{Preliminaries}
\label{preliminaries}
\subsection{Monoidal categories with additional structure}

In mathematical physics, one encounters a multitude of linear monoidal categories with additional structure and functors preserving this structure.
Usually, the category $\Vect$ of finite dimensional vector spaces over $\Co$ serves as a trivial example for these.
The additional structures often arise as special cases of higher categorical structures,
for example, monoidal categories are bicategories with one object and braided categories are in some sense tricategories with one 1-morphism.
This beautiful motivation is explained more closely in the literature, e.g. \cite[section B.3]{Schommer-Pries:PhD}.
Here the definitions are given in a closely related manner by discussing their suitability for graphical calculus.
Monoidal categories are needed for a graphical calculus of one-dimensional ribbon tangles in two dimensions; similarly one needs the braided structure for evaluating tangle diagrams in three dimensions.
An overview of most commonly used definitions of monoidal categories with additional structure,
together with their graphical calculus,
can be found in \cite{Selinger:graphical}.

\subsubsection{Semisimple and linear categories}

\begin{defn}
	A $\Co$-\textbf{linear category} is a category enriched in $\Vect_\Co$.
	If not mentioned otherwise, all categories in this work are $\Co$-linear categories and all functors are \textbf{linear functors}, that is, functors in the enriched category.
	This implies that they are linear on the morphism spaces and preserve direct sums.
\end{defn}
\begin{defn}
	An object $X \in \ob \mathcal{C}$ is called \textbf{simple}
	if $\mathcal{C}(X,X) \cong \Co$.
\end{defn}
\begin{examples}
	\begin{itemize}
		\item In $\Vect$, $\Co$ is the only simple object up to isomorphism.
		\item In $\Rep(G)$, the representation category of a finite group $G$,
		the simple objects are the irreducible representations.
	\end{itemize}
\end{examples}
Note that simple objects are called scalar objects in \cite{Petit:dichromatic}.
\begin{defn}
	\label{def:semisimple}
	A linear category $\mathcal{C}$ is called \textbf{semisimple} if it has biproducts, idempotents split (i.e. it has subobjects) and there is a set of inequivalent simple objects $\Lambda_\mathcal{C}$ such that for each pair of objects $X$, $Y$, the map
	\[\Phi\colon \bigoplus_{Z\in\Lambda_\mathcal{C}} \mathcal C(X,Z)\otimes \mathcal C(Z,Y)\to \mathcal C(X,Y)\]
	obtained by composition and addition is an isomorphism.
	If the set $\Lambda_\mathcal{C}$ is finite, then the category is called \textbf{finitely semisimple}. 
\end{defn}
\begin{remark}
	The requirements of biproducts and subobjects in this definition are not very restrictive.
	According to the discussion in \cite{Mueger:FromSubfactorsI},
	any category that satisfies all of the conditions in the definition of a semisimple category except for the existence of biproducts and subobjects can be embedded as a full subcategory of a semisimple category.
\end{remark}
\begin{exam}
	For every finite group $G$, $\Rep(G)$ is finitely semisimple.
	The simple objects are the irreducible representations.
\end{exam}

\begin{lemma}
	Let $Z_1$ and $Z_2$ be two nonisomorphic simple objects.
	Then there are no nontrivial morphisms between them,
	i.e. ${\mathcal{C}(Z_1, Z_2) = 0}$.
\end{lemma}
\begin{proof}
	Decompose $\mathcal{C}(Z_1, Z_2)$ according to Definition \ref{def:semisimple}.
	Both $\mathcal{C}(Z_1, Z_2) \otimes \mathcal{C}(Z_2, Z_2)$ and $\mathcal{C}(Z_1, Z_1) \otimes \mathcal{C}(Z_1, Z_2)$ occur as summands.
	But since $\mathcal{C}(Z_1, Z_1) \cong \mathcal{C}(Z_2, Z_2) \cong \Co$,
	$\mathcal{C}(Z_1, Z_2) \otimes \Co^2$ is a subspace of $\mathcal{C}(Z_1, Z_2)$,
	which implies that $\mathcal{C}(Z_1, Z_2) \cong 0$.
\end{proof}
\begin{definition}
	For a simple object $Z$ and any object $X$ in a linear category, there is a bilinear pairing:
	\begin{align*}
		&(-,-) \colon \mathcal{C}(Z,X) \times \mathcal{C}(X,Z) \to \Co\\
		&(f,g) \cdot 1_Z = g \circ f
	\end{align*}
	The $-$ are placeholders.
\end{definition}
\begin{lemma}
	\label{lem:ss-pairing}
	In a semisimple category, the bilinear pairing is non-degenerate.
\end{lemma}
\begin{proof}
	Let $g\colon X \to Z$ such that all $f\colon Z \to X$ satisfy $g \circ f = 0$.
	Then decompose $1_X = \sum_{Z',i} \alpha_{Z'}^i \circ \alpha_{Z',i}$ according to Definition \ref{def:semisimple},
	which implies $g = g \circ 1_X = g \circ \sum_{Z',i} \alpha_{Z'}^i \circ \alpha_{Z',i}$.
	From the previous lemma we know that if $Z$ and $Z'$ are not isomorphic then $g \circ \alpha_{Z'}^i = 0$,
	therefore the sum reduces to $g \circ \sum_{i} \alpha_{Z}^i \circ \alpha_{Z,i}$.
	But $\alpha_{Z}^i\colon Z \to X$, so by assumption $g \circ \alpha_Z^i = 0$ and therefore $g = 0$.
	
	An analogous argument holds for $f$.
\end{proof}

\subsubsection{Monoidal categories, functors and natural transformations}

\begin{defn}
	A \textbf{monoidal category} consists of:
	\begin{itemize}
		\item A category $\mathcal{C}$,
		\item a functor ${- \otimes -\colon \mathcal{C} \times \mathcal{C} \to \mathcal{C}}$ called the \textbf{monoidal product},
		\item a unit object $I$ called the \textbf{monoidal identity},
		\item natural associativity isomorphisms $\alpha_{X,Y,Z}\colon (X \otimes Y) \otimes Z \to X \otimes (Y \otimes Z)$
		and natural unit isomorphisms $\lambda_X \colon I \otimes X \to X$
		and $\rho_X \colon X \otimes I \to X$
		subject to coherence conditions which can be found e.g. in \cite[Section 3.1]{Selinger:graphical}.
	\end{itemize}
	In a \textbf{strict monoidal category}, the coherence morphisms $\alpha$, $\lambda$ and $\rho$ are all identity morphisms.
\end{defn}
If a monoidal category is also linear,
$\otimes$ is assumed to be bilinear:
\begin{align}
	(f + g) \otimes h & = (-\otimes-)(f + g, h) = f \otimes h + g \otimes h\nonumber\\
	f \otimes (g + h) & = (-\otimes-)\mathrlap{(f, g + h)}\hphantom{(f + g, h)} = f \otimes g + f \otimes h
\end{align}

In the graphical calculus for monoidal categories,
morphisms $f\colon X \to Y$ are drawn as boxes and lines in the plane, from the bottom to the top:
\begin{align}
	1_X &=
	\tikzbo{
		\draw[diagram,directed] (0,-1) node[below] {$X$} -- (0,1);
	}
	&
	f & =
	\tikzb{
		\node[draw] (B) at (0,0) {$f$};
		\draw[diagram,directed2] (0,-1) node[below] {$X$} -- (B) -- (0,1) node[above] {$Y$};
	}
	&
	f_1 \otimes f_2 & =
	\tikzbo{
		\begin{scope}[xshift=-1cm]
			\foreach \i in {1,2} {
				\node[draw] (C\i) at (\i,0) {$f_\i$};
				\draw[diagram,directed2] (\i,-1) node[below] {$X_\i$} -- (C\i) -- (\i,1) node[above] {$Y_\i$};
			}
		\end{scope}
	}
\end{align}
The upward-pointing arrow on the lines is optional at this point but will be a useful device when duals are introduced.
The coherence morphisms are not shown in the diagrammatic calculus.
This is due to MacLane's famous coherence theorem which states that any composition of coherence morphisms between two given objects is unique \cite{MacLane}.
Hence there is no ambiguity in the way the coherence morphisms are inserted.
Also, the coherence theorem shows that every monoidal category is monoidally equivalent to a strict monoidal category.
Hence one can alternatively view the diagrammatic calculus as determining morphisms in the equivalent strict category.
Throughout the paper, monoidal categories (possibly with extra structure) will be indicated by the name of the mere category whenever standard notation for all the additional data is used,
and they will often be assumed to be strict.

\begin{defn}
	A \textbf{monoidal functor} is a tuple $(F, F^2, F^0)$, where
	\begin{itemize}
		\item $F\colon \mathcal{C} \to \mathcal{D}$ is a functor between monoidal categories,
		\item $F^2_{X,Y} \colon FX \otimes_{\mathcal D} FY \Rightarrow F(X \otimes_{\mathcal C} Y)$ is a natural isomorphism,
		\item $F^0\colon I_\mathcal{D} \to FI_\mathcal{C}$ is an isomorphism in $\mathcal{D}$.
	\end{itemize}
	$F^2$ and $F^0$ are required to commute with the coherence morphisms,
	see e.g. \cite[Section 3.1]{Selinger:graphical}.
	A \textbf{monoidal natural transformation} is a natural transformation that commutes with $F^0$ and $F^2$.
\end{defn}
Note that here $F^2$ and $F^0$ are assumed to be isomorphisms.
Such functors are also sometimes called ``strong monoidal''.

\subsubsection{Rigid and fusion categories}

\begin{defn}
	A \textbf{duality} is a quadruple $(X, Y, \ev\colon X \otimes Y \to I, \coev\colon I \to Y \otimes X)$ satisfying the ``snake identities'':
	\begin{align}
		(\ev \otimes 1_X) \circ (1_X \otimes \coev) &= 1_X\nonumber\\
		(1_Y \otimes \ev) \circ (\coev \otimes 1_Y) &= 1_Y\label{snake identities}
	\end{align}
	In this situation, $(X, \ev, \coev)$ is called the left dual of $Y$, and $(Y, \ev, \coev)$ the right dual of $X$.
	The morphisms $\ev$ and $\coev$ are called ``evaluation'' and ``coevaluation'', respectively.
	(In the context of adjunctions, they are also called ``unit'' and ``counit''.)
\end{defn}
\begin{defn}
	A monoidal category with left (right) duals for every object is called a \textbf{left (right) rigid category}.
	A \textbf{rigid}, or ``autonomous'' category is a category that is left rigid and right rigid, i.e., every object has a left and a right dual.
\end{defn}
\begin{defn}
	Finitely semisimple rigid categories with simple $I$ are known as \textbf{fusion categories}.
\end{defn}
In this work, each object $X$ in a rigid category will have a particular choice of duals.
The right dual is denoted $(X^*, \ev_X, \coev_X)$ and the left dual $(\prescript{*}{}{X}, \widetilde{\ev}_X, \widetilde{\coev}_X)$.
Pre- and postcomposing morphisms with $\ev$ and $\coev$ (resp. $\widetilde{\ev}$ and $\widetilde{\coev}$) defines right (resp. left) dual contravariant op-monoidal functors $-^*$ (resp. $\prescript{*}{}{-}$).
They are contravariant in the sense that source and target are switched,
and op-monoidal in the sense that the monoidal product is reversed via canonical isomorphisms
$\delta_{X,Y}\colon (X \otimes Y)^* \cong Y^* \otimes X^*$.

Applying a monoidal functor $\left(F, F^2, F^0\right)$ to the snake identities shows that dualities are preserved, i.e. that the following morphism is an evaluation:
\[FX \otimes FY \xrightarrow{F^2_{X,Y}} F\left(X \otimes Y\right) \xrightarrow{F\ev} FI_\mathcal{C} \xrightarrow{\left(F^0\right)^{-1}} I_\mathcal{D}\]
A similar statement holds for the coevaluation.
Proving this requires all the naturality axioms of a monoidal functor.

A standard result on dualities is that any two duals of a given object $X$ are canonically isomorphic.
Applying this to $F$ shows \cite{Pfeiffer} that there are canonical isomorphisms for the right duals
\begin{equation}
	u_X\colon F\left(X^*\right) \to (FX)^*
	\label{right dual canonical iso}
\end{equation}
determined by $F$.
These satisfy the defining equations
\begin{align}
	\ev_{FX} &= \left(F^0\right)^{-1}\circ F\ev_X \circ  F^2_{X,X^*} \circ \left(1\otimes u_X^{-1}\right)\\
	\coev_{FX} &= \left(u_X \otimes 1\right)\circ \left(F^2_{X^*,X}\right)^{-1} \circ F\coev_X \circ F^0
\end{align}

There are also separate canonical isomorphisms in a similar way for the left duals.

\subsubsection{Pivotal and spherical categories}

There exist rigid categories in which every left dual is also a right dual,
i.e. $X^* \cong {^*}\hspace{-1.4pt}X$.
Since there already exist canonical natural isomorphisms $l_X\colon X \xrightarrow{\cong} (\leftdual{X})^*$ and $\tilde{l}_X\colon X \xrightarrow{\cong} \leftdual{(X^*)}$
in any rigid category,
isomorphisms between left and right duals are equivalent to isomorphisms to the double dual, $X \cong X^{**}$.
Choosing such an isomorphism naturally and monoidally for each object leads to the following definition.
\begin{defn}
	A \textbf{pivotal category} is a right rigid category $\mathcal{C}$ (with chosen right duals) together with a monoidal natural isomorphism $i\colon 1_\mathcal{C} \to -^{**}$, the \textbf{pivotal structure}.
	They are also called ``sovereign'' categories.
\end{defn}
\begin{lem}
	A pivotal category is also left rigid, and thus rigid, with the following choice of left dual:
	\begin{align}
		{^*}\hspace{-1.4pt}X &\coloneqq X^*\\
		\widetilde{\ev}_X &\coloneqq \ev_{X^*} \circ \left( 1_{X^*} \otimes i_X \right)\\
		\widetilde{\coev}_X &\coloneqq \left( i_X^{-1} \otimes 1_{X^*} \right) \circ \coev_{X^*}
	\end{align}
\end{lem}
In a pivotal category,
evaluation and coevaluation morphisms are drawn as caps and cups.
The arrow in the diagram is an orientation for the line that points towards the dual object.
\begin{align}
	\ev_X &= \tikzbo{\dual{below}{X}{X^*}} &
	\coev_X &= \tikzbo{\begin{scope}[yscale=-1]\dual{above}{X^*}{X}\end{scope}} 
\end{align}
The arrow notation means that it is possible to regard the object $X$ as a label on the whole line (rather than one end of it).
The convention at the ends of the line is that an upward-pointing arrow indicates $X$ and a downward-pointing arrow $X^*$.

In this graphical calculus, the snake identities now become:
\newcommand{\snakeradius}{0.5}
\newcommand{\snakeid}[1]{
	\tikzbo{
		\begin{scope}[yscale=#1]
			\draw[directed2,thick] (-2*\snakeradius,-\snakeradius) -- (-2*\snakeradius,0) arc (180:0:\snakeradius cm) arc (180:360:\snakeradius cm) -- (2*\snakeradius,\snakeradius);
		\end{scope}
	}
	& =
	\tikzbo{
		\begin{scope}[yscale=#1]
			\draw[directed,thick] (0,-\snakeradius) -- (0,\snakeradius);
		\end{scope}
	}
}
\begin{align}
	\snakeid{1}
	&
	\snakeid{-1}
\end{align}
Indeed, every identity of strings that is true as an isotopy in the plane is true for morphisms in a pivotal category.
\begin{defn}
	Left and right traces $\tr_L, \tr_R\colon \mathcal{C}(X,X) \to \mathcal{C}(I,I) \cong \Co$ can be defined with a pivotal structure:
	\newcommand{\righttrace}{
		\node[draw] (B) at (0,0) {$f$};
		\draw[directed, thick] (B.north) arc (180:0:0.5cm) -- +($(B.south)-(B.north)$) arc (0:-180:0.5cm);
	}
	\begin{align}
		\tr_R(f) & \coloneqq \tikzb{\righttrace}
		&&= \ev_X \circ \left( f \otimes 1_{X^*} \right) \circ \widetilde{\coev}_X \nonumber\\
		&&&= \ev_X \circ \left( \left( f \circ i_X^{-1} \right) \otimes 1_{X^*} \right) \circ \coev_{X^*}
		\\
		\tr_L(f) & \coloneqq \tikzb{
			\begin{scope}[xscale=-1]
				\righttrace
			\end{scope}
		}
		&&= \widetilde{\ev}_X \circ \left( 1_{\prescript{*}{}{X}} \otimes f \right) \circ \coev_X \nonumber\\
		&&&= \ev_{\prescript{*}{}{X}} \circ \left( 1_{\prescript{*}{}{X}} \otimes \left( i_X \circ f \right) \right) \circ \coev_X
	\end{align}
\end{defn}
There are pivotal categories for which $\tr_R \neq \tr_L$ for some objects.
Spherical categories eliminate this discrepancy.
\begin{defn}
	A \textbf{spherical category} is a pivotal category with $\tr_R = \tr_L$ for every object.
	This trace will then simply be called $\tr$.
	The pivotal structure of a spherical category is also called a ``spherical structure''.
	The \textbf{dimension} of an object $X$ is defined as $\qdim{X} \coloneqq \tr\left(1_X\right)$.
	It is also called ``categorical'' dimension,
	or, for representations of Hopf algebras, ``quantum'' dimension.
\end{defn}
The diagram for the dimension of an object is a circle.
Note that because of sphericality, it is not necessary to specify a direction on the circle.
\newcommand{\dimcircle}[2][{}]{
	\node[draw,diagram,circle,label={[label distance=-3]above right:#2},#1] (B) at (0,0) {\phantom{Bla}};
}
\begin{equation}
	\qdim{X}
	= \tr (1_X) = \tikzb{\dimcircle{$X$}}
\end{equation}

Note that the dimension of a simple object is known to be nonzero in fusion categories \cite{OnFusionCategories}.
This follows from the facts that for a simple object $Z$ the spaces $\mathcal{C}(I,Z\otimes Z^*)$ and $\mathcal{C}(Z\otimes Z^*,I)$ have dimension $1$,
evaluations and coevaluation are non-zero elements of these spaces,
and Lemma \ref{lem:ss-pairing}.
\begin{rema}
	The name ``spherical'' arises from the fact that the diagram of a morphism can be embedded on the 2-sphere, and every isotopy on the sphere amounts to a relation in the category.
	The additional axiom of a spherical category corresponds to moving a strand ``around the back'' of the sphere.
	However, the spherical axiom implies further identities that don't come from isotopies on the sphere.
\end{rema}
\begin{defn}
	\label{def:sph-pair}
	Let $X$ and $Y$ be two arbitrary objects in a spherical fusion category.
	The \textbf{spherical pairing} of two morphisms $f\colon X \to Y$ and $g\colon Y \to X$ is defined as
	\begin{equation}
		\braket{f}{g} \coloneqq \tr(g \circ f) = \tr(f \circ g) \label{spherical pairing}
	\end{equation}
\end{defn}

\begin{lem}
	\label{lem:sph-pair}
	The spherical pairing on a spherical fusion category is nondegenerate.
\end{lem}
\begin{proof}
	With the notation from Definitions \ref{def:semisimple} and \ref{def:sph-pair},
	decompose
	$f=\sum_{Z,i} \beta_{Z}^i \circ \alpha_{Z,i}$ and $g = \sum_{Z',j} \delta_{Z'}^j \circ \gamma_{Z',j}$.
	Then
	\begin{align*}
		\braket{f}{g} &= \sum_{\mathclap{Z,i,Z',j}} \tr \left( \delta_{Z'}^j \circ \gamma_{Z',j} \circ \beta_Z^i \circ \alpha_{Z,i} \right)
		&& = \sum_{\mathclap{Z,i,Z',j}} \tr \left( \gamma_{Z',j} \circ \beta_Z^i \circ \alpha_{Z,i} \circ \delta_{Z'}^j \right) 
		\intertext{But $\gamma_{Z',j} \circ \beta_Z^i$ is a map from $Z$ to $Z'$, and so is non-zero only if $Z=Z'$.
		In this case, it is equal to $\left(\beta_Z^i, \gamma_{Z,j} \right) 1_Z$,
		thus the expression reduces to}
		&= \sum_{Z,i,j} \left( \beta_Z^i, \gamma_{Z,j} \right) \tr \left(\alpha_{Z,i} \circ \delta_Z^j \right)
		&& = \sum_{Z,i,j} \left( \beta_Z^i, \gamma_{Z,j} \right) \left( \delta_Z^j, \alpha_{Z,i} \right) \qdim Z
	\end{align*}
	The dimensions $\qdim{Z}$ of simple objects are nonzero,
	hence with Lemma \ref{lem:ss-pairing} this is non-degenerate.
\end{proof}

\begin{defn}
	A \textbf{pivotal functor} $F\colon \mathcal{C} \to \mathcal{D}$ is a strong monoidal functor preserving the pivotal structure
	(and thus the isomorphisms between left and right duals)
	up to canonical isomorphisms.
	More specifically, the following diagram must commute:
	\begin{equation}
		\begin{tikzcd}
			FX
			\arrow[r,"i_{FX}"]
			\arrow[d,"Fi_X", swap]
				& (FX)^{**}
				\arrow[d,"u_X^*"]\\
			F(X^{**})
			\arrow[r,"u_{X^*}"]
				& \left(F(X^*)\right){^*} \\
		\end{tikzcd}
		\label{pivotal functor}
	\end{equation}
	In this diagram, $u$ is the canonical isomorphism from \eqref{right dual canonical iso}.
\end{defn}
\begin{lem}
	Pivotal functors preserve traces and therefore dimensions and the spherical pairing.
	As elements of $\Co \cong \mathcal{C}\left(I_\mathcal{C}, I_\mathcal{C}\right) \cong \mathcal{D}\left(I_\mathcal{D}, I_\mathcal{D}\right)$,
	it follows that for any endomorphism $f\colon X \to X$ the following holds:
	\begin{equation}
		\tr(f) = \tr(Ff)
	\end{equation}
\end{lem}

	\begin{proof}
		Insert the isomorphism $\mathcal{C}(I_\mathcal{C}, I_\mathcal{C}) \xrightarrow{F} \mathcal{D}\left(FI_\mathcal{C}, FI_\mathcal{C}\right) \xrightarrow{F^0 \circ - \circ \left(F^0\right)^{-1}} \mathcal{D}\left(I_\mathcal{D}, I_\mathcal{D}\right)$ explicitly.
		It is now necessary to prove $(F\tr(f)) \circ F^0 = F^0 \circ \tr(Ff)$.
		\begin{align*}
			&\mathrel{\hphantom{=}} F\tr(f) \circ F^0 \\
			&= F\left( \ev_X \circ \left( \left( f \circ i_X^{-1} \right) \otimes 1_{X^*} \right) \circ \coev_{X^*} \right) \circ F^0 \\
			&= F \ev_X \circ F^2_{X,X^*} \circ \left( \left( Ff \circ Fi_X^{-1} \right) \otimes 1_{F\left(X^*\right)} \right) \circ \left(F^2_{X^{**},X^*}\right)^{-1} \circ F\coev_{X^*} \circ F^0\\
			&= F^0 \circ \ev_{FX} \circ \left( \left( Ff \circ Fi_X^{-1} \circ u_{X^*}^{-1} \right) \otimes u_X \right) \circ \left(F^2_{X^{**},X^*}\right)^{-1} \circ \coev_{F\left(X^*\right)}\\
			&= F^0 \circ \ev_{FX} \circ \left( \left( Ff \circ i_{FX}^{-1} \circ \left(u^*_X\right)^{-1} \right) \otimes u_X \right) \circ \coev_{F\left(X^*\right)}\\
			&= F^0 \circ \ev_{FX} \circ \left( \left( Ff \circ i_{FX}^{-1} \right) \otimes 1_{(FX)^*} \right) \circ \coev_{(FX)^*}\\
			& = F^0 \circ \tr(Ff)
		\end{align*}
	\end{proof}

\subsubsection{Braided, balanced and ribbon categories}

\begin{defn}
	A \textbf{braided monoidal category} (or simply ``braided category'') is a monoidal category $\mathcal{C}$ with a dinatural isomorphism $c$ (the ``braiding'') with components $c_{X,Y}\colon X \otimes Y \to Y \otimes X$ satisfying compatibility axioms with the monoidal product,
	called the \emph{braid axioms}, or hexagon identities:
	\begin{equation}
		\begin{tikzcd}[column sep=-0.6cm]
				& (X \otimes Y) \otimes Z
				\arrow{ld}[swap]{\alpha_{X,Y,Z}}
				\arrow{rd}{c_{X,Y} \otimes 1_Z}
			\\
			X \otimes (Y \otimes Z)
			\arrow{d}[swap]{c_{X,Y \otimes Z}}
				&
					& (Y \otimes X) \otimes Z
					\arrow{d}{\alpha_{Y,X,Z}}
			\\
			(Y \otimes Z) \otimes X
			\arrow{rd}[swap]{\alpha_{Y,Z,X}}
				&
					& Y \otimes (X \otimes Z)
					\arrow{ld}{1_Y \otimes c_{X,Z}}
			\\
				& Y \otimes (Z \otimes X)
		\end{tikzcd}
		\begin{tikzcd}[column sep=-0.6cm]
				& (X \otimes Y) \otimes Z
				\arrow{ld}[swap]{\alpha_{X,Y,Z}}
				\arrow{rd}{c^{-1}_{Y,X} \otimes 1_Z}
			\\
			X \otimes (Y \otimes Z)
			\arrow{d}[swap]{c^{-1}_{Y \otimes Z,X}}
				&
					& (Y \otimes X) \otimes Z
					\arrow{d}{\alpha_{Y,X,Z}}
			\\
			(Y \otimes Z) \otimes X
			\arrow{rd}[swap]{\alpha_{Y,Z,X}}
				&
					& Y \otimes (X \otimes Z)
					\arrow{ld}{1_Y \otimes c^{-1}_{Z,X}}
			\\
				& Y \otimes (Z \otimes X)
		\end{tikzcd}
		\label{eq:braid axioms}
	\end{equation}
\end{defn}
As the name suggests, the graphical calculus for braidings consists of strings which can cross each other:

\begin{align}
	c_{X,Y} &=
	\tikzbo{
		\draw[diagram,directed2] (1,-1) node[below] {$Y$} to[out=90,in=-90] (0,1);
		\draw[diagram,directed2] (0,-1) node[below] {$X$} to[out=90,in=-90] (1,1);
	}&
	c_{Y,X}^{-1} &=
	\tikzbo{
		\draw[diagram,directed2] (0,-1) node[below] {$X$} to[out=90,in=-90] (1,1);
		\draw[diagram,directed2] (1,-1) node[below] {$Y$} to[out=90,in=-90] (0,1);
	}
\end{align}

The coherence isomorphisms $\alpha$ are invisible in the graphical calculus.
Therefore, the braid axioms become
\newcommand{\stringdist}{0.7}
\begin{align}
	\tikzbo{
		\draw[thickdirected2,doubleknot] (1,-1) node[below] {$Y \otimes Z$} to[out=90,in=-90] +(-1,2);
		\draw[diagram,directed2] (0,-1) node[below] {$X$} to[out=90,in=-90] (1,1);
	}
	& =
	\tikzbo{
		\draw[diagram,directed] (\stringdist,-1) node[below] {$Y$} to[out=90,in=-90] +(-\stringdist,2);
		\draw[diagram,directed] (2*\stringdist,-1) node[below] {$Z$} to[out=90,in=-90] +(-\stringdist,2);
		\draw[diagram,directed] (0,-1) node[below] {$X$} to[out=90,in=-90] (2*\stringdist,1);
	}
	&
	\tikzbo{
		\draw[diagram,directed2] (0,-1) node[below] {$X$} to[out=90,in=-90] (1,1);
		\draw[thickdirected2,doubleknot] (1,-1) node[below] {$Y \otimes Z$} to[out=90,in=-90] +(-1,2);
	}
	& =
	\tikzbo{
		\draw[diagram,directed] (0,-1) node[below] {$X$} to[out=90,in=-90] (2*\stringdist,1);
		\draw[diagram,directed] (\stringdist,-1) node[below] {$Y$} to[out=90,in=-90] +(-\stringdist,2);
		\draw[diagram,directed] (2*\stringdist,-1) node[below] {$Z$} to[out=90,in=-90] +(-\stringdist,2);
	}
\end{align}

\begin{defn}
	A \textbf{balanced monoidal category} is a braided category $\mathcal{C}$ with a natural isomorphism $\theta\colon 1_\mathcal{C} \Rightarrow 1_\mathcal{C}$,
	the \textbf{twist},
	satisfying the \emph{balance equation}:
	\begin{equation}
		\theta_{X \otimes Y} = c_{Y,X} \circ c_{X,Y} \circ \left( \theta_Y \otimes \theta_X \right)\label{balanced}
	\end{equation}
\end{defn}
(This term should not be confused with the unrelated concept of a ``balanced category'',
where every morphism that is mono and epi is also an isomorphism.)
\begin{theorem}
	\label{thm:rigid balanced pivotal bijection}
	In a rigid, braided category,
	there exists a (noncanonical) bijection between twists satisfying the balance equation and pivotal structures.
	For a given pivotal structure,
	one possible balanced structure can be defined as:
	\begin{equation}
		\theta_X \coloneqq
		\tikzbo{
			\coordinate (R) at (1,0);
			\node[draw, inner sep=2pt] (M) at (0,0.9) {$i_X^{-1}$};
			\draw[diagram]
				   (M)
				-- (0,1.5);
			\draw[diagram,directed=0.95,looseness=2]
				(R)
				to[out=-90,in=-90] (M);
			\draw[diagram,looseness=2]
				(0,-1.5)
				node[below] {$X$}
				-- (0,-0.5)
				to[out=90,in=90] (R);
		}
		\label{eq:balanced from pivotal}
	\end{equation}
	For further details,
	consult e.g. \cite[Lemma 4.20]{Selinger:graphical},
	and the sources cited therein.
\end{theorem}
There are other possibilities to construct a pivotal structure from a balanced structure,
but they will coincide in the case of the following definition.
\begin{defn}
	A \textbf{ribbon category} is a balanced monoidal, rigid category satisfying the \emph{ribbon equation}:
	\begin{equation}
		\theta_{X^*} = \theta_X^* \label{ribbon}
	\end{equation}
	Ribbon categories are also called ``tortile'' categories.
\end{defn}

The graphical representation of the twist is usually a ribbon that has been twisted by $2 \pi$.
The thickening to two-dimensional ribbons is meant to express the fact that the twist cannot be undone by an ambient isotopy in three-dimensional space.
In two-dimensional diagrams,
ribbons can still be drawn as lines
-- possibly with crossings --
when the blackboard framing is implicitly assumed.
After recognising that the pivotal structure is a coherence and can be omitted from \eqref{eq:balanced from pivotal},
the diagram for the twist becomes:
\begin{equation}
	\theta_X =
	\tikzbo{
		\coordinate (T) at (0.2,0);
		\coordinate (N) at (0,0.1);
		\draw[diagram,directed,looseness=2] (T) to[out=-90,in=-90] (N) -- (0,1);
		\draw[diagram,directed,looseness=2] (0,-1) node[below] {$X$} -- ($(0,0)-(N)$) to[out=90,in=90] (T);
	}
\end{equation}
The graphical representations of the balance equation and the ribbon equation are thus:
\begin{align}
	\tikzbo{
		\coordinate (T) at (0.4,0);
		\coordinate (N) at (0,0.2);
		\draw[thickdirected=0.7,doubleknot,looseness=2] (T) to[out=-90,in=-90] (N) -- (0,1);
		\draw[thickdirected,doubleknot,looseness=2] (0,-1) node[below] {$X \otimes Y$} -- ($(0,0)-(N)$) to[out=90,in=90] (T);
	}
	& =
	\tikzbo{
		\coordinate (T) at (0.2,0);
		\coordinate (N) at (0,0.1);
		\coordinate (C) at (0,0.3);
		\draw[diagram,directed=0.3] ($(1,0)+(C)$) to[out=90,in=-90] (0,1);
		\draw[diagram,directed=0.3] ($(C)$) to[out=90,in=-90] (1,1);
		\draw[diagram,directed=0.9] ($(1,-1)+2*(C)$) to[out=90,in=-90] ($(C)$);
		\draw[diagram,directed=0.9] ($(0,-1)+2*(C)$) to[out=90,in=-90] ($(1,0)+(C)$);
		\draw[diagram,looseness=2] ($(0,-1)+(T)+(C)$) to[out=-90,in=-90] ($(0,-1)+(N)+(C)$) -- ($(0,-1)+2*(C)$);
		\draw[diagram,looseness=2] (0,-1) node[below] {$X$} -- ($(0,-1)-(N)+(C)$) to[out=90,in=90] ($(0,-1)+(T)+(C)$);
		\draw[diagram,looseness=2] ($(1,-1)+(T)+(C)$) to[out=-90,in=-90] ($(1,-1)+(N)+(C)$) -- ($(1,-1)+2*(C)$);
		\draw[diagram,looseness=2] (1,-1) node[below] {$Y$} -- ($(1,-1)-(N)+(C)$) to[out=90,in=90] ($(1,-1)+(T)+(C)$);
	}
	&
	\tikzbo{
		\coordinate (T) at (0.2,0);
		\coordinate (N) at (0,0.1);
		\draw[diagram,directed,looseness=2] (0,1) -- (N) to[out=-90,in=-90] (T);
		\draw[diagram,directed,looseness=2] (T) to[out=90,in=90] ($(0,0)-(N)$) -- (0,-1) node[below] {$X^*$};
	}
	& =
	\tikzbo{
		\coordinate (T) at (0.2,0);
		\coordinate (N) at (0,0.1);
		\draw[diagram,directed,looseness=2] (T) to[out=-90,in=-90] (N) -- (0,0.5) arc (180:0:0.5cm) -- (1,-1) node[below] {$X^*$};
		\draw[diagram,directed,looseness=2] (-1,1) -- (-1,-0.5) arc (180:360:0.5cm) -- ($(0,0)-(N)$) to[out=90,in=90] (T);
	}
	= 
	\tikzbo{
		\coordinate (T) at (-0.2,0);
		\coordinate (N) at (0,0.1);
		\draw[diagram,directed,looseness=2] (T) to[out=90,in=90] ($(0,0)-(N)$) -- (0,-1) node[below] {$X^*$};
		\draw[diagram,directed,looseness=2] (0,1) -- (N) to[out=-90,in=-90] (T);
	}
\end{align}
The last equality introduced the graphical representation for $\theta_X^*$.
\begin{defn}
	\textbf{Ribbon fusion categories}
	are simply ribbon categories that are also fusion categories.
	They are also called ``premodular categories''.
\end{defn}
\begin{rema}
	Ribbon categories have a canonical pivotal structure that is spherical.
	The spherical condition is a consequence of \eqref{ribbon}.
	As a partial converse, the twist of a braided spherical category is ribbon structure if it is fusion.
	For more details see \cite[definition 2.29]{OnBraidedFusionCats} and the references therein.
\end{rema}

\subsubsection{Symmetric categories}
\begin{defn}
	A braided category is called \textbf{symmetric} iff $c_{X,Y} = c^{-1}_{Y,X}$.
	A symmetric category which is also fusion is called a \textbf{symmetric fusion category}.
\end{defn}
\begin{rema}
	As a consequence of \eqref{balanced}, a ribbon category is symmetric if the twist is trivial,
	although there exist symmetric ribbon categories with non-trivial twist.
\end{rema}
If the braiding is symmetric, over- and underbraiding are set equal in the diagrammatic calculus:
\begin{equation}
	c_{X,Y} =
	c_{Y,X}^{-1} =
	\tikzbo{
		\draw[thick,directed2] (1,-1) node[below] {$Y$} to[out=90,in=-90] (0,1);
		\draw[thick,directed2] (0,-1) node[below] {$X$} to[out=90,in=-90] (1,1);
	}
\end{equation}

\begin{thm}[After Deligne, \cite{Deligne:2002}]
	In a symmetric fusion category, dimensions of simple objects are integers.
	If the twist is trivial and all dimensions are positive,
	then there exists a (pivotal) fibre functor to vector spaces,
	and the symmetric fusion category is equivalent to the representations of the finite automorphism group of the fibre functor.
\end{thm}

\subsection{Diagrammatic calculus on spherical fusion categories}

\begin{defn}
	\label{fusion algebra}
	For a fusion category $\mathcal{C}$, let the \textbf{fusion algebra} $\Co\left[\mathcal{C}\right]$ be the complex algebra generated by its objects, modulo isomorphisms and the relations $X \oplus Y = X + Y$ and $X \otimes Y = XY$.	
\end{defn}
\begin{rema}
	If $\mathcal{C}$ is braided, $\Co\left[\mathcal{C}\right]$ is commutative.
\end{rema}

\begin{definition}
\label{def:colours}
By a handy generalisation of notation, closed loops involving only natural or extranatural transformations $\alpha$ can also be labelled with elements of the fusion algebra,
in this context called \textbf{colours},
instead of mere objects.
The evaluation of a diagram with a linear combination of objects is defined as the sum of the evaluations of the diagrams with the individual objects:
\tikzset{
	pics/trace/.style 2 args={code={
		\node[draw] (B) {#1};
		\draw[thick,directed] (B.north) arc [start angle=0,end angle=180,radius=0.5cm] node[near start,above] {#2} -- +($(B.south)-(B.north)$) arc [start angle=180,end angle=360,radius=0.5cm];
	}}
}
\begin{align}
	X &\coloneqq \sum_i \lambda_i X_i\\
	\tikzb{\pic {trace={$\alpha$}{$X$}}}
	&\coloneqq
	\sum_i \lambda_i \;
	\tikzb{\pic {trace={$\alpha_{X_i}$}{$X_i$}}}
	\label{eq:colours}
\end{align}
\end{definition}
Since braiding and twist are natural transformations, colours can be used in the diagrammatic calculus.

\begin{defn}
	The \textbf{Kirby colour} $\Omega_{\mathcal{C}} $ of a spherical fusion category $\mathcal{C}$ is defined as the sum over the simple objects in $\Lambda_\mathcal{C}$ weighted by their dimensions:
	\begin{align}
		\Omega_{\mathcal{C}} \coloneqq \sum_{X \in \Lambda_{\mathcal{C}}} \qdim{X} X
	\end{align}
	Its dimension $\qdim{\Omega_{\mathcal{C}}} = \sum_{X \in \Lambda_{\mathcal{C}}} \qdim{X}^2$ is known as the \textbf{global dimension} of the category.
	It is always positive, since the field $\Co$ has characteristic zero \cite{OnFusionCategories}.
\end{defn}

The following two lemmas are well-known, e.g., in \cite[Section 2]{CraneYetterKauffman:1997177}.
\begin{lem}[Schur's lemma]
	Any endomorphism $f\colon X \to X$ of a simple object with non-zero dimension satisfies:
	\begin{equation}
		f = 1_X \cdot \frac{\tr(f)}{\qdim{X}}
	\end{equation}
\end{lem}

	\begin{proof}
		Since $X$ is simple, $\mathcal{C}(X, X) \cong \Co$, so every endomorphism is a multiple of the identity.
		Taking the trace on both sides of the equation $f = \lambda 1_X$ yields the result.
	\end{proof}

\begin{lem}[Insertion lemma]
	\label{insertion lemma}
	For any object $X$ in a spherical fusion category,
	its identity morphism can be decomposed into a weighted sum of identities of simple objects $Z$:
	\newcommand{\idx}{\tikzb{\draw[diagram,directed] (0,0) -- node[right] (B) {$X$} (0,3);}}
	\newcommand{\inserted}[3]{
		\node[draw] (i) at (0,0.8) {#2};
		\node[draw] (i*) at (0,2.2) {#3};
		\draw[thick,directed] (0,0) -- node[right] {$X$} (i) -- node[right] (B) {$#1$} (i*) -- node[right] {$X$} (0,3);
	}
	\begin{align}
		\idx & = \sum_{Z \in \Lambda_\mathcal{C}} \sum_{\substack{\iota_{Z,i} \in \mathcal{C}(X,Z) \\ \braket{\iota^i_Z}{\iota_{Z,j}} = \delta_{i,j}}} \qdim Z \tikzb{\inserted{Z}{$\iota_{Z,i}$}{$\iota^i_Z$}} \label{insertion}
	\end{align}
	The $\iota_{Z,i}$ form a basis of $\mathcal{C}(X, Z)$ to which the $\iota^j_Z \in \mathcal{C}(Z, X)$ are the dual basis with respect to the spherical pairing $\braket{-}{-}$ defined in \eqref{spherical pairing}.
\end{lem}

	\begin{proof}
		The definition of semisimplicity \ref{def:semisimple} implies that for some $\beta_Z^i\in\mathcal C(Z,X)$,
		one can decompose:
		\begin{equation}
			1_X = \sum_{Z,i} \beta_Z^i \circ \iota_{Z,i}
		\end{equation}
		Inserting this equality into $\braket{\iota_{Z,j}}{1_X \circ \iota^k_Z}$ shows that
		$\beta_Z^i = \qdim{Z} \iota^i_Z$.
	\end{proof}
\begin{remark}
	The insertion lemma is a generalisation of the fact from linear algebra that any vector can be decomposed uniquely into a linear combination of basis vectors.

	Due to its similarity to \eqref{eq:colours},
	it is common to say that (the identity of)
	the Kirby colour $\Omega_\mathcal{C} = \sum_{Z \in \Lambda_\mathcal{C}} \qdim{Z} Z$ can always be inserted in $X$'s identity.
	This explains the particular name of the lemma.
	
\end{remark}

\subsubsection{Ribbon fusion categories}

This subsection introduces some notation and known lemmas in ribbon fusion categories.
These are also known as premodular categories.

\begin{defn}[Graphical calculus for links]
	\label{def:graphical calculus}
	Let $L$ be an oriented framed link with a partition of its components into $N$ sets.
	Choose a regular diagram of the link in the plane such that the blackboard framing from the diagram matches the original framing of the link.
	Given a labelling $(X_1, X_2,\dots{}, X_N)$ of the sets with colours from a ribbon fusion category $\mathcal{C}$,
	label the link components in each set with the colour of the set and interpret the diagram as a morphism in $\mathcal{C}$,
	in the following way:
	Insert identity morphisms for vertical lines,
	braidings for crossings, evaluations for maxima of lines and coevaluations for minima.
	They are composed and tensored according to the vertical and horizontal structure of the diagram.
	The whole procedure is explained rigorously in \cite{Shum:1994TortileTensorCats} and \cite{Selinger:graphical}.
	
	Since a link has no open ends, the resulting morphism will be an endomorphism of $I$,
	which is essentially a complex number.
	This number is denoted as $\left<L(X_1, X_2, \dots{}, X_N)\right>$
	and called the \textbf{evaluation} of the labelled link diagram
	(not to be confused with the evaluation morphisms $\ev_X$).
	A labelled link diagram will sometimes be used interchangeably with its evaluation.
\end{defn}

\begin{remarks}
	Note that the choice of diagram for the framed link doesn't matter as two diagrams only differ by isotopies and (second and third) Reidemeister moves,
	which amount to identities (e.g. naturality squares or axioms like the snake identity) in the category.
	
	It is necessary that $\mathcal{C}$ is ribbon since this ensures that the framing coefficients of the link are translated into twists of $\mathcal{C}$.
\end{remarks}

\begin{defn}
	\label{def:symmetric centre}
	An object $X$ is called \textbf{transparent} (or ``central'') in $\mathcal{C}$ if it braids trivially with any object $Y$ in $\mathcal{C}$,
	that is, $c_{Y,X} \circ c_{X,Y} = 1_{X \otimes Y}$.
	The graphical representation of this condition is found in \eqref{eq:transparent}.
	
	The full symmetric monoidal subcategory $\mathcal{C}' \subset \mathcal{C}$ with all transparent objects of $\mathcal{C}$ is called the \textbf{symmetric centre} (or ``centraliser'') of $\mathcal{C}$,
	as for example in \cite{Mueger:2003ModularCats} or \cite{OnBraidedFusionCats}.
	The set of equivalence classes of simple transparent objects in $\mathcal{C}$ is then $\Lambda_{\mathcal{C}'}$.
	Dotted lines represent transparent objects.
	\begin{equation}
		X \in \ob \mathcal{C}' \iff
		\tikzbo{
			\draw[diagram,directed=0.999] (1,-1) node[below] {$Y$} to[out=90,in=-90] (0,0);
			\draw[diagram,dotted] (1,0) to[out=90,in=-90] (0,1);
			\draw[diagram,dotted,directed=0.999] (0,-1) node[below] {$X$} to[out=90,in=-90] (1,0);
			\draw[diagram] (0,0) to[out=90,in=-90] (1,1);
		}
		=
		\tikzbo{
			\draw[diagram,directed] (1,-1) node[below] {$Y$} -- (1,1);
			\draw[diagram,dotted,directed] (0,-1) node[below] {$X$} -- (0,1);
		}
		\forall Y \in \ob \mathcal{C}
		\label{eq:transparent}
	\end{equation}
\end{defn}
\begin{defn}
	\label{transparent colour}
	Assume $X$ is a colour $X = \lambda_1 X_1 + \lambda_2 X_2 + \dots + \lambda_N X_N$ with all $X_i$ simple, further assume that $X_1$ to $X_k$ are transparent and $X_{k+1}$ to $X_N$ are not.
	Then define the \textbf{transparent colour} $X' \coloneqq \lambda_1 X_1 + \lambda_2 X_2 + \dots + \lambda_k X_k + 0 X_{k+1} + \dots + 0 X_N$.
\end{defn}
\begin{defn}
	The \textbf{transparent Kirby colour} is defined as follows.
	\begin{align}
		\Omega_{\mathcal{C}'} = \Omega_{C}' &= \sum_{X \in \Lambda_{\mathcal{C}'}} \qdim{X} X
		\intertext{In the same manner, the \textbf{transparent dimension} is defined:}
		\qdim{\Omega_{\mathcal{C}'}} =
		\tikzb{\dimcircle[dotted]{$\Omega_{\mathcal{C}'}$}}
		&= \sum_{X \in \Lambda_{\mathcal{C}'}} \qdim{X}^2
	\end{align}
\end{defn}
\begin{defn}
	A category is called \textbf{modular} if it has $\Lambda_{C'} = \{I\}$,
	i.e. the monoidal identity $I$ is the only transparent object.
\end{defn}
The transparent dimension of a modular category is therefore 1.
Note that the multifusion case, where $I$ is not a simple object, is excluded here.

\begin{rema}
	An object that is not transparent in $\mathcal{C}$ can still be transparent in a subcategory $\mathcal{B} \subset \mathcal{C}$.
\end{rema}

\subsubsection{Encirclement}

The technique of encirclement allows for many elegant and powerful calculations.
It is indispensable when defining invariants derived from ribbon fusion categories and Kirby diagrams.
Its power comes from the so-called \emph{killing property}.
This is also known as the Lickorish encircling lemma \cite{Lickorish:1993Skein}, see also \cite{Roberts:1995SkeinTheoryTV}.
It can be generalised from modular to ribbon fusion categories \cite[Lemma 1.4.2, in different notation]{Bruguieres:Modularisations}.
\begin{lem}[Killing property]
	\label{Killing property}
	In a ribbon fusion category, the following holds for any object:
	\tikzset{
		pics/killing/.style n args={3}{code={
			\coordinate (a) at (0,0.7);
			\draw[diagram,directed,#1] ($(0,0)-(a)$) node[right] {#3} -- (a);
			\draw[#2, thick] ($(0,0)-(a)$) arc[start angle=0, end angle=-70,radius=0.8cm];
			\draw[#2, thick] (a)           arc[start angle=0, end angle= 70,radius=0.8cm];
		}}
	}
	\begin{align}
		\tikzbo{
			\encircle{\pic {killing={{}}{{}}{$X$}};}
		}
		& =
		\tikzbo{\pic {killing={dotted}{dotted}{$X'$}};}
		\tikzbo{
			\dimcircle{$\Omega_\mathcal{C}$};
		}
	\end{align}
	Let in particular $X$ be simple.
	Then $X' = X$ if it is transparent, and 0 otherwise.
	In the latter case one says that $X$ is ``killed off''.
	
	Note that the orientation for the circle containing $\Omega_\mathcal{C}$ does not need to be specified since the colour is self-dual.
\end{lem}

\newcommand{\encircledstrands}[1]{
	\encircle{
		\foreach \x in {-0.6,#1,...,0.61}
		{
			\draw[diagram] (\x,-1.5) -- (\x,1.5);
		}
	}
	\coordinate (B) at (0,0);
}

\newcommand{\onlycutstrands}[1]{
	\coordinate (B) at (0,0);
	\foreach \m in {-1,1}
	{
		\foreach \x in {-0.6,#1,...,0.61}
		{
			\draw[thick] (\x,1.5*\m) -- (0,0.8*\m);
		}
	}
}
\newcommand{\cutstrands}[2]{
	\encircle{
		\onlycutstrands{#1}
		\draw[ultra thick, white] (0,-0.8) -- (0, 0.8);
		\draw[dotted, thick] (0,-0.8) node[right] {$\iota^i$} -- node[right] {#2} (0, 0.8) node[right] {$\iota_i$};
	}
}

The combination of the killing property \ref{Killing property} and the insertion lemma \ref{insertion lemma} gives the explicit morphism of an arbitrary object encircled with the Kirby colour.
\begin{lem}[Cutting strands] Let $X$ be an arbitrary object of a modular category $\mathcal{C}$.
	\label{cutting strands}
	Then:
	\newcommand{\diagheight}{1.5}
	\begin{align}
		\tikzbo{
			\encircle{
				\draw[diagram,directed] (0,-\diagheight) node[below] {$X$} -- (0,\diagheight);
			}
		}
		&=
		\sum_{Z \in \Lambda_\mathcal{C}}
		\sum_{\substack{\iota_i \in \mathcal{C}(X, Z) \\ \braket{\iota_i}{\iota^j} = \delta_{i,j}}} \qdim{Z}
		\tikzbo{
			\node[draw] (iota) at (0,0.8) {$\iota^i$};
			\node[draw] (iota*) at ($(0,0)-(iota)$) {$\iota_i$};
			\draw[thick,directed] (0,-\diagheight) node[below] {$X$} -- (iota*);
			\draw[thick,opdirected] (0,\diagheight)
				-- (iota);
			\encircle{
				\draw[diagram,directed] (iota*) -- node[right] {$Z$} (iota);
			}
		}
		&=
		\sum_{\mathclap{\substack{\iota_i \in \mathcal{C}(X, I) \\ \braket{\iota_i}{\iota^j} = \delta_{i,j}}}} \qdim{\Omega_\mathcal{C}}
		\tikzbo{
			\node[draw] (iota) at (0,0.8) {$\iota^i$};
			\node[draw] (iota*) at ($(0,0)-(iota)$) {$\iota_i$};
			\draw[thick,directed] (0,-\diagheight) node[below] {$X$} -- (iota*);
			\draw[thick,opdirected] (0,\diagheight)
				-- (iota);
		}
	\end{align}
	The last step uses the fact that in a modular category, $I$ is the only transparent object.
\end{lem}
\begin{lem}[Cutting two strands] Let $Z_1, Z_2$ be simple objects of a modular category $\mathcal{C}$.
	Then as a special case of the previous lemma:
	\label{cutting two strands}
	
	\newcommand{\diagheight}{1}
	\begin{align}
		\tikzbo{
			\encircle{
				\foreach \x/\i in {-0.5/1,0.5/2}
				{
					\draw[diagram,directed] (\x,-\diagheight) node[below] {$Z_\i$} -- (\x,\diagheight);
				}
			}
		}
		= \delta_{Z_1^*,Z_2}\qdim{Z_1}^{-1} \qdim{\Omega_\mathcal{C}}
		\tikzbo{
			\coordinate (tl) at (-0.5,\diagheight);
			\draw[thick,opdirected,looseness=2] (tl) to[out=270,in=270] ($2*(tl -| 0,0)-(tl)$);
			\draw[thick,directed,looseness=2] ($2*(tl |- 0,0)-(tl)$) node[below] {$Z_1$} to[out=90,in=90] ($(0,0)-(tl)$) node[below] {$Z_2$};
		}
	\end{align}
	To see the prefactors, observe that $\mathcal{C}(Z_1 \otimes Z_2, I)\cong \mathcal{C}(Z_2, Z_1^*)$.
	This is isomorphic to $\Co$ if ${Z_1^* \cong Z_2}$, and 0 otherwise.
	If ${Z_1^* \cong Z_2}$, then $\mathcal{C}(Z_1 \otimes Z_2, I)$ is spanned by $\ev_{Z_1}$.
	Since $\widetilde\coev_{Z_1} \circ \ev_{Z_1} = \qdim{Z_1}$, the dual basis element must be $\widetilde\coev_{Z_1} \cdot \qdim{Z_1}^{-1}$.
\end{lem}

\subsection{4-Manifolds and Kirby calculus}
An extensive treatment of these topics is found in \cite{GompfStipsicz},
\cite{Akbulut:4manifolds} and \cite{Kirby:Topology4Manifolds}.
The essential definitions and facts are highlighted here.

\subsubsection{Handle decompositions}
Let $D^k$ denote the closed $k$-disk,
or $k$-ball.
The space $D^k \times D^{4-k}$, $k \in \{0,1,2,3,4\}$, is called a 4-dimensional $k$-handle.
All handles have the same underlying topological space, but they differ in the way they are attached to each other.
The boundary of a $k$-handle is $\partial \left(D^k \times D^{4-k}\right) = S^{k-1} \times D^{4-k} \cup D^k \times S^{3-k}$, where $S^{-1} = \emptyset$.
The first component of the boundary is called the \textbf{attaching boundary} or ``attaching region'',
the second component the \textbf{remaining boundary} or ``remaining region''.
Some examples are shown in Table \ref{table:handle boundaries}.
\begin{table}
	\begin{center}
		\begin{tabular}{llll}
			\toprule
			$k$
				& Space
					& Attaching boundary
						& Remaining boundary
			\\\midrule
			0
				& $D^0 \times D^4$
					& $\emptyset$
						& $S^3 \cong \R^3 \cup \{\infty\}$
			\\
			1
				& $D^1 \times D^3$
					& $S^0 \times D^3 \cong \{-1, 1\} \times D^3$
						& $D^1 \times S^2 \cong [-1, 1] \times S^2$
			\\
			2
				& $D^2 \times D^2$
					& $S^1 \times D^2$
						& $D^2 \times S^1$
			\\\bottomrule
		\end{tabular}
	\end{center}
	\caption{Some relevant special cases of 4-dimensional $k$-handles and their boundaries.}
	\label{table:handle boundaries}
\end{table}

Smooth manifolds admit handle decompositions.
A $k$-handle can be attached to a manifold with boundary by embedding its attaching region into the boundary of the manifold.
A $k$-handlebody is obtained by attaching a disjoint union of $k$-handles to a $k-1$ handlebody, and is thus a union of 0-, 1-, \dots{} and $k$-handles.
Note that 0-handles have no attaching region, and a 0-handlebody is just a disjoint union of 0-handles, which are $D^4$s.
Every $n$-manifold can be decomposed into handles, that is, it is diffeomorphic to an $n$-handlebody.

The handle decomposition is by no means unique.
Two handle decompositions of diffeomorphic manifolds are always related by ``handle moves'', which are either cancellations of a $k$- and a $(k+1)$-handle, or a slide of a $(k+m)$- over a $k$-handle.

For a connected manifold it is always possible to arrive at a handle decomposition with exactly one 0-handle by cancelling 0-1-handle pairs.
Similarly, for a closed connected $n$-manifold it is always possible to have exactly one $n$-handle by cancelling $(n-1)$-$n$-handle pairs.

\subsubsection{Kirby diagrams and dotted circle notation}
\label{Kirby diagrams}

\tikzmath{
	\lowknot = 0.3;
	\dknot   = 0.3;
	\width   = 0.4;
	\hddist  = 1.3;
}
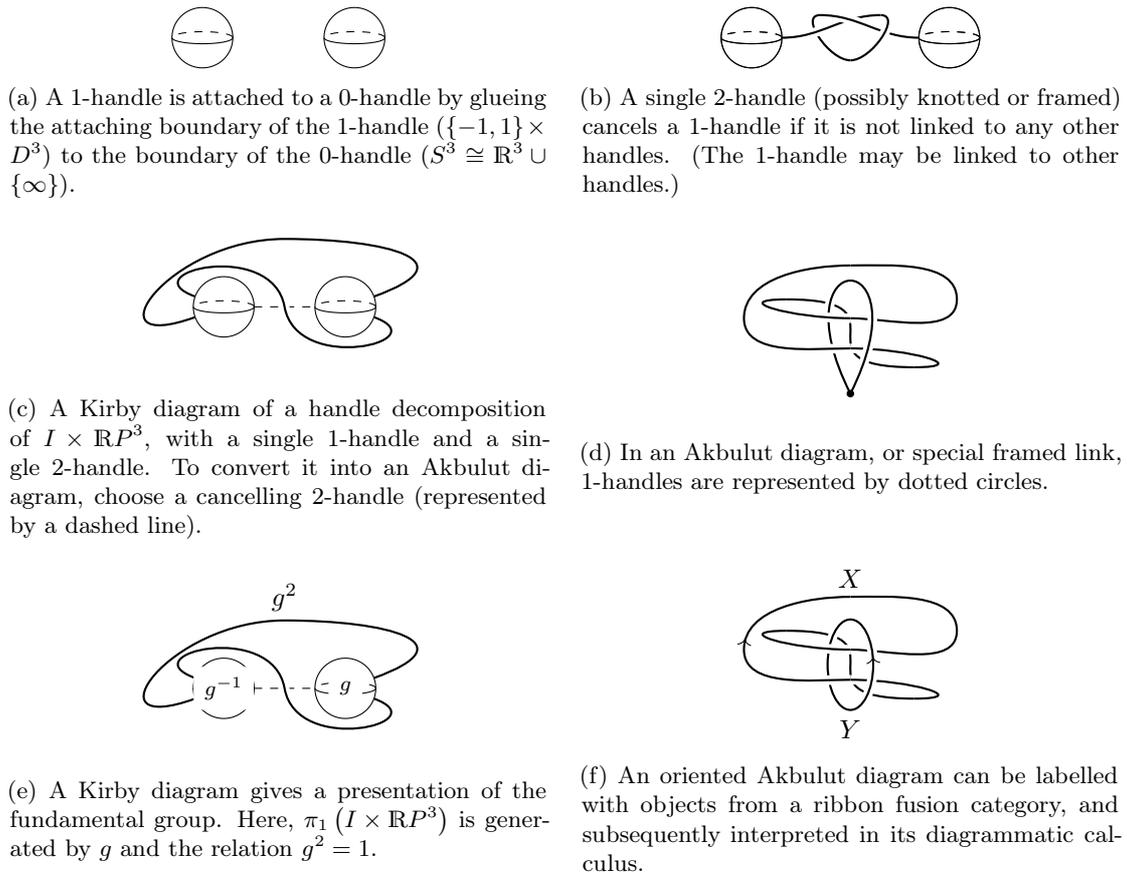
\begin{figure}
	\centering
	\begin{subfigure}{0.48\textwidth}
		\centering
		\begin{tikzpicture}
			\pic at (-1,0) {sphere};
			\pic at ( 1,0) {sphere};
		\end{tikzpicture}
		\caption{A 1-handle is attached to a 0-handle by glueing the attaching boundary of the 1-handle ($\{-1, 1\} \times D^3$) to the boundary of the 0-handle ($S^3 \cong \R^3 \cup \{\infty\}$).}
	\end{subfigure}
	\quad
	\begin{subfigure}{0.48\textwidth}
		\centering
		\begin{tikzpicture}
			\pic (hd1) at (-\hddist,0) {sphere};
			\pic (hd2) at ( \hddist,0) {sphere};
			\draw[diagram]
				(      0, -\lowknot) to[out=180, in=180]
				(-\width,    \dknot);
			\draw[diagram]
				(hd1-r)              to[out=  0, in=180]
				( \width,    \dknot);
			\draw[diagram]
				(-\width,    \dknot) to[out=  0, in=180]
				(hd2-l);
			\draw[diagram]
				( \width,    \dknot) to[out=  0, in=  0]
				(      0, -\lowknot);
			\pic at (-\hddist,0) {sphere};
			\pic at ( \hddist,0) {sphere};
		\end{tikzpicture}
		\caption{A single 2-handle (possibly knotted or framed) cancels a 1-handle if it is not linked to any other handles.
		(The 1-handle may be linked to other handles.)}
		\label{subfig:cancelling 1-handle}
	\end{subfigure}
	\medskip
	\begin{subfigure}{0.48\textwidth}
		\centering
		\begin{tikzpicture}
			\pic[RP2/nolabels] {RP2};
		\end{tikzpicture}
		\caption{A Kirby diagram of a handle decomposition of $I \times \R P^3$,
		with a single 1-handle and a single 2-handle.
		To convert it into an Akbulut diagram,
		choose a cancelling 2-handle (represented by a dashed line).}
		\label{subfig: from Kirby to Akbulut}
	\end{subfigure}
	\quad
	\begin{subfigure}{0.48\textwidth}
		\centering
		\begin{tikzpicture}
			\pic {RP2Akbulut={}{}};
		\end{tikzpicture}
		\caption{In an Akbulut diagram, or special framed link,
		1-handles are represented by dotted circles.}
		\label{subfig:RP2Akbulut}
	\end{subfigure}
	\medskip
	\begin{subfigure}{0.48\textwidth}
		\centering
		\begin{tikzpicture}
			\pic {RP2};
		\end{tikzpicture}
		\caption{A Kirby diagram gives a presentation of the fundamental group.
		Here, $\pi_1\left(I \times \R P^3\right)$ is generated by $g$ and the relation $g^2 = 1$.}
		\label{subfig:Kirby example fundamental group}
	\end{subfigure}
	\quad
	\begin{subfigure}{0.48\textwidth}
		\centering
		\begin{tikzpicture}
			\pic[RP2Akbulut/yeslabels] {RP2Akbulut={$X$}{$Y$}};
		\end{tikzpicture}
		\caption{An oriented Akbulut diagram can be labelled with objects from a ribbon fusion category,
		and subsequently interpreted in its diagrammatic calculus.}
	\end{subfigure}
	\caption{Kirby diagrams and Akbulut diagrams}
	\label{fig:Kirby examples}
\end{figure}

For the 2-handlebody of a four-manifold, one can specify the handles and their attaching maps by identifying the boundary of the single 0-handle with $\R^3 \cup \infty$ and drawing pictures of the attaching regions of the 1- and 2-handles. This is explained in \cite[Section 5.1]{GompfStipsicz}.
An attachment of a 1-handle amounts to choosing two 3-balls $D^3 \times \{-1, 1\} \cong D^3 \sqcup D^3 \subset \R^3$,
which are identified by an orientation-reversing map.
A 2-handle attachment is an embedding of $D^2 \times S^1$,
which is, up to isotopy, a framed embedding of $S^1$, i.e. a framed knot.
When a part of the 2-handle is attached to a 1-handle,
the $S^1$ of the 2-handle will enter one of the 3-balls of the 1-handle and leave the other 3-ball with which the former has been identified.
The diagram of the attaching regions in $\R^3$ is called a \textbf{Kirby diagram}.
Some examples can be found in Section \ref{non-simply-connected}.

A theorem ensures that for a closed four-manifold $M$, specifying the 2-handle\-body of a handle decomposition determines $M$ up to diffeomorphism,
i.e. any way of adding the 3- and 4-handles will yield the same manifold.
Thus a closed manifold is specified uniquely (up to diffeomorphism) by its Kirby diagram.

The \textbf{dotted circle notation} for 1-handles developed by Akbulut is sometimes more convenient.
Instead of adding a 1-handle, one can add a cancelling 1-2-handle pair (as shown in Figures \ref{subfig:cancelling 1-handle} and \ref{subfig: from Kirby to Akbulut}) and,
after adding all further 2-handles, remove the cancelling 2-handle.
In the diagram, the step of adding the cancelling pair does not require any notation because it does not change the topology.
However one needs a notation to indicate how the cancelling 2-handle is removed \cite[Section 1.2]{Kirby:Topology4Manifolds}.
Recall that a 2-handle is attached by $D^2\times S^1\subset D^2\times D^2$ and so the remaining part of the boundary is  $S^1\times D^2$.
This thickened (0-framed) circle is sufficient to indicate the 2-handle and is included in the diagram to represent the 1-handle it cancels.
To distinguish the 1- and 2-handles,
dots are drawn on those circles representing 1-handles,
as in Figure \ref{subfig:RP2Akbulut}.

In Section \ref{proof of invariance},
it will be detailed which moves one can perform on handle decompositions without changing the diffeomorphism class of the manifold.
Further examples can be found in Section \ref{non-simply-connected}.

Note that the sublink consisting of only dotted circles is an unlinked union of 0-framed unknots,
but 2-handle circles can be linked with each other.
The 2-handle circles can also be linked with the dotted circles;
this happens whenever a 2-handle runs over a 1-handle.
Links of this type are called \textbf{special framed links}.

To produce an Akbulut picture from a Kirby picture \cite[Section 5.4]{GompfStipsicz},
take the two 3-balls of a 1-handle.
The cancelling 2-handle connects them with a framed interval,
or an embedding of $D^2 \times [-1,1]$, with the ends on the 3-balls.
Now instead of drawing the balls,
draw the dotted circle $S^1 \times \{0\} \subset D^2 \times [-1,1]$.
A 2-handle running over this 1-handle is then drawn as a continuous line going through the dotted circle.

\begin{definition}[Evaluation of Akbulut pictures]
	\label{def:eval of links}
	The dotted circle notation of a handle decomposition of a closed, oriented four-manifold will be important in the definition of the invariant.
	The dots specify a partition of the special framed link diagram $L$ in two sublinks,
	corresponding to the 1-handles and the 2-handles, respectively.
	After arbitrarily chosing orientations on each $S^1$,
	the two sublinks can be labelled with two colours $X$ and $Y$ of a ribbon fusion category.
	(The colours then need to be self-dual such that the chosen orientations don't matter.)
	Each dotted link component (1-handle) is labelled with the colour $X$ and each of the remaining components (2-handles) is labelled with $Y$,
	and the dots can then be removed.
	The labelled link is denoted $L(X,Y)$.
	As in Definition \ref{def:graphical calculus},
	the evaluation is then $\left<L(X,Y)\right>$.
\end{definition}

Note that the relation between the two graphical notations for 1-handles mimicks the diagrammatic representation of Lemma \ref{cutting strands}.
This will be exploited in Section \ref{sec:cutting strands},
where a definition of the invariant in terms of the Kirby diagram is given.

\subsubsection{The fundamental group}
\label{Kirby-calculus-fundamental-group}
A Kirby diagram for a manifold $M$ gives rise to a presentation of its fundamental group $\pi_1(M)$.
Each 1-handle is a generator, while the 2-handles are the relations.

More specifically, choosing a basepoint in the 0-handle and an arbitrary direction on each 1-handle,
there is a homotopy class of noncontractible curves going through a 1-handle once.
A 2-handle 
gives a way of contracting the $S^1$ on its own attaching region,
which is drawn in the Kirby diagram.
Thus the composition of the curves going through the 1-handles along which the 2-handle is attached can be equated with the contractible curve.

This can be visualised as follows.
Each 1-handle is associated to a generator.
One of its corresponding 3-balls is labelled with the generator and the other with its inverse,
thus fixing a direction on the 1-handle.
For every circle coming from a 2-handle,
choose an orientation and construct a word of generators by going once along the circle,
writing down the generator (or its inverse) when entering a 1-handle through a 3-ball.
(No action needs to be taken when leaving a ball.)
The resulting word is then a relation in the presentation of the fundamental group.
An example is given in Figure \ref{subfig:Kirby example fundamental group}.

\section{The generalised dichromatic invariant}
\label{section invariant}
\subsection{The generalised sliding property}

\begin{lem}
	\label{sliding property}
	In a ribbon fusion category $\mathcal{C}$,
	the \emph{sliding property} (in its original form due to Lickorish \cite{Lickorish:1993Skein}) holds:
	\begin{equation}
		\tikzbo{\pic {shortid=$X$};}
		\tikzbo{\pic {unthrough={$A$}{$\Omega_{\mathcal{C}}$}};} =
		\tikzbo{\pic {slid={$\Omega_{\mathcal{C}}$}{$X$}};}
	\end{equation}
\end{lem}
\begin{proof}
	\begin{equation*}
		\tikzbo{\pic {shortid=$X$};}
		\tikzbo{\pic {unthrough={$A$}{$\Omega_{\mathcal{C}}$}};} =
		\tikzbo{\pic {shortid=$X$};}
		\tikzbo{\pic[/slid/throughstyle/.append style={dotted}] {unthrough={$A$}{$\Omega_{\mathcal{C}}$}};} =
		\tikzbo{\pic[/slid/throughstyle/.append style={dotted}] {slid={$\Omega_{\mathcal{C}}$}{$X$}};} =
		\tikzbo{\pic {slid={$\Omega_{\mathcal{C}}$}{$X$}};}
	\end{equation*}
	The killing property \ref{Killing property} has been used twice.
\end{proof}

As the diagrams suggest, the sliding property will later ensure that the invariant doesn't change under handle slides.
To label 2-handles differently from 1-handles, it is necessary to generalise the sliding property of Lemma \ref{sliding property} to ensure invariance under the 2-2-handle slide.
The idea will be to label the 2-handles with $F\Omega_\mathcal{C}$, where $F$ is a suitable functor.
Then encirclements with $F\Omega_\mathcal{C}$ must also satisfy a sliding property.

Lemma \ref{insertion lemma}, which states that the Kirby colour can be inserted into the identity of any object, can be generalised.
\begin{lem}[Generalised insertion lemma]
	\label{generalised insertion lemma}
	Let $F\colon \mathcal{C} \to \mathcal{D}$ be a pivotal functor,
	and $X$ an object in $\mathcal{C}$.
	Then the identity of $FX$ decomposes over $F\Omega_{\mathcal{C}} = \bigoplus_X \qdim{X} FX$.
	In this situation, we say that $F\Omega_{\mathcal{C}}$ can be ``inserted'' into the identity of $FX$.
\end{lem}
\begin{proof}
	Apply $F$ to both sides of \eqref{insertion} in the insertion lemma.
	Since pivotal functors preserve traces, they also preserve (categorical) dimensions and dual bases.
\end{proof}

The sliding property can also be generalised in a similar way.
\begin{lem}[Generalised sliding property]
	\label{generalised sliding property}
	Let $F\colon \mathcal{C} \to \mathcal{D}$ be a pivotal functor from a spherical fusion category to a ribbon fusion category.
	Then the following generalisation of the sliding property holds for all objects $X \in \ob \mathcal{C}, A \in \ob \mathcal{D}$:
	\begin{align}
		\tikzbo{\draw[diagram,directed] (0,-1) node[below] {$FX$} -- (0,1);}
		\tikzbo{\encircle[angle=30,label={$F\Omega_{\mathcal{C}}$}]{\draw[diagram,directed] (0,-1) node[below] {$A$} -- (0,1);}} = 
		\tikzbo{
			\pic {slid={$F\Omega_{\mathcal{C}}$}{$FX$}};
		}
	\end{align}
\end{lem}
	\newcommand{\tikzoffset}{\node at (1.8,0) {};}
	\newcommand{\Aid}{\draw[diagram,directed] (0,-1.5) node[below] {$A$} -- (0,1.5); \tikzoffset}
	\newcommand{\minipagewidth}{6.5cm}
	\begin{proof}
		The proof proceeds diagrammatically.
		\renewcommand{\minipagewidth}{5.6cm}
		\begin{align}
			& & \tikzbo{
				\draw[diagram,directed] (-\x,-1.5) node[below] {$FX$} -- (-\x,1.5);
				\encircle[angle=30,label={$F\Omega_{\mathcal{C}}$}]{\Aid}
			}\nonumber\\
			& = \sum_{\mathclap{Y \in \Lambda_\mathcal{C}}} \qdim{Y}
			& \tikzbo{
				\draw[diagram,directed] (-\x,-1.5) node[below] {$FX$} -- (-\x,1.5);
				\encircle[angle=35,label=FY]{\Aid}
			}
			& \; \begin{minipage}{\minipagewidth}
				Using the definition of $\Omega_{\mathcal{C}}$.
			\end{minipage} \\
			& = \sum_{\mathclap{\substack{Y, Z \in \Lambda_\mathcal{C} \\ \iota_i \in \mathcal{C}(Z, X \otimes Y^*) \\ \braket{\iota^i}{\iota_j} = \delta_{i,j}}}} \qdim{Y}\qdim{Z}
			& \tikzb{
				\node[morbox] (i) at (-0.5,-0.45) {$F\iota_i$};
				\node[morbox] (i*) at (-0.5,0.45) {$F\iota^i$};
				\draw[diagram,directed] (-\x,-1.5) node[below] {$FX$} to[out=90, in=270] (i.245);
				\draw[diagram,directed] (i) -- node[left] (B) {$FZ$} (i*);
				\draw[diagram,directed] (i*.115) to[out=90,in=270] (-\x,1.5);
				\coordinate (Y) at (0.5,0);
				\draw[diagram,directed=0.1] (Y) to[out=90,in=90,looseness=2] (i*.65);
				\Aid
				\draw[diagram] (i.295) to[out=270,in=270,looseness=2] (Y) node[right] {$FY$};
			}
			& \; \begin{minipage}{\minipagewidth}Insertion of $F\Omega_{\mathcal{C}}$, according to Lemma \ref{generalised insertion lemma}.\end{minipage} \\
			& = \sum_{\mathclap{\substack{Y, Z \in \Lambda_\mathcal{C} \\ \iota_i \in \mathcal{C}(Z, X \otimes Y^*) \\ \braket{\iota^i}{\iota_j} = \delta_{i,j}}}} \qdim{Y}\qdim{Z}
			& \tikzb{
				\node[morbox] (i) at (0.5,-0.6) {$F\iota_i$};
				\node[morbox] (i*) at (0.5,0.6) {$F\iota^i$};
				\draw[diagram,directed] (i*.115) to[out=90,in=270] (-\x,1.5);
				\draw[diagram] (-0.3,0) coordinate (M) node[left] (B) {$FZ$} to[out=90,in=270] (i*);
				\Aid
				\draw[diagram,directed] (-\x,-1.5) node[below] {$FX$} to[out=90, in=270] (i.245);
				\draw[diagram,opdirected=0.01] (M) to[out=270,in=90] (i);
				\draw[diagram,directed,looseness=2.5] (i.295) to[out=270,in=270] (1.35,0) node[right] {$FY$} to[out=90,in=90] (i*.65);
			}
			& \;
				\begin{minipage}{\minipagewidth}
					Naturality of the braiding as isotopy.
				\end{minipage}
			\\
			& = \sum_{\mathclap{\substack{Y, Z \in \Lambda_\mathcal{C} \\ \tilde \iota_i \in \mathcal{C}(Y, Z^* \otimes X) \\ \braket{\tilde \iota^i}{\tilde \iota_j} = \delta_{i,j}}}} \qdim{Y} \qdim{Z}
			& \tikzb{
				\node[morbox] (i) at (0.5,-0.4) {$F \tilde \iota_i$};
				\node[morbox] (i*) at (0.5,0.4) {$F \tilde \iota^i$};
				\draw[diagram,directed] (i*.65) to[out=90,in=270] (-\x,1.5);
				\draw[diagram,directed=0.1,looseness=1.4] (-0.5,0) coordinate (M) node[left] (B) {$FZ$} to[out=90,in=90] (i*.115);
				\Aid
				\draw[diagram,directed] (-\x,-1.5) node[below] {$FX$} to[out=90, in=270] (i.295);
				\draw[diagram,looseness=1.4] (M) to[out=270,in=270] (i.245);
				\draw[diagram,directed] (i) to[out=90,in=270] node[right] {$FY$} (i*);
			}
			& \;
				\begin{minipage}{\minipagewidth}
					Isomorphism $\iota_i \mapsto \tilde\iota_i$ between $\mathcal{C}(Z, X \otimes Y^*)$ and $\mathcal{C}(Y, Z^* \otimes X)$ given by composition with evaluation and coevaluation.
					The $\tilde \iota$ form dual bases because of pivotality of $F$ and sphericality of $\mathcal{D}$.
					(Explained below.)
				\end{minipage}
				\label{TwistingFIntertwinersAround}
			\\
			& =  \sum_{\mathclap{Z \in \Lambda_\mathcal{C}}} \qdim{Z}
			& \tikzbo{
				\draw[diagram] (-\x,1.5) to[out=-90,in=90] (\x,0);
				\encircle[label=FZ,angle=150]{\Aid}
				\draw[diagram,directed=0.95] (-\x,-1.5) node[below] {$FX$} to[out=90,in=-90] (\x,0);
			}
			& \;
				\begin{minipage}{\minipagewidth}
					Inverse insertion of $F \Omega_{\mathcal{C}}$.
					Note that the two $F\Omega_{\mathcal{C}}$'s have swapped roles during the process.
				\end{minipage}
			\\
			& = & \tikzbo{
				\draw[diagram] (-\x,1.5) to[out=-90,in=90] (\x,0);
				\encircle[angle=150,label={$F\Omega_{\mathcal{C}}$}]{\Aid}
				\draw[diagram,directed=0.95] (-\x,-1.5) node[below] {$FX$} to[out=90,in=-90] (\x,0);
			}
			&
		\end{align}
		The non-obvious part of the calculation is \eqref{TwistingFIntertwinersAround}.
		The assumption that $\{\iota_i\}$ and $\{\iota^i\}$ form dual bases with respect to the spherical pairing looks like this in the graphical calculus:
		\begin{equation}
		\tikzbo{
			\node[morbox] (A) at (0,0.4) {$\iota_i$};
			\node[morbox] (C) at (0,-0.4) {$\iota^j$};
			\draw[thick,directed] (C.115) to[out=90,in=270] (A.255);
			\draw[thick,opdirected] (C.65) to[out=90,in=270] (A.295);
			\draw[thick,opdirected] (C.270) to[out=270,in=0] (-0.2,-0.8) to[out=180,in=180] (-0.2,0.8) to[out=0,in=90] (A.90);
		} = \delta_{i,j}
		\end{equation}
		It is necessary to show that this property is also true for $\{\tilde \iota_i\}$ and $\{\tilde \iota^j\}$.
		After composing the $F\iota_i$ and $F\iota^j$ with evaluations and coevaluations, this again results in $F$ applied to morphisms $\{\tilde \iota_i\}$ and $\{\tilde \iota^j\}$
		since $F$ is monoidal and therefore preserves duals (up to the natural isomorphism $F^2$ which is implicit here):
		\tikzset{
			intylabel/.style={above},
			intxzlabel/.style={below},
			inttopdown/.style={
				intylabel/.style={below},
				intxzlabel/.style={above}
			},
			pics/intertwiner/.style 2 args={code={
				\node[morbox] (B) at (0,0) {#1};
				\coordinate (Z) at (-0.32,0.5);
				\coordinate (Y) at (0,-0.5);
				\coordinate (X) at (0.32,0.5);
				\draw[thick,#2] (B) -- (Y) node[intylabel] {$FY$};
				\draw[thick,#2] (B) -- (Z) node[intxzlabel] {$F\!Z$};
				\draw[thick,#2] (X) node[intxzlabel] {$F\!X$} -- (B);
			}},
			pics/intertwinercupcap/.style 2 args={code={
				\node[morbox] (B) at (0,0) {#1};
				\coordinate (X) at (-0.2,0.5);
				\draw[thick,#2] (B) -- (0,-0.25) arc (0:-180:0.4cm) -- +(0,0.75) node[intxzlabel] {$F\!Z$};
				\draw[thick,#2] (X) node[intxzlabel] {$F\!X$} -- (B);
				\draw[thick,#2] (B) -- (0.15,0.2) arc (150:0:0.4cm) -- +(0,-0.75) coordinate node[intylabel] {$FY$};
			}}
		}
		\newcommand{\intertwinereq}[3]{\tikzbo{
				\begin{scope}[scale=1.2]
					\pic[#1] (int-) {intertwiner={$F \tilde \iota#2$}{#3}};
				\end{scope}
			}
			\coloneqq
			\tikzbo{
				\begin{scope}[scale=1.2]
					\pic[#1] (intcc-) {intertwinercupcap={$F\iota#2$}{#3}};
				\end{scope}
			}
		}
		\begin{equation}
			\intertwinereq{yscale=-1}{_i}{directed} \qquad \intertwinereq{inttopdown}{^j}{opdirected}
			\label{def:intertwiner tilde}
		\end{equation}
		It is necessary to show now that $\{\tilde \iota_i\}$ and $\{\tilde \iota^j\}$ are dual bases again.
		But this follows from pivotality of $F$ (preservation of traces) and sphericality of $\mathcal{D}$:
		\tikzset{
			/dualbasesbare/leftconn/.style={opdirected},
			/dualbasesbare/rightconn/.style={opdirected},
			dualbasesbare/.pic={
				\node[morbox] (A) at (0,0.4) {$#1_i$};
				\node[morbox] (C) at ($(0,0)-(A)$) {$#1^j$};
				\coordinate (side) at (-0.1,0);
				\coordinate (top) at (0,0.18);
				\coordinate (right) at (0.43,0);
				\draw[thick,/dualbasesbare/leftconn] (A.255) to[out=270,in=90] (C.115);
				\draw[thick,directed] (A) to[out=90,in=0] ($(side)+(A.90)+(top)$) to[out=180,in=180] ($(side)+(C.270)-(top)$) to[out=0,in=270] (C);
			},
			dualbases/.pic={
				\pic () {dualbasesbare=#1};
				\draw[thick,/dualbasesbare/rightconn] (A.295) to[out=270,in=90] (C.65);
			},
			dualbasesaround/.pic={
				\pic () {dualbasesbare=#1};
				\draw[thick,directed=0.51,looseness=1.3] (A.295) to[out=270,in=220] ($(right)+(top)$) to [out=40,in=0] ($(side)+(A.90)+2*(top)$) to[out=180,in=180] ($(side)+(C.270)-2*(top)$) to [out=0,in=-40]  ($(right)-(top)$) to[out=140,in=90] (C.65);
			},
			dualbasesturnip/.pic={
				\pic () {dualbasesbare=#1};
				\draw[thick,directed] (A.295) to[out=270,in=220] ($(right)+(top)$) to [out=40,in=180] ($1.5*(right)+3*(top)$) to[out=0,in=0] ($1.5*(right)-3*(top)$) to [out=180,in=-40]  ($(right)-(top)$) to[out=140,in=90] (C.65);
			},
			flipleftconn/.style={/dualbasesbare/leftconn/.style={directed}},
			fliprightconn/.style={/dualbasesbare/rightconn/.style={directed}}
		}
		\begin{align}
			\hphantom{\stackrel{\eqref{def:intertwiner tilde}}{=}} &\tikzbo{\pic[flipleftconn] {dualbases=\tilde\iota};}
			&&\stackrel{\hphantom{\text{sphericality}}}{=} F\left(\tikzbo{\pic[flipleftconn] {dualbases=\tilde\iota};}\right)
			&&\stackrel{\text{pivotality}}{=} \tikzbo{\pic[flipleftconn] {dualbases=F\tilde\iota};}
			\nonumber\\
			\stackrel{\eqref{def:intertwiner tilde}}{=} &\tikzbo{\pic {dualbasesaround=F\iota};}
			&& \stackrel{\text{sphericality}}{=} \tikzbo{\pic {dualbasesturnip=F\iota};} 
			&& \stackrel{\hphantom{\text{pivotality}}}{=} \tikzbo{\pic[fliprightconn] {dualbases=F\iota};}
			\nonumber\\
			\stackrel{\text{pivotality}}{=} &F\left(\tikzbo{\pic[fliprightconn] {dualbases=\iota};}\right)
			&&\stackrel{\hphantom{\text{sphericality}}}{=} \tikzbo{\pic[fliprightconn] {dualbases=\iota};}
			&&\stackrel{\hphantom{\text{pivotality}}}{=} \delta_{i,j}
		\end{align}
		In words, pivotal functors preserve dual bases (with respect to the spherical pairing).
\end{proof}
\begin{lemma}
	\label{lem:generalised sliding with braidings}
	The previous lemma holds as well when the encircling by $F\Omega_\mathcal{C}$ is an arbitrary framed knot,
	and also if the single strand $A$ is generalised to multiple strands
	(i.e. a tensor product),
	which may be arbitrarily linked with the encircling.
\end{lemma}
\begin{proof}
	Braidings and twists are natural transformations and can therefore be pushed past the $F \tilde \iota_i$,
	so they will be passed on to the new encircling morphism and the slid handle.
	Assume e.g. that after \eqref{TwistingFIntertwinersAround},
	there still is a twist on $FY$.
	Then the right hand side of the diagram is:
	\tikzmath{
		\height = 1;
		\length = 0.6;
		\width  = 0.7;
		\boxwidth = 0.12;
	}
	\begin{align}
		\tikzbo{
			\node[morbox] (i*) at (0, \height) {$F\tilde{\iota}^i$};
			\node[morbox] (i)  at (0,-\height) {$F\tilde{\iota}_i$};
			\draw[thick] (i*.65)  -- +(0,2*\length);
			\draw[thick] (i*.115) -- +(0,2*\length);
			\draw[thick]  (i.245) -- +(0, -\length);
			\draw[thick]  (i.295) -- +(0, -\length);
			\draw[thick]
				(i*) to[out=270, in=270]
				(\width, 0)
					node[right] {$FY$};
			\draw[diagram]
				(i)  to[out= 90, in= 90]
				(\width, 0);
		}
		=
		\tikzbo{
			\node[morbox] (i*)   at (0, \height+\length) {$F\tilde{\iota}^i$};
			\node[morbox] (i)    at (0, \height)         {$F\tilde{\iota}_i$};
			\node         (ilow) at (0,-\height)         {\phantom{$F\tilde{\iota}_i$}};
			\draw[thick]   (i)      -- (i*);
			\draw[thick]   (i*.65)  -- +(0, \length);
			\draw[thick]   (i*.115) -- +(0, \length);
			\draw[thick] (ilow.245) -- +(0,-\length);
			\draw[thick] (ilow.295) -- +(0,-\length);
			\draw[thick] (ilow.245) -- (ilow.115);
			\draw[thick] (ilow.295) -- (ilow.65);
			\draw[thick, looseness=2]
				(\width+\boxwidth,0) to[out=270, in=270]
				(i.245);
			\draw[thick, looseness=1.5]
				(\width-\boxwidth,0) to[out=270, in=270]
				(i.295);
			\draw[diagram, looseness=2]
				(ilow.115)           to[out= 90, in= 90]
				(\width+\boxwidth,0);
			\draw[diagram, looseness=1.5]
				(ilow.65)            to[out= 90, in= 90]
				(\width-\boxwidth,0);
		}
	\end{align}
	The single strand of the encircling is replaced by the sliding strand and the new encircling strand, by \emph{cabling}.
	Therefore, the encircling may have an arbitrary framing,
	which is passed on to the sliding strand.
	
	This argument can be easily generalised to braidings,
	and thus holds for knots and links.
\end{proof}

\begin{remark}
	It is remarkable that it is not necessary to demand $\mathcal{C}$ is ribbon, neither that $F$ is braided or ribbon.
	In fact, the proof this lemma stems from a better-known sliding lemma in spherical fusion categories used for example in understanding the Hilbert spaces assigned to surfaces in the Turaev-Viro-Barrett-Westbury-TQFT \cite[Corollary 3.5]{Kirillov2011:Stringnet}.
	
	For $F$ the identity of a ribbon fusion category, the sliding lemma would have followed directly from the killing property, as demonstrated in Lemma \ref{sliding property}.
	But if $F$ is not the identity, it is unclear whether there is an analogue to the killing property.
\end{remark}

\subsection{The definition}

The definition of the generalised dichromatic invariant can now be given.
\begin{defn}
	\label{definvariant}
	Assume the following:
	\begin{itemize}
		\item Let $\mathcal{C}$ be a spherical fusion category.
		\item Let $\mathcal{D}$ be a ribbon fusion (premodular) category with trivial twist on all transparent objects.
		\item Let $F\colon \mathcal{C} \to \mathcal{D}$ be a pivotal functor.
		\item Let $L$ be the special framed link obtained from a handlebody decomposition of a smooth, oriented, closed four-manifold $M$.
	\end{itemize}
	Then the \textbf{generalised dichromatic invariant} of $L$ associated with $F$ is defined as:
	\begin{align}
		I_F(L) \coloneqq \frac{\left<L\left(\Omega_{\mathcal{D}},F\Omega_{\mathcal{C}}\right)\right>}{\qdim{\Omega_{\mathcal{C}}}^{h_2-h_1} \left(\qdim{\Omega_{\mathcal{D}}}\qdim{\left(F\Omega_{\mathcal{C}}\right)'}\right)^{h_1}}
		\label{the invariant}
	\end{align}
	Here, $h_i$ is the number of $i$-handles of the handle decomposition, or, the number of components in the first, respective, second set of the special framed link.
	$\left<L\left(\Omega_{\mathcal{D}},F\Omega_{\mathcal{C}}\right)\right>$ is the evaluation of the special framed link diagram as an endomorphism of $I_\mathcal{D}$,
	or equivalently a complex number,
	as in Definition \ref{def:eval of links}.
	The 1-handles are labelled with $\Omega_{\mathcal{D}}$,
	the 2-handles with $F\Omega_{\mathcal{C}}$.
	$\left(F\Omega_{\mathcal{C}}\right)'$ is the transparent part of $F\Omega_{\mathcal{C}}$,
	as in Definition \ref{transparent colour}.

\end{defn}
\begin{rema}
	It might be counter-intuitive that the unknotted, 0-framed, unlinked 1-handles are labelled by $\Omega_{\mathcal{D}}$,
	while the 2-handles are labelled by $F\Omega_{\mathcal{C}}$,
	but $\mathcal{D}$ is the ribbon category (which has algebraic counterparts of knots and framings) and $\mathcal{C}$ is only spherical.
	But this is indeed a valid definition,
	while a functor in the other direction does not lead to an invariant in an obvious way.
	
	Note also that $F\Omega_\mathcal{D}$ does not depend on the monoidal coherences $F^2$ and $F^0$.
	Two functors with different coherences will give the same invariant.
	Furthermore, any two isomorphic functors will also yield the same invariant.
\end{rema}

From now on,
the conditions in the definition will be assumed,
unless stated otherwise.

\subsection{Proof of invariance}

\label{proof of invariance}

\begin{lem}[Multiplicativity under disjoint union]
	\label{multiplicativity under disjoint union}
	For two links $L_1$ and $L_2$, $I_F$ is multiplicative under disjoint union $\sqcup$:
	\begin{equation}
		I_F(L_1 \sqcup L_2) = I_F(L_1) \cdot I_F(L_2)
	\end{equation}
\end{lem}

	\begin{proof}
		Evaluation of the graphical calculus is multiplicative under disjoint union:
		A link corresponds to an endomorphism of $\Co$, so two links correspond to an endomorphism of $\Co \otimes \Co$.
		The evaluation is a monoidal functor with coherence isomorphism $- \cdot - \colon \Co \otimes \Co \to \Co$, so the numerator of $I_F$ is multiplicative.
		Obviously, $h_i(L_1 \sqcup L_2) = h_i(L_1) + h_i(L_2)$, so the denominator is multiplicative as well.
	\end{proof}

\tikzset{
	/slid/throughlabel=,
	/slid/throughstyle/.style={diagram},
	/slid/aroundstyle/.style={diagram},
	slide11/.style={/slid/arounddot/.style={},/slid/circledot/.style={}},
	slide21/.style={/slid/circledot/.style={}},
	pics/handle21/.style n args={3}{code={
		\encircle[angle=160,label=#1]{
			\begin{scope}[xshift=1.4cm,rotate=-30]
				\trefoil{#2}
			\end{scope}
			\node[right=0.5cm] at (C21) {#3};
		}
	}}
}

Given two different handle decompositions of a manifold can be transformed into each other by a series of handle slides and handle cancellations, as described for example in \cite[Theorem 4.2.12]{GompfStipsicz}.
The relevant moves for link diagrams of four-manifolds have been studied in \cite{deSa} and are explained further in \cite[Section 5.1]{GompfStipsicz}.
They are shown in Table \ref{table:handle moves}.
\begin{table}[!ht]
	\begin{center}
		\begin{tabular}{lcc}
			\toprule
			Handle move
				& before
					& after
			\\\midrule
			1-1-handle slide
				& \tikzbo{\pic[slide11] {unaround={}};} \tikzbo{\pic[slide11] {unthrough handle={}{}};}
					& \tikzbo{\pic[slide11] {slid handle={}{}};}
			\\\midrule
			2-1-handle slide
				& \tikzbo{\pic {unaround={}};} \tikzbo{\pic[slide21] {unthrough handle={}{}};}
					& \tikzbo{\pic[slide21] {slid handle={}{}};}
			\\\midrule
			2-2-handle slide
				& \tikzbo{\pic {unaround={}};} \tikzbo{\pic {unthrough handle={}{}};}
					& \tikzbo{\pic {slid handle={}{}};}
			\\\midrule
			1-2-handle cancellation
				& \tikzbo{\pic {handle21={}{}{}}; \node[dot] at (-1,0) {};}
					& (empty)
			\\\midrule
			2-3-handle cancellation
				& \tikzbo{\dimcircle{}{}}
					& (empty)
			\\\bottomrule
		\end{tabular}
	\end{center}
\caption{Handle moves and cancellations for 4-handlebodies.
As usual, a dot denotes a 1-handle.
The grey area stands for an arbitrary number of 1- and 2-handles passing through.
Note, that for the 1-2-handle cancellation,
the 2-handle may be knotted arbitrarily,
but not linked to other handles.
In the 2-2-handle slide,
the 2-handle on the right hand side can be arbitrarily knotted and linked,
in which case the sliding handle needs to follow the blackboard framing.
}
\label{table:handle moves}
\end{table}

\begin{thm}[Independence of handlebody decomposition]
	The generalised dichromatic invariant is independent of the handlebody decomposition and is thus an invariant of smooth four-manifolds.
\end{thm}

	\begin{proof}
		It is only necessary to check invariance of $I_F$ under each of the handle moves in order to prove the theorem.
		\begin{itemize}
			\item Invariance under the 1-1-handle slide and the 2-1-handle slide are ensured by the sliding property \ref{sliding property}. Since 1- and 2-handles are labelled with objects in $\mathcal{D}$, they can slide over a 1-handle which is labelled with $\Omega_\mathcal{D}$.
			\item Invariance under the 2-2-handle slide is ensured by the generalised sliding property \ref{generalised sliding property},
				and its adaption to arbitrary knots and links in Lemma \ref{lem:generalised sliding with braidings}.
				Every object in the image of $F$ can slide over $F\Omega_\mathcal{C}$, so since 2-handles are labelled with $F\Omega_\mathcal{C}$, they can slide over each other.
			\item The 1-2-handle cancellation leaves $I_F$ invariant because of its normalisation.
				Assume that there is a linked pair of a 1-handle and a 2-handle that is not linked to the rest of the diagram.
				Then it will be shown that $I_F$ does not change if the pair is removed from the diagram.
				The 2-handle can be knotted, as is illustrated here with a trefoil knot.
				Since $I_F$ is multiplicative under disjoint union of link diagrams, it only remains to show that the invariant of the pair of handles evaluates to 1.
				The numerator is just the evaluation of the graphical calculus:
				\begin{align}
					\left<
					\tikzbo{
						\pic {handle21={}{}{}};
						\node[dot] at (-1,0) {};
					}
					\right>
					&=
					\tikzbo{
						\node at (-1,0) {};
						\pic {handle21={$\Omega_\mathcal{D}$}{}{$F\Omega_\mathcal{C}$}};
					}
					\nonumber\\
					&=
					\tikzbo{
						\pic {handle21={$\Omega_\mathcal{D}$}{,dotted}{$\left(F\Omega_\mathcal{C}\right)'$}};
					}\nonumber\\
					&=
					\tikzbo{\dimcircle{$\Omega_\mathcal{D}$}}
					\tikzbo{\dimcircle[dotted]{$\left(F\Omega_\mathcal{C}\right)'$}}
					\nonumber\\
					&= \qdim{\Omega_\mathcal{D}}\qdim{\left(F\Omega_\mathcal{C}\right)'}
				\end{align}
				The number of 1-handles and 2-handles are both 1,
				so the denominator equals the above expression,
				thus the invariant is 1.
				
				Note that it was necessary here to demand that the twist is trivial on transparent objects.
			\item Invariance under the 2-3-handle cancellation is even easier to show:
				3-handles don't appear in the link picture.
				A 2-3-handle cancellation thus amounts to the removal of an unlinked, unknotted 2-handle.
				By a similar argument as before, one can evaluate the invariant on the link diagram of such a 2-handle and find that it is 1 as well.
		\end{itemize}
	\end{proof}

\begin{remark}
	For a manifold $M$ and a special framed link $L$ representing a handle decomposition of it,
	$I_F$ has now been shown not to depend on the choice of the link $L$.
	From here on, the notation $I_F(M)$ will be used,
	and stands for $I_F(L)$
	with an arbitrary choice of $L$.
\end{remark}

\begin{rema}
	Pivotality of $F$ is essential for the invariance of $I_F$.
	As an easy counterexample, take the category of super vector spaces,
	which is defined as follows.
	As monoidal category, choose the category of finite dimensional representations of $\Z_2$.
	Choose the pivotal structure such that the sign representation $\sigma$ has dimension $-1$.
	
	There is an obvious forgetful strong monoidal functor $U$ to vector spaces sending both simple objects to $\Co$.
	This functor is \emph{not} pivotal since the dimension of $\Co$ is $+1$.
	One finds that the evaluation of the (undotted) unknot is
	\begin{align*}
		\qdim{U\Omega_{\Rep(\Z_2)}}
		&= \sum_{\mathclap{X \in \Lambda_{\Rep(\Z_2)}}} \qdim{X} \qdim{UX}\\
		&= 1 \cdot 1 + (-1) \cdot 1 = 0
	\end{align*}
	However, the corresponding manifold is $S^4$ and the empty diagram
	(which would result from cancelling the single 2-handle with a 3-handle)
	evaluates to 1.
	It is apparent now that a non-pivotal functor can break invariance.
\end{rema}

\subsection{Simply-connected manifolds and multiplicativity under connected sum}

\label{ssection:sc manifolds and multiplicativity}

As was shown in Lemma \ref{multiplicativity under disjoint union},
the generalised dichromatic invariant is multiplicative under disjoint union of link diagrams.
This operation, in turn, corresponds to connected sum of manifolds.
As a consequence, for two manifolds $M^1$ and $M^2$,
the invariant satisfies $I_F(M^1 \connsum M^2) = I_F(M^1) \cdot I_F(M^2)$,
where $\connsum$ denotes connected sum.
This has far reaching consequences,
as is shown in the following known lemma.

\begin{lem}
	\label{simply-connected}
	Assume $I$ is any invariant of oriented, closed four-manifolds that is multiplicative under connected sum on simply-connected manifolds.
	Furthermore, assume that $I\left(\CP 2\right)$ and $I\left(\opCP 2\right)$ are invertible.
	Then $I$ is given on a \emph{simply-connected} four-manifold $M$ by
	\begin{align}
		I(M) &= \left(I\left(\CP 2\right) I\left(\opCP 2\right)\right)^{-1+\frac{\chi(M)}{2}} \left(\frac{I\left(\CP 2\right)}{I\left(\opCP 2\right)}\right)^{\frac{\sigma(M)}{2}}
		\label{eq:simply connected invariant}
	\end{align}
	$\chi$ and $\sigma$ are Euler characteristic and signature, respectively.
\end{lem}
\begin{proof}
	The first, and by Poincar\'e duality third, homologies of $M$ are trivial, so the Euler characteristic $\chi(M)$ is equal to $2 + b_2(M)$, where $b_2(M) = b^+_2(M) + b^-_2(M)$ is the rank of the second homology and $b^\pm_2(M)$ are the dimensions of the subspaces on which the intersection form is positive or negative.
	Since the signature is $\sigma(M) = b^+_2(M) - b^-_2(M)$, then it follows that $b^\pm_2(M) = (\chi(M) \pm \sigma(M))/2 - 1$.
	
	But it is well-known \cite[Corollary 9.1.14]{GompfStipsicz} that simply-connected manifolds stably decompose into $\CP 2$ and $\opCP 2$ under connected sum,
	i.e. there exist natural numbers $m, n^+, n^-$ such that:
	\begin{equation}
		 M \connsum^m \CP 2 \connsum^m \opCP 2
		 \cong \connsum^{n^+} \CP 2 \connsum^{n^-} \opCP 2
	\end{equation}
	($M \connsum^n N$ denotes the connected sum of $M$ and $n$ copies of $N$.)
	By comparing the intersection forms on both sides,
	one sees that that the numbers of positive and negative eigenvalues are $b^\pm_2(M) = n^\pm - m$.
	Therefore by multiplicativity under connected sum:
	\begin{align*}
		I(M) I\left(\CP 2\right)^m I\left(\opCP 2\right)^m
		&= I\left(\CP 2\right)^{n^+} I\left(\opCP 2\right)^{n^-}\\
		\implies I(M) &= I\left(\CP 2\right)^{b^+_2(M)} I\left(\opCP 2\right)^{b^-_2(M)}
	\end{align*}
	Now \eqref{eq:simply connected invariant} follows easily.
\end{proof}
\begin{remark}
	Such invariants cannot distinguish the homotopy-inequivalent manifolds $S^2 \times S^2$ and $\Co P^2 \connsum \opCP 2$.
	In particular, these manifolds have different intersection forms,
	but the same signature.
	Effectively, invariants with the above properties are insensitive to this homotopical information.
\end{remark}

\begin{lem}
	The generalised dichromatic invariant is invertible on $\CP 2$ and $\opCP 2$.
\end{lem}
\begin{proof}
	This is best seen by directly calculating the invariants on these manifolds.
	It is known that $\CP 2 \connsum \opCP 2 \cong S^2 \tilde{\times} S^2$,
	where the latter denotes the total space of a twisted $S^2$-bundle over $S^2$,
	which has the following Kirby diagram \cite[Figure 4.34]{GompfStipsicz}:
	\tikzmath{
		\radiusS2  = 0.54;
		\widthkink = 0.23;
	}
	\begin{equation*}
		S^2 \tilde{\times} S^2 =
		\tikzbo{
			\draw[diagram]
				(0,0) arc [radius=\radiusS2, start angle=0, end angle=170]
				to[out=260, in=270, looseness=1.5] (\widthkink-2*\radiusS2,0);
			\dimcircle{}{}
			\draw[diagram]
				(0,0) arc [radius=\radiusS2, start angle=360, end angle=190]
				to[out=100, in=90, looseness=1.5] (\widthkink-2*\radiusS2,0);
		}	
	\end{equation*}
	To the show the invertibility of both $I\left(\CP 2\right)$ and $I\left(\opCP 2\right)$,
	calculate the following:
	\begin{align*}
		I\left(\CP2\right) \cdot I\left(\opCP2\right)
		&= I\left(\CP2 \connsum \opCP2\right)\\
		&= \frac{\left<L_{S^2 \tilde{\times} S^2} \left(\Omega_{\mathcal{D}},F\Omega_{\mathcal{C}}\right) \right>}{\qdim{\Omega_{\mathcal{C}}}^{h_2-h_1} \left(\qdim{\Omega_{\mathcal{D}}}\qdim{\left(F\Omega_{\mathcal{C}}\right)'}\right)^{h_1}}
		\intertext{The killing property \ref{Killing property} and the handle numbers $h_1 = 0$, $h_2 = 2$ give:}
		&= \frac{\qdim{F\Omega_{\mathcal{C}}} \sum_{X \in \Lambda_\mathcal{C}} \tr \left(\theta_{(FX)'}\right)}{\qdim{\Omega_{\mathcal{C}}}^2}
		\intertext{Recall that the twist is required to be trivial on transparent objects in $\mathcal{D}$.
		Furthermore, $F$ is pivotal and preserves quantum dimensions.}
		&= \frac{\qdim{(F\Omega_{\mathcal{C}})'}}{\qdim{\Omega_{\mathcal{C}}}}\numberthis
	\end{align*}
	Since $F\Omega_{\mathcal{C}}$ contains at least the monoidal unit,
	the result cannot be 0.
\end{proof}
\begin{corollary}
	Lemma \ref{simply-connected} applies to the generalised dichromatic invariant.
\end{corollary}
\begin{proof}
	${I\left(\CP 2\right) \cdot I\left(\opCP 2\right)}$ is invertible due to the previous lemma.
	Multiplicativity under connected sum has already been shown in Lemma \ref{multiplicativity under disjoint union}.
\end{proof}
It remains to calculate the invariants of $\CP 2$ and $\opCP 2$ in order to be able to give concrete values for simply-connected manifolds.
$\CP 2$ can be composed of a 0-handle, a 2-handle and a 4-handle.
A link diagram for it is given by an unknotted circle with framing $+1$,
denoted by $L_{+1}$.
The value of the invariant is therefore:
\begin{align}
	I\left(\CP 2\right)
	&= \frac{\left<L_{+1} \left(\Omega_{\mathcal{D}},F\Omega_{\mathcal{C}}\right) \right>}{\qdim{\Omega_{\mathcal{C}}}^{h_2-h_1} \left(\qdim{\Omega_{\mathcal{D}}}\qdim{\left(F\Omega_{\mathcal{C}}\right)'}\right)^{h_1}}\nonumber\\
	&= \frac{\sum_{X \in \Lambda_\mathcal{C}} \tr \left(\theta_{FX}\right)}{\qdim{\Omega_{\mathcal{C}}}}
	\label{eq:CP2}
	\intertext{Analogously:}
	I\left(\opCP 2\right)
	&= \frac{\sum_{X \in \Lambda_\mathcal{C}} \tr \left(\theta^{-1}_{FX}\right)}{\qdim{\Omega_{\mathcal{C}}}}
\end{align}
For many cases of $F$, more concrete values can be calculated.
This is done in Section \ref{sssection:sc manifolds}.

\subsection{Petit's dichromatic invariant and Broda's invariants}
\label{petit broda}
Broda defined two invariants of four-manifolds using the category of tilting modules for $U_qsl(2)$ at a root of unity \cite{Broda,Roberts:1995SkeinTheoryTV}.
The original invariant, called here the \textbf{Broda invariant},
labelled both 1- and 2-handles with simple objects in this category (the ``spins''),
whereas the \textbf{refined Broda invariant} labelled 2-handles with just the integer spins.
The Broda invariants were investigated by Roberts \cite{Roberts:1995SkeinTheoryTV, Roberts1997:RefinedInvariants},
who showed that the Broda invariant depends on the signature of the four-manifold whereas the refined Broda invariant detects also the first Betti number with $\Z_2$ coefficients,
and is sensitive to the second Stiefel-Whitney class (which decides whether the manifold admits a spin structure).

Generalising Broda's constructions to other ribbon fusion categories leads to the following two classes of examples,
which will turn out to be special cases of the generalised dichromatic invariant.
\begin{exam}
	\label{identity}
	Petit recovers \cite[Remark 4.4]{Petit:dichromatic},
	up to a factor depending on the Euler characteristic,
	a \textbf{generalised Broda invariant} for a (not necessarily modular) ribbon fusion category $\mathcal D$ satisfying the conditions of Definition \ref{definvariant}.
	Petit shows that this invariant depends only on the signature (and Euler characteristic) of the four-manifold.
	
	This invariant will turn out to be the generalised dichromatic invariant associated to the identity functor $\eins_\mathcal{D}\colon \mathcal{D} \to \mathcal{D}$,
	as can be seen from the next, more general example.
\end{exam}
\begin{exam}
	\label{dichromatic}
	The refined Broda invariant,
	which will be discussed again in Section \ref{sssec:refined Broda invariant},
	can be generalised to arbitrary ribbon fusion subcategories.

	Let $\mathcal{C}$ and $\mathcal{D}$ be ribbon fusion categories,
	satisfying the conditions of Definition \ref{definvariant}.
	(Still, neither category is required to be modular.)
	For a full ribbon inclusion functor $F\colon \mathcal{C} \hookrightarrow \mathcal{D}$,
	\textbf{Petit's dichromatic invariant} $I_0$ \cite[(4.4)]{Petit:dichromatic} is recovered,
	again up to a factor depending on the Euler characteristic $\chi(M)$,
	which will be calculated in the following.
	
	The notation in \cite{Petit:dichromatic} is subtly different:
	$\mathcal{C}'$ denotes an arbitrary subcategory there,
	not necessarily the symmetric centre.
	Also, the notation for categorical dimensions is different from this presentation.
	Redefining the symbols from \cite{Petit:dichromatic} in the notation established here gives $\Delta_\mathcal{C} \coloneqq \qdim{\Omega_\mathcal{C}}$ and $\Delta_{\mathcal{D},\mathcal{C}}'' \coloneqq \qdim{\left(F\Omega_\mathcal{C}\right)'}$.
	
	The nullity of the linking matrix of the link diagram has to be introduced,
	but since $M$ is closed, it equals $h_3$, the number of 3-handles.
	Petit's invariant is then in the present notation:
	\begin{equation}
		I_0(L) \coloneqq \frac{\left<L\left(\Omega_\mathcal{D}, \Omega_\mathcal{C}\right)\right>}{\qdim{\Omega_\mathcal{C}}^{h_3}\left(\qdim{\Omega_\mathcal{D}} \qdim{(F\Omega_\mathcal{C})'}\right)^{\frac{h_1+h_2-h_3}{2}}}
	\end{equation}
	Note that the numerators of $I_F$ and $I_0$ do not differ, but the normalisations do.
	To compare the normalisation of invariants,
	their ratio is calculated using a handle decomposition with exactly one 0-handle and 4-handle.
	
	The ratio of invariants is then
	\begin{align}
		\frac{I_F(M)}{I_0(M)} & = \frac{\qdim{\Omega_\mathcal{C}}^{h_3} \cdot \left(\qdim{\Omega_{\mathcal{D}}}\qdim{\left(F\Omega_{\mathcal{C}}\right)'}\right)^{\frac{h_1+h_2-h_3}{2}} }{\qdim{\Omega_\mathcal{C}}^{h_2-h_1} \cdot \left(\qdim{\Omega_{\mathcal{D}}}\qdim{\left(F\Omega_{\mathcal{C}}\right)'}\right)^{h_1}}\nonumber\\
		& = \left(\frac{\sqrt{\qdim{\Omega_{\mathcal{D}}}\qdim{\left(F\Omega_{\mathcal{C}}\right)'}}}{\qdim{\Omega_\mathcal{C}}}\right)^{\chi(M)-2}
		\label{eq:petit euler factor}
	\end{align}
	\label{petit}
	The same calculation can be used to show that the refined Broda invariant from \cite{Broda} is Petit's invariant $I_0$ for the subcategory of integer spins in the category of tilting modules of $U_qsl(2)$.
\end{exam}
\begin{rema}
	Whenever a full inclusion into a ribbon category is encountered, it will be assumed that the subcategory inherits braiding and ribbon structures from the bigger category.
	Also, it will be assumed that the canonical pivotal structure is chosen on both sides,
	which is then automatically preserved.
\end{rema}

\begin{rema}
	Petit called his invariant ``dichromatic'' since the special framed link arising from the handle decomposition is labelled with two different Kirby colours.
	The invariant presented here uses two different colours as well,
	so it seems appropriate to keep the name ``dichromatic'',
	but to point out that it is somewhat more general.
\end{rema}

\section{Simplification of the invariant}
\label{sec:simplification}

Here it is shown that a general argument allows the generalised dichromatic invariant to be simplified in many cases.
\begin{pro}
	\label{invariant reduces}
	Let $\mathcal{A} \xrightarrow{F} \mathcal{B} \xrightarrow{G} \mathcal{C} \xrightarrow{H} \mathcal{D}$ be a chain of pivotal functors on spherical fusion categories.
	Let furthermore $\mathcal{C} \xrightarrow{H} \mathcal{D}$ be ribbon,
	and let the symmetric centres (Definition \ref{def:symmetric centre}) $\mathcal{C}'$ and $\mathcal{D}'$ have trivial twist.
	Assume these three conditions on $F$ and $H$, for some $m, n \in \Co$:
	\begin{align*}
		F\Omega_\mathcal{A} &= n \Omega_\mathcal{B}\\
		H\Omega_\mathcal{C} &= m \Omega_\mathcal{D}\\
		H\left(\left(G\Omega_\mathcal{B}\right)'\right) &= \left(HG\Omega_\mathcal{B}\right)'
	\end{align*}
	Then $I_{HGF} = I_G$.
\end{pro}
	\begin{proof}
		Note that since $F$ and $H$ are pivotal,
		the values of $m$ and $n$ can be inferred by taking the dimensions on each side of the first two conditions:
		\begin{align*}
			\qdim{\Omega_\mathcal{A}} &= n \cdot \qdim{\Omega_\mathcal{B}}\\
			\qdim{\Omega_\mathcal{D}} &= m^{-1} \cdot \qdim{\Omega_\mathcal{C}}
			\intertext{Let now $L$ be a special framed link for the four-manifold $M$.}
			\left<L\left(\Omega_\mathcal{D}, HGF\Omega_\mathcal{A}\right)\right> &= \left<L\left(H\Omega_\mathcal{C}, HG\Omega_\mathcal{B}\right)\right> \cdot m^{-h_1} n^{h_2}\\
			&= \left<L\left(\Omega_\mathcal{C}, G\Omega_\mathcal{B}\right)\right> \cdot m^{-h_1} n^{h_2}
			\intertext{The first two assumptions were inserted,
			then it was used that $H$ is ribbon,
			to arrive at the enumerator of $I_F(M)$ up to the factors of $m$ and $n$.
			Using the first and the third assumption,
			the missing part in the denominator of $I_F(M)$ can be calculated:}
			\qdim{\left(HGF\Omega_\mathcal{A}\right)'} = n \cdot \qdim{\left(HG\Omega_\mathcal{B}\right)'} &= n \cdot \qdim{H \left(G\Omega_\mathcal{B}\right)'} = n \cdot \qdim{\left(G\Omega_\mathcal{B}\right)'}
		\end{align*}
		It is easy to see now that all factors of $n$ and $m$ cancel.
	\end{proof}

In the following, it is shown that there is an abundance of functors satisfying these conditions,
allowing a simplification of the generalised dichromatic invariant in many cases.
Examples include cases where either $H$ or $F$ is the identity functor.

\subsection{Simplification for unitary fusion categories}

One case, in which the generalised dichromatic invariant simplifies to Petit's dichromatic invariant is the case of unitary fusion categories,
which are certain non-degenerate $\Co$-linear $\dagger$-categories.
The unitarity condition is important in mathematical physics, and many examples are known.
The theory of unitary fusion categories is well developed, and many important properties are found in the literature, e.g. \cite{OnBraidedFusionCats}.
Instead of giving a self-contained introduction, the relevant known facts are listed.
\begin{itemize}
	\item A fusion $\dagger$-category with a rigid structure has a canonical spherical structure (see \cite[Lemma 7.5]{Selinger:graphical}) defined by the $\dagger$-structure and the chosen duals.
	\item A unitary functor, or $\dagger$-functor, is a functor that preserves the $\dagger$-structure.
	A strong monoidal unitary functor is pivotal, so it preserves the canonical spherical structure.
\end{itemize}

\begin{defn}
	A strong monoidal functor of fusion categories $F\colon \mathcal{C} \to \mathcal{D}$ is called \textbf{dominant} if for any object $Y \in \ob \mathcal{D}$ there exists an object $X \in \ob \mathcal{C}$ such that $Y$ is a subobject of $FX$.
	In \cite{OnFusionCategories} these are also known as ``surjective functors''.
\end{defn}
\begin{lem}
	\label{FPdimensions reduce}
	Let $F\colon \mathcal{C} \to \mathcal{D}$ be a dominant unitary functor of unitary fusion categories.
	Let furthermore both categories have the canonical spherical structure coming from the unitary structure.
	Then the following holds:
	\begin{align}
		F\Omega_\mathcal{C} &= \frac{\qdim{\Omega_\mathcal{C}}}{\qdim{\Omega_\mathcal{D}}} \Omega_\mathcal{D}
	\end{align}
\end{lem}
	\begin{proof}
		An analogous equation holds for the Frobenius-Perron dimensions \cite[Proposition 8.8]{OnFusionCategories}.
		In unitary fusion categories with the canonical spherical structure Frobenius-Perron dimensions and categorical dimensions coincide.
	\end{proof}
\begin{definition}
	For a strong monoidal functor $F$ of fusion categories,
	define the image category $\Ima F$ \cite[Definition 2.1]{OnBraidedFusionCats}.
	Its objects are all objects of $\mathcal{D}$ that are isomorphic to a subobject of $FX$,
	where $X$ is any object in $\mathcal{C}$.
	The morphisms of $\Ima F$ are chosen such that it is a full fusion subcategory of $\mathcal{D}$.
\end{definition}
\begin{lem}
	Let $F\colon \mathcal{C} \to \mathcal{D}$ be a strong monoidal functor of fusion categories.
	Then $F = F_2 \circ F_1$,
	where $F_1\colon \mathcal{C} \to \Ima F$ is a dominant functor,
	and $F_2\colon \Ima F \to \mathcal{D}$ the full inclusion from the previous definition.
\end{lem}
\begin{proof}
	By construction of the image category, $F$ factors through it,
	and $F$ restricted to $\Ima F$ is dominant.
\end{proof}

\begin{cor}
	\label{unitary is dichromatic}
	Let $F\colon \mathcal{C} \to \mathcal{D}$ be a strong monoidal unitary functor of unitary fusion categories,
	and again $\mathcal{D}$ ribbon such that its symmetric centre $\mathcal{D}'$ has trivial twist.
	Then $I_F = I_{F_2}$, and so is equal to Petit's dichromatic invariant $I_0$ for the inclusion $F_2\colon \Ima F \hookrightarrow \mathcal{D}$,
	multiplied by the Euler characteristic factor from \eqref{eq:petit euler factor}.
\end{cor}
	\begin{proof}
		Use the previous lemma to decompose $F$ into a dominant functor and a full inclusion.
		By the lemma before, the dominant part satisfies the conditions of Proposition \ref{invariant reduces},
		so $I_F$ is reduced to the invariant for the full inclusion.
		The fusion subcategory inherits the pivotal structure from $\mathcal{D}$.
		An invariant from a full inclusion is a case of Petit's dichromatic invariant, as explained in Example \ref{dichromatic}.
	\end{proof}

\subsection{Modularisation}
This subsection considers examples that will be compared to the Crane-Yetter invariant in Section \ref{examples invariant}.
\begin{defn}
	A ribbon fusion category $\mathcal{D}$ is called \textbf{modularisable} if its symmetric centre $\mathcal{D}'$ has trivial twist and dimensions in $\N$.
	For modularisable categories,
	there exists a faithful functor $H\colon \mathcal{D} \to \widetilde{\mathcal{D}}$ with $\widetilde{\mathcal{D}}$ modular,
	called the \textbf{modularisation} (also ``deequivariantisation'') of $\mathcal{D}$.
	Some standard references are \cite{Bruguieres:Modularisations} or \cite{Mueger:2000ModularClosure}.
\end{defn}
\begin{rems}
	\begin{itemize}
		\item $H$ is usually not full.
		\item The name ``deequivariantisation'' comes from thinking of $\mathcal{D}'$ as the representations of some finite group.
		$H$ restricted to $\mathcal{D}'$ then plays the role of a fibre functor, while not disturbing the nontransparent objects.
		$\widetilde{\mathcal{D}}$ has the same objects as $\mathcal{D}$, but additional isomorphisms from any transparent object to a direct sum of $I$s.
		\item For any symmetric fusion category without twist or pivotal structure,
		one can choose the trivial twist $\theta_X = 1_X$.
		With the corresponding pivotal structure, the categorical dimensions of objects are then in $\Z$.
		Alternatively, one can choose a pivotal structure with categorical dimensions in $\N$, but then the twist will usually not be trivial.
		To adhere to the conditions in Definition \ref{definvariant}, the trivial twist will always be chosen for symmetric fusion categories.
	\end{itemize}
\end{rems}

\begin{pro}
	\label{modularisation}
	Let $F\colon \mathcal{C} \to \mathcal{D}$ be pivotal with $\mathcal{C}$ spherical fusion and $\mathcal{D}$ modularisable.
	Such a functor satisfies the conditions of our invariant in Definition \ref{definvariant}.
	Let $H\colon \mathcal{D} \to \widetilde{\mathcal{D}}$ be the modularisation functor.
	Then $I_F = I_{H \circ F}$.
\end{pro}
	\begin{proof}
		In \cite[Proposition 3.7]{Bruguieres:Modularisations} it is stated that $H\Omega_\mathcal{D} \cong \qdim{\Omega_{\mathcal{D}'}} \Omega_{\widetilde{\mathcal{D}}}$.
		It is easy to check that ${\widetilde{\mathcal{D}}\left(I, H\left(X'\right)\right) = \widetilde{\mathcal{D}}\left(I, (HX)'\right)}$ follows from the original definition, and therefore $H\left(\left(F\Omega_\mathcal{C}\right)'\right) = \left(HF\Omega_\mathcal{C}\right)'$ since both sides are multiples of $I$.
		Thus, Proposition \ref{invariant reduces} can be applied.
	\end{proof}

Intuitively, the transparent objects on the 1-handles can be removed and don't contribute to the invariant.
The modularisation $H$ makes this explicit by sending all objects in $\mathcal{D}'$ to multiples of $I$.

One can make use of this fact by noting that many generalised dichromatic invariants are equal to an invariant arising from a functor into a modular category.
It is necessary to demand all dimensions of simple objects in $\mathcal{D}'$ to be positive, but this is the sole restriction.
In Section \ref{SSM}, it will be shown that invariants with a modular target category can be expressed in terms of a state sum and therefore extend to topological quantum field theories.

\begin{rema}
	The modularisation $H$ is not a full inclusion if the source $\mathcal{D}$ is not modular (and the identity otherwise).
	Therefore, the composition $H \circ F$ will usually not be full either, even if $F$ is.
	However, in the unitary case,
	the following corollary is helpful.
\end{rema}
\begin{cor}
	\label{unitary is CY}
	Let $F\colon \mathcal{C} \to \mathcal{D}$ be a strong monoidal unitary functor of unitary fusion categories.
	Let also $\mathcal{D}$ be modularisable,
	and $H\colon \mathcal{D} \to \widetilde{\mathcal{D}}$ the modularisation functor.
	Then there is a full inclusion $G\colon \Ima (H \circ F) \hookrightarrow \widetilde{D}$,
	and $I_F = I_G$.
\end{cor}
\begin{proof}
	From Proposition \ref{modularisation}, $I_{H \circ F} = I_F$, where $H$ is the modularisation.
	Therefore Corollary \ref{unitary is dichromatic} can be applied to $H \circ F$.
\end{proof}

\subsection{Cutting strands}
\label{sec:cutting strands}

If the target category $\mathcal{D}$ of the pivotal functor is modular, each 1-handle is labelled by $\Omega_\mathcal{D}$.
The strands of the 2-handles going through it can be cut, using Lemma \ref{cutting strands}.
This is the algebraic analogue of reverting from Akbulut's dotted handle notation to Kirby's original notation for handle decompositions where each 1-handle is represented by a pair of $D^3$s.
There is now a simpler definition of the generalised dichromatic invariant, which is obtained by cutting the strands through the 1-handles.
\begin{defn}
	\label{kirby direct labelling}
	Let $K$ be a Kirby diagram for a handle decomposition of a smooth, closed four-manifold $M$.
	Choose orientations on the $S^1$ of the attaching boundary of each 2-handle,
	and a choice of + and $-$ signs on the respective 3-balls for each 1-handle.
	\begin{enumerate}
		\item An \textbf{object labelling} is a map $X$ from the set of 2-handles to the set of simple objects in $\mathcal{C}$.
		The object assigned to the $i$-th 2-handle is written $X_i$.
		\item Now, for every 1-handle with 2-handles $i \in \{1, 2, \dots N\}$ entering or leaving the ball labelled with +,
		dual bases for the morphism spaces $\mathcal{D}(FX_1 \otimes FX_2 \otimes \cdots \otimes FX_N, I)$ and $\mathcal{D}\left(I, FX_1 \otimes FX_2 \otimes \cdots \otimes FX_N\right)$ are chosen.
		(The objects on leaving 2-handles are dualised.)
		
		A \textbf{morphism labelling} for a given object labelling is a choice of basis morphism for the +-ball of every 1-handle,
		and the corresponding dual morphism on the ball labelled with $-$.
		\item For a given object and morphism labelling, the evaluation of the labelling is the evaluation of the labelled diagram as a morphism in $\mathcal{D}(I,I) \cong \Co$,
		multiplied with the factor $\prod_{i} \qdim{X_i}$,
		where $i$ ranges over all 2-handles.
		\item The evaluation $\left<K(F)\right>$ of the Kirby diagram $K$ is the sum of evaluations over all labellings.
	\end{enumerate}
\end{defn}
\begin{pro}
	\label{kirby direct definition}
	Let $K$ be a Kirby diagram for a handle decomposition of a smooth, closed four-manifold $M$.
	Let $F\colon \mathcal{C} \to \mathcal{D}$ be a pivotal functor from a spherical fusion category to a modular category, and let $n$ be the multiplicity of $I$ in $F\Omega_\mathcal{C}$.
	Then the generalised dichromatic invariant is:
	\begin{equation}
		I_F(M) = \frac{\left<K(F)\right>}{\qdim{\Omega_\mathcal{C}}^{h_2-h_1}n^{h_1}}
	\end{equation}
\end{pro}
\begin{proof}
	Application of Lemma \ref{cutting strands} to the labelled special framed link $L$ shows:
	\[\left<L\left(\Omega_{\mathcal{D}}, F\Omega_{\mathcal{C}}\right)\right> =  \qdim{\Omega_\mathcal{D}}^{h_1}\left<K(F)\right>\]
	Since $\mathcal{D}$ is modular, $\qdim{\left(F\Omega_{\mathcal{C}}\right)'}=\qdim{nI}=n$ and the result follows.
\end{proof}

This proposition can be used as an alternative definition of the invariant in most cases.
However to prove invariance under all handle slides, it is more convenient to refer to the original Definition \ref{definvariant}.

\section{The state sum model}

\label{SSM}
The Crane-Yetter invariant is originally defined using a state sum model on a triangulation of a four-manifold \cite{CraneYetterKauffman:1997177}.
However, it was not presented as a state sum model in Section \ref{Crane Yetter}.
This is possible using a reformulation of the original definition due to Roberts,
as presented in \cite[Section 4.3]{Roberts:1995SkeinTheoryTV}.
He showed that for modular categories,
the Crane-Yetter state sum $CY$ is equal to the Broda invariant $B$ up to a normalisation involving the Euler characteristic,
through a process called ``chain mail'', which will be described in the following.

This is not true for nonmodular $\mathcal{C}$:
As will be shown in the next section, $CY$ and $B$ indeed differ in this case.
The nonmodular Crane-Yetter invariant arises from Petit's dichromatic invariant and does not depend only on the signature and Euler characteristic, but also at least on the fundamental group.

Previously, it wasn't known how to derive the nonmodular Crane-Yetter invariant from a handle picture.
With the generalised dichromatic invariant,
it is possible to do so.
Through chain mail one can recover a state sum description of the generalised dichromatic invariant $I_F$, whenever $F\colon \mathcal{C} \to \mathcal{D}$ such that $\mathcal{D}$ is modular.
So the generalised dichromatic invariant has a purely combinatorial description in terms of triangulations in that case.
The nonmodular Crane-Yetter invariant will turn out to be a special case.

In general, the state sum model will be useful to understand the physical interpretation of a particular model, while the handle picture is very convenient for calculations.

\subsection{The chain mail process and the generalised 15-j symbol}
Given a four-dimensional manifold $M$ with triangulation $\Delta$,
there is always a
handle decomposition via the following process:
Replace the triangulation by its dual complex, i.e. $4$-simplices $s \in \Delta_4$ by vertices, tetrahedra $t \in \Delta_3$ by edges, triangles $\tau \in \Delta_2$ by polygons and in general $(4-k)$-simplices by $k$-cells.
A $k$-cell, $k\le3$, will then have a valency (the number of adjacent $(k+1)$-cells) of $5-k$, coming from the number of faces of the original simplex.

Then consider the handle decomposition arising from a thickening of this dual complex.
This handle decomposition has $h_0$ 0-handles, where $h_0$ is then the number of 4-simplices in the triangulation, $\Delta_4$.
A handle decomposition of a Kirby diagram has to have only one 0-handle and so is obtained from the previous one by cancelling $(h_0-1)$ 0-1-handle pairs.
As in \cite[Section 4.3]{Roberts:1995SkeinTheoryTV},
the dichromatic invariant obtained from this Kirby diagram is equal to the adjusted formula obtained by
adding to the link a dotted circle for each of the cancelled 1-handles and multiplying by the overall factor $I_F(S^1 \times S^3)^{1-h_0} = \qdim{\Omega_\mathcal{C}}^{1-h_0}$.
Then the formula has an encircling by $\Omega_\mathcal{D}$ coming from a dotted circle for each tetrahedron in the triangulation. 
Inserting morphism boxes for the two $D^3$s at every 1-handle arising from a tetrahedron disconnects the whole link diagram into pentagram-shaped subdiagrams for every 4-simplex.

To arrive at the pentagram shape, first realise that the boundary of a 4-simplex is $S^3 \cong \R^3 \cup \{\infty\}$.
This is the boundary of the 0-handle to which 1-handles and 2-handles are attached.
Visualise the triangulation of the boundary by arranging four vertices of the 4-simplex as a tetrahedron around the origin and putting the remaining vertex at infinity.
Connecting the first four vertices to the vertex at infinity gives the remaining four tetrahedra.
Now draw one copy of $D^3$ for each tetrahedron (the respective copy belonging to a neighbouring 4-simplex) and connect each pair of $D^3$s with lines from the triangles as 2-handles.
The resulting subdiagram is now a big tetrahedron of $D^3$s with a further $D^3$ in the centre of the tetrahedron.
Project this subdiagram onto the plane, for every 4-simplex, and apply Definition \ref{kirby direct definition}.
After applying an arbitrary isotopy in the plane, the evaluation of such a subdiagram labelled with objects $X_i \in \ob \mathcal{C}$ and morphisms $\iota_i$, $i \in \{0,\dots,4\}$, $j \in \{0,\dots,9\}$ in $\mathcal{D}$ is:

\begin{align}
	\pentagram\left(FX_i,\iota_i\right) &\coloneqq
	\pentagram\left(FX_0,\ldots,FX_9,\iota_0,\ldots,\iota_4\right)\nonumber\\
	& \coloneqq \quad \tikzbo{\pic {pentagram};}
\end{align}
The over- and under-braidings follow the convention of Roberts.
It involves a ``splitting convention'' to arrive at a correct blackboard framing,
see \cite[Figure 17]{Roberts:1995SkeinTheoryTV}.

The diagram does not yet correspond to a morphism.
To evaluate it in terms of diagrammatic calculus of the ribbon category $\mathcal{D}$, one has to orient the lines upwards or downwards and insert evaluations and coevaluations as needed, in order to specify where an object or its dual is the source or the target of a morphism.
To arrive at such a choice, fix a total ordering of the vertices.
This ordering induces an orientation on the tetrahedra.
Each tetrahedron occurs as the face of two 4-simplices, which are oriented as submanifolds of $M$, and the tetrahedron inherits two opposite orientations from each of them.
Since a tetrahedron corresponds to a 1-handle, the + and $-$ signs need to be distributed onto the attaching $D^3$s.
Put the + sign on the $D^3$ attaching to the 4-simplex from which the tetrahedron inherits the orientation agreeing with the ordering of the vertices.
Consequently, its morphism is $\iota\colon X_{i_1} \otimes X_{i_2} \otimes X_{i_3} \otimes X_{i_4} \to I$,
while the morphism of the other $D^3$ goes in the other direction.

\subsection{The state sum}
\label{general ssm}
Since the whole diagram is a disconnected sum of diagrams of the above shape, its evaluation will be a product of $\pentagram$-quantities.
Recall Definition \ref{def:colours},
where colours, such as the Kirby colour $\Omega_\mathcal{C}$ are understood in terms of evaluating the diagram as a sum over simple objects.
This sum leads to a state sum formula for $I_F$.
The $X_i$ in the definition of $\pentagram$ are then summands of $F\Omega_\mathcal{C}$, which was labelling the 2-handles.
The $\iota_i$ label the $D^3$s of a 1-handle.
The invariant $I_F$ will then be a big sum over the summands of all these copies of $F\Omega_\mathcal{C}$ and the dual morphism bases.
\begin{defn}
	\label{labelling}
	An $F$-object labelling of the triangulation $\Delta$ is a function
	\begin{align}
		X &\colon \Delta_2 \to \Lambda_\mathcal{C}
	\end{align}
	For a given $F$-object labelling and a total ordering of the vertices $\Delta_0$,
	fix bases of the morphism spaces in the following way:
	For every tetrahedron $t \in \Delta_3$ with vertices $v_0 < v_1 < v_2 < v_3$, denote by $\tau_i$ the face triangle of $t$ where the vertex $v_i$ is left out.
	Now choose dual bases for the space $\mathcal{D}\left(FY(\tau_0) \otimes FY(\tau_2) \otimes FY(\tau_1) \otimes FY(\tau_3), I\right)$ and its dual.

	Then, using the same convention, an $F$-morphism labelling is a function
	\begin{align}
		\iota \colon& \Delta_3 \to \mor \mathcal{D}
	\end{align}
	where $\iota(t)$ is a basis vector of the space $\mathcal{D}\left(FY(\tau_0) \otimes FY(\tau_2) \otimes FY(\tau_1) \otimes FY(\tau_3), I\right)$.
\end{defn}
\begin{defn}
	For given labellings $X$ and $\iota$, define as their \textbf{amplitude} the evaluation of the labelled link diagram:
	\begin{align}
		\left[ X,\iota\right] \coloneqq \prod_{\tau \in \Delta_2} \qdim{X(\tau)} \prod_{s \in \Delta_4} \pentagram(FX(\tau_i), \iota(t_i))
	\end{align}
	Here, the $t_i$ are the faces of $s$ and the $\tau_i$ their faces in turn, in the appropriate order.
	Whenever the orientation of the $D^3$ of a tetrahedron $t_i$ induced from the total ordering matches the face orientation from the 4-simplex, evaluate the $\pentagram$-quantity with the morphism $\iota(t_i)$ and otherwise with its dual basis vector.
	Since every tetrahedron is the face of exactly two 4-simplices, for every morphism $\iota(t)$, its dual will appear exactly once in the labelling, so the sum in the following will indeed range over dual bases.
	
	Note that since the 2-handles are labelled with $F\Omega_\mathcal{C}$, the $\pentagram$-diagram must be labelled with $FX(\tau_i)$.
\end{defn}
	From the normalisation from the multiple 4-simplices (0-handles),
	the evaluation of a Kirby diagram $K$ is:
	\begin{align}
		\left<K(F)\right> & = \qdim{\Omega_\mathcal{C}}^{1-\lvert\Delta_4\rvert} \sum_{\mathclap{\substack{\text{labellings} \\ X, \iota}}} \left[X,\iota \right]
	\end{align}
	This quantity has to be multiplied by the normalisation,
	which is:
	\[ \qdim{\Omega_\mathcal{C}}^{-h_2+h_1}n^{-h_1} = \Omega_\mathcal{C}^{-\lvert \Delta_2 \rvert+\lvert \Delta_3 \rvert}\qdim{(F\Omega_\mathcal{C})'}^{-\lvert \Delta_3 \rvert}\]
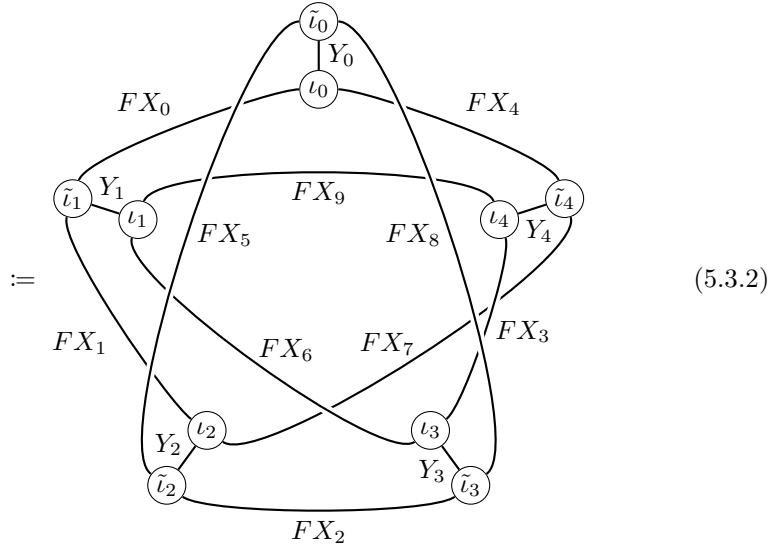
\begin{thm}
	For $F\colon \mathcal{C} \to \mathcal{D}$ being a pivotal functor satisfying the conditions of Definition \ref{definvariant} with $\mathcal{D}$ modular,
	the generalised dichromatic invariant has the following state sum formula:
	\begin{align}
		I_F(M) &= \qdim{\Omega_\mathcal{C}}^{1 - \lvert\Delta_2\rvert + \lvert\Delta_3\rvert - \lvert\Delta_4\rvert}
		\qdim{(F\Omega_\mathcal{C})'}^{-\lvert \Delta_3 \rvert}
		\sum_{\substack{\text{labellings} \\ X, \iota}} \left[X,\iota \right]\nonumber\\
		&= \qdim{\Omega_\mathcal{C}}^{1-\chi(M) + \lvert \Delta_0 \rvert - \lvert \Delta_1 \rvert}
		\qdim{(F\Omega_\mathcal{C})'}^{-\lvert \Delta_3 \rvert}\nonumber\\
		& \quad \cdot \sum_{\substack{\text{labellings} \\ X, \iota}}
		\prod_{\tau \in \Delta_2} \qdim{X(\tau)}
		\prod_{s \in \Delta_4} \pentagram(FX(\tau_i), \iota(t_i))
		\label{ssm formula}
	\end{align}
\end{thm}

\subsection{Trading four-valent for trivalent morphisms}

In order to compare it to the Crane-Yetter model, the state sum needs to be reformulated slightly.
There, the vertices in the $\pentagram$-diagram are trivalent,
which is more convenient when working with $U_qsl(2)$ tilting modules.
The four-valent morphisms appeared when applying Lemma \ref{cutting strands} to the four 2-handles (triangles) going through a 1-handle (tetrahedron) in Proposition \ref{kirby direct definition}.
If one inserts two $\Omega_\mathcal{D}$s instead,
one can produce two trivalent vertices:
\newcommand{\twostrands}[1]{
	\node[morbox] (E\xlabel\ylabel) at (0,-0.4) {#1};
	\coordinate (L\xlabel\ylabel) at (-0.25,0);
	\coordinate (R\xlabel\ylabel) at (0.25,0);
	\draw[thick] (L\xlabel\ylabel) -- (E\xlabel\ylabel) -- (R\xlabel\ylabel);
}
\newcommand{\fourstrands}{
	\foreach \yscale/\ylabel/\ud in {1/a/^,-1/b/_}
	{
		\foreach \xscale/\xlabel/\ij in {1/a/i,-1/b/j}
		{
			\begin{scope}[xscale=\xscale,yscale=\yscale,shift={(-0.5,1.5)}]
				\twostrands{$\iota\ud\ij$}
			\end{scope}
		}
	}
	\foreach \xlabel/\ioffset/\s in {a/0/,b/3/-}
	{
		\foreach \lrlabel/\iformula [evaluate=\iformula as \i] in {L/\ioffset,R/int(\ioffset+(\s1))}
		{
			\node[below] at (\lrlabel\xlabel b) {$F\!X_\i$};
		}
	}
}
\begin{align}
	\pgfkeys{/encircle/label={$\Omega_{\mathcal{D}}$}}
	\tikzbo{
		\begin{scope}[xscale=1.8]
		\encircle{
			\foreach \i in {0,...,3}
			{
				\draw[diagram,directed] ($(-0.75,-1.5)+\i*(0.5,0)$) node[below] {$F\!X_\i$} -- +(0,3);
			}
		}
		\end{scope}
	}
	&=
	\sum_{\mathclap{\substack{\iota_i, \iota_j \\ Y, \tilde Y \in \Lambda_\mathcal{D}}}} \qdim{Y} {\qdim{\tilde Y}}
	\tikzbo{
		\begin{scope}[xscale=1.8]
		\fourstrands
		\encircle[label={$\Omega_{\mathcal{D}}$}]{
			\draw[diagram,directed] (Eab) -- node[left] {$Y$} (Eaa);
			\draw[diagram,directed] (Ebb) -- node[right] {$\tilde Y$} (Eba);
		}
		\end{scope}
	}\nonumber\\
	&=
	\qdim{\Omega_{\mathcal{D}}}\, \sum_{\mathclap{\substack{\iota_i, \iota_j \\ Y \in \Lambda_\mathcal{D}}}} \qdim{Y}
	\tikzbo{
		\begin{scope}[xscale=1.8]
		\fourstrands
		\foreach \yscale/\ylabel/\ud/\omegaplacement/\pole/\direction in {1/a/^/below/south/opdirected,-1/b/_/above/north/directed}
		{
			\begin{scope}[yscale=\yscale]
				\draw[thick,\direction] (Ea\ylabel.\pole) to[out=-90,in=-90] node[\omegaplacement] {$Y$} (Eb\ylabel.\pole);
			\end{scope}
		}
		\end{scope}
	}\label{encircling}
\end{align}
For the last step, Lemma \ref{cutting two strands} has been used,
cancelling the factor $\qdim(\tilde{Y})$.
Note that the additional objects now range over the simple objects in $\mathcal{D}$, not $\mathcal{C}$.

The alternative $\pentagram$-quantity is then defined as:
\begin{align}
	\widetilde\pentagram\left(FX_i,Y_i,\iota_i,\tilde\iota_i\right) &\coloneqq
	\widetilde\pentagram\left(FX_0,\ldots,FX_9,Y_0,\ldots,Y_4,\iota_0,\ldots,\iota_4,\tilde\iota_0,\ldots,\tilde\iota_4\right)\nonumber\\
	& \coloneqq
	\newcommand{\innerradius}{2.5}
	\newcommand{\outerradius}{3.4}
	\tikzset{
		intertwiner/.style={circle,draw,inner sep=0pt,minimum size=5mm},
		loosediag/.style={diagram,looseness=0.3},
		inout2/.style 2 args={in={72*#2},out={180+72*#1}}
	}
	\tikzbo{
		\foreach \angle in {0,1,...,4}
		{
			\node[intertwiner] (iota-\angle) at (90+72*\angle:\innerradius) {$\iota_\angle$};
			\node[intertwiner] (iota-tilde-\angle) at (90+72*\angle:\outerradius) {$\tilde{\iota}_\angle$};
			\draw[thick] (iota-\angle) -- node[rotate=72*\angle,xshift=0.3cm,rotate=-72*\angle] {$Y_\angle$} (iota-tilde-\angle);
		}
		\draw[loosediag] (iota-0) to[inout2={0}{1}] node[above left] {$FX_0$} (iota-tilde-1) (iota-tilde-1) to[inout2={1}{2}] node[below left] {$FX_1$} (iota-2) (iota-2) to[inout2={2}{4}] node[above left=-4pt] {$FX_7$} (iota-tilde-4) (iota-tilde-4) to[inout2={4}{0}] node[above right] {$FX_4$} (iota-0);
		\draw[loosediag] (iota-1) to[inout2={1}{3}] node[above right=-4pt] {$FX_6$} (iota-3) (iota-3) to[inout2={3}{4}] node[right] {$FX_3$} (iota-4) (iota-4) to[inout2={4}{1}] node[below] {$FX_9$} (iota-1);
		\draw[loosediag] (iota-tilde-0) to[inout2={0}{2}] node[right] {$FX_5$} (iota-tilde-2) (iota-tilde-2) to[inout2={2}{3}] node[below] {$FX_2$} (iota-tilde-3) (iota-tilde-3) to[inout2={3}{0}] node[left] {$FX_8$} (iota-tilde-0);
	}
\end{align}
Again, it has to be specified where an object or its dual is the source or the target of a morphism.
Each tetrahedron corresponds to an encirclement.
It occurs as the face of two 4-simplices, which are oriented as submanifolds of $M$, and the tetrahedron inherits two opposite orientations from each of them.
Orient the encircling \eqref{encircling} such that the 4-simplex from which the tetrahedron inherits the orientation agreeing with the ordering of the vertices appears on the top.

Object and morphism labellings now have different definitions than in Section \ref{general ssm}:
\begin{defn}
	An $F$-object labelling of the triangulation $\Delta$ is a pair of functions $(X, Y)$, where
	\begin{align}
		X &\colon \Delta_2 \to \Lambda_\mathcal{C}\\
		Y &\colon \Delta_3 \to \Lambda_\mathcal{D}
	\end{align}
	Choose dual bases for the spaces $\mathcal{D}\left(FX(\tau_0) \otimes FX(\tau_2), Y(t)\right)$ and $\mathcal{D}(FX(\tau_1) \allowbreak \otimes FX(\tau_3),\allowbreak Y(t))$ and their duals.

	An $F$-morphism labelling is a pair of functions $(\iota, \tilde\iota)$
	\begin{align}
		\iota, \tilde\iota \colon \Delta_3 \to \mor \mathcal{D}
	\end{align}
	where $\iota(t)$ is a basis vector of the space $\mathcal{D}\left(FX(\tau_0) \otimes FX(\tau_2), Y(t)\right)$ and $\tilde\iota(t)$ is a basis vector of $\mathcal{D}\left(FX(\tau_1) \otimes FX(\tau_3), Y(t)\right)$.
\end{defn}
\begin{defn}
	For given labellings $(X,Y)$ and $(\iota, \tilde\iota)$, the amplitude is:
	\begin{align}
		\left<(X,Y),(\iota, \tilde\iota)\right> & \coloneqq
		\prod_{\tau \in \Delta_2} \qdim{X(\tau)} 
		\prod_{t \in \Delta_3} \qdim{Y(t)} \qdim{\Omega_\mathcal{D}} \nonumber\\
		& \quad \cdot \prod_{s \in \Delta_4} \widetilde\pentagram(FX(\tau_i), Y(t_i), \iota(t_i), \tilde\iota(t_i))
	\end{align}
\end{defn}
\begin{lem}
	From the Killing property and the normalisation from the multiple vertices, the evaluation of the special framed link $L$ associated to the triangulation is:
	\begin{align}
		\left<L\left(\Omega_\mathcal{D},F\Omega_\mathcal{C}\right)\right> & = \qdim{\Omega_\mathcal{C}}^{1-\lvert\Delta_4\rvert} \sum_{\mathclap{\substack{\text{labellings} \\ (X, Y), (\iota,\tilde\iota)}}} \left<(X,Y),(\iota, \tilde\iota)\right>
	\end{align}
\end{lem}
\begin{thm}
	Using the original Definition \eqref{the invariant},
	the state sum formula can also be written as:
	\begin{align}
		I_F(M) &= \qdim{\Omega_{\mathcal{C}}}^{1 - \lvert\Delta_2\rvert + \lvert\Delta_3\rvert - \lvert\Delta_4\rvert} \qdim{\Omega_{\mathcal{D}}}^{-\lvert\Delta_3\rvert} \qdim{\left(F\Omega_{\mathcal{C}}\right)'}^{-\lvert\Delta_3\rvert}
		\sum_{\mathclap{\substack{\text{labellings} \\ (X, Y), (\iota,\tilde\iota)}}} \left<(X,Y),(\iota, \tilde\iota)\right>\nonumber\\
		&= \qdim{\Omega_{\mathcal{C}}}^{1 - \chi(M) + \lvert\Delta_0\rvert - \lvert\Delta_1\rvert} \qdim{\left(F\Omega_{\mathcal{C}}\right)'}^{-\lvert\Delta_3\rvert} \nonumber\\
		& \quad \cdot \sum_{\substack{\text{labellings} \\ (X, Y), (\iota,\tilde\iota)}} \Bigg(
		\prod_{\tau \in \Delta_2} \qdim{X(\tau)}
		\prod_{t \in \Delta_3} \qdim{Y(t)} \nonumber\\
		& \quad \mathbin{\hphantom{\sum_{\substack{\text{labellings} \\ (X, Y), (\iota,\tilde\iota)}}\Bigg(}}
		\cdot \prod_{s \in \Delta_4} \widetilde\pentagram(FX(\tau_i), Y(t_i), \iota(t_i), \tilde\iota(t_i)) \Bigg)
		\label{ssm formula CY style}
	\end{align}
\end{thm}

\section{Examples}
\label{examples invariant}

\subsection{The Crane-Yetter state sum}
\label{CYSSM}

If $F\colon \mathcal{C} \to \mathcal{D}$ is a full inclusion (Petit's dichromatic invariant, Example \ref{dichromatic}) and $\mathcal{D}$ is already modular, the generalised dichromatic invariant simplifies:
\begin{pro}
	\label{CY-case}
	Let $F\colon \mathcal{C} \hookrightarrow \mathcal{D}$ be a full pivotal inclusion of a spherical fusion category into a modular category.
	\begin{enumerate}
		\item $I_F$ depends only on $\mathcal{C}$, with the inherited ribbon structure.
		It will henceforth be denoted as $\widehat{CY}_\mathcal{C}$.
		\item $\widehat{CY}_\mathcal{C}$ is the Crane-Yetter state sum $CY_\mathcal{C}$ from \cite{CraneYetterKauffman:1997177} for $\mathcal{C}$ up to the Euler characteristic $\chi$:
		\begin{equation}
			\widehat{CY}_\mathcal{C}(M) = CY_\mathcal{C}(M) \cdot \qdim{\Omega_\mathcal{C}}^{1-\chi(M)} \label{proof crane yetter}
		\end{equation}
	\end{enumerate}
\end{pro}
\begin{proof}
	\begin{enumerate}
		\item Since $\mathcal{D}$ is modular,  the simplified definition in Proposition \ref{kirby direct definition} can be used, with $n=1$.
		Object labellings already take values in $\Lambda_\mathcal{C}$.
		Morphism labellings take values in $\mathcal{D}(FX_1 \otimes \cdots \otimes FX_N,I)$, but this is isomorphic to $\mathcal{C}(X_1 \otimes \cdots \otimes X_N,I)$ since $F$ is full.
		The evaluation of the Kirby diagram can thus be carried out in $\mathcal C$ and depends only on data from $\mathcal{C}$ and the  ribbon structure inherited from $\mathcal D$.
		\item In the state sum description, an additional $\Omega_\mathcal{D}$ is inserted in \eqref{encircling} to transform the four-valent vertex into two trivalent vertices,
		introducing additional objects $X$ labelling the tetrahedra.
		Here, this can be achieved instead by using the insertion lemma \ref{insertion lemma}  in $\mathcal C$.
		Thus the labellings of the state sum can be taken to range over $X\colon \Delta_3 \to \Lambda_\mathcal{C}$ and ${\iota, \tilde \iota\colon \Delta_3 \to \mor \mathcal{C}}$.

		A direct comparison of the state sum formula \eqref{ssm formula CY style} to \cite[Theorem 3.2]{CraneYetterKauffman:1997177} shows the equality to $CY_\mathcal{C}$.
		The version of the insertion lemma \ref{insertion lemma} slightly differs by inserting $\Omega_\mathcal{C} = \bigoplus_X \qdim{X} X$ whereas Crane, Yetter and Kauffman insert $\bigoplus_X X$, leading to different dimension factors.
	\end{enumerate}
\end{proof}

\begin{rema}
	Let $\mathcal{C}$ be a ribbon fusion category with braiding $c$.
	Then there is a full inclusion of $\mathcal{C}$ into its Drinfeld centre $\mathcal{Z(C)}$ by mapping an object $X$ to $\left(X, c_{X,-}\right)$.
	So the Crane-Yetter invariant can always be studied as a special case of Petit's dichromatic invariant.
	This is a significant generalisation since the original derivation of the Crane-Yetter state sum from a handlebody picture required $\mathcal{C}$ to be modular, while the version presented here does not.
\end{rema}

\begin{rema}
	Recall that if $\mathcal{D}$ is not modular, but modularisable,
	then the associated state sum model via the modularisation $H$ can be considered.
	But $H \circ F$ will not always be full and may thus fail to give rise to a case of Petit's dichromatic invariant.
	However, if both categories are unitary, Corollary \ref{unitary is CY} can be used to return to a full inclusion,
	but in other cases, a new state sum model might arise.
\end{rema}

\subsubsection{Simply-connected manifolds}

\label{sssection:sc manifolds}

For simply-connected manifolds, the Crane-Yetter invariant reduces to known invariants of the ribbon fusion category.
Recalling the results from Section \ref{ssection:sc manifolds and multiplicativity},
the value for $\CP 2$ is:
\begin{align}
	I\left(\CP 2\right)
	&= \frac{\sum_{X \in \Lambda_\mathcal{C}} \tr \left(\theta_X\right)}{\qdim{\Omega_{\mathcal{C}}}} \nonumber
	\intertext{Since $X$ is simple,
	the morphism $\theta_X$ amounts for multiplying by a complex number,
	which will be denoted by the same symbol:
	}
	&= \frac{\sum_{X \in \Lambda_\mathcal{C}} \qdim{X}^2 \theta_X}{\qdim{\Omega_{\mathcal{C}}}}
\end{align}
The result is also known as the ``normalised Gauss sum'' of the category $\mathcal{C}$.

As another basic example, the manifold $S^2 \times S^2$ has the Hopf link of two 0-framed 2-handles as Kirby diagram, and thus its invariant is:
\begin{align}
	\widehat{CY}_\mathcal{C}\left(S^2 \times S^2\right) &= \frac{\qdim{\Omega_{\mathcal{C}'}}}{ \qdim{\Omega_\mathcal{C}}}
\end{align}
The same value could be calculated from \eqref{eq:simply connected invariant},
but in this case,
the direct calculation is more convenient.

An overview over the Crane-Yetter invariant of several simply-connected 4-mani\-folds is given in Table \ref{tab:CY sc examples}.

\begin{table}[t]
	\begin{center}
		\begin{tabular}{llll}
			\toprule
			Manifold $M$
				& $\widehat{CY}_\mathcal{C}(M)$
					& $\chi(M)$
						& $\sigma(M)$
			\\\midrule
			$\CP 2$
				& $\sum_{X \in \Lambda_\mathcal{C}} \qdim{X}^2 \theta_X \cdot \qdim{\Omega_\mathcal{C}}^{-1}$
					& 3
						& 1
			\\
			$\opCP 2$
				& $\sum_{X \in \Lambda_\mathcal{C}} \qdim{X}^2 \theta_X^{-1} \cdot \qdim{\Omega_\mathcal{C}}^{-1}$
					& 3
						& -1
			\\
			$S^2 \times S^2$
				& $\qdim{\Omega_\mathcal{C}'} \cdot \qdim{\Omega_\mathcal{C}}^{-1}$
					& 4
						& 0
			\\
			$S^2 \tilde{\times} S^2 \cong \CP 2 \connsum \opCP 2$
				& $\qdim{\Omega_\mathcal{C}'} \cdot \qdim{\Omega_\mathcal{C}}^{-1}$
					& 4
						& 0
			\\
			$S^4$ (including exotic candidates)
				& $1$
					& 2
						& 0
			\\\bottomrule
		\end{tabular}
	\end{center}
	\caption{The Crane-Yetter invariant for several simply-connected manifolds.}
	\label{tab:CY sc examples}
\end{table}

\subsection{Non-simply-connected manifolds}
\label{non-simply-connected}
If the four-manifold $M$ is not simply-connected,
then the observation in Lemma \ref{simply-connected}
(that on simply-connected manifolds, the invariant is not stronger than Euler characteristic and signature)
is not applicable any more.
And indeed, already the Crane-Yetter invariant is stronger than the Broda invariant on such manifolds, in that it depends at least on the fundamental group.
This can be seen in the following examples,
and also in the next subsection.

Consider the Crane-Yetter model of a ribbon fusion category $\mathcal{C}$ that is not modular.
This is, up to Euler characteristic and a constant factor,
the generalised dichromatic invariant $\widehat{CY}_\mathcal{C}$ for a full inclusion $F$ of $\mathcal{C}$ into a modular category $\mathcal{D}$.

\subsubsection{Manifolds of the form $S^1 \times M^3$}

Assume for now that our manifold of interest is a product $S^1 \times M$,
for some closed 3-manifold $M$.
Since ${S^1 \times M} = {\partial (D^2 \times M)}$, its signature must be 0.
The Euler characteristic is also ${\chi(S^1 \times M) =} {\chi(S^1) \cdot \chi(M)} = 0$.

Let us study the cases $M = S^3$ and $M= S^1 \times S^2$.
The manifold $S^1 \times S^3$ has a handle decomposition with one 1-handle and no 2-handles and its link diagram in Akbulut notation is the dotted unknot.
Its invariant is therefore:
\begin{align}
	\widehat{CY}_\mathcal{C}\left(S^1 \times S^3\right) &= \qdim{\Omega_\mathcal{C}}
\end{align}
For $S^1 \times S^1 \times S^2$,
a handle decomposition is derived by following \cite[4.3.1, 4.6.8 and 5.4.2]{GompfStipsicz},
and starting from a Heegaard diagram of $S^1 \times S^2$.
It is presented here in the form of a 2-handle attaching curve on the boundary of a solid torus,
which is $\R^2 \cup \{ \infty \}$ with two disks identified.
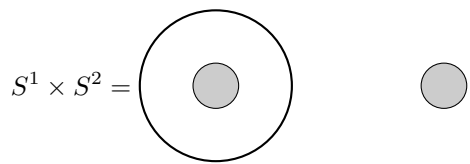
\begin{align*}
	S^1 \times S^2 &=
	\tikzbo{
		\node[handle] at (0,0) {};
		\node[handle] at (3,0) {};
		\node[circle,thick,draw,minimum size=2cm] at (0,0) {};
	}
	\intertext{The two disks are the attaching disks of the 1-handle in $\partial D^3 = S^2 = \R^2 \cup \{\infty\}$.
	The circle is the attaching circle of the 2-handle.
	Thickening this picture gives a Kirby diagram for $I \times S^1 \times S^2$ and adding a further 1- and 2-handle gives:}
	S^1 \times S^1 \times S^2 &=
	\tikzbo{\pic {S1S1S2};}
\end{align*}
The left and the right 3-ball are the attaching balls of the thickened 1-handle, the front and the back ones come from the additional 1-handle.

The simplified definition of the invariant from Proposition \ref{kirby direct definition} is used.
Since there are the same number of 2-handles and 1-handles and $n = 1$, the normalisation is 1, and the invariant evaluates to
\begin{align*}
	\widehat{CY}_\mathcal{C} \left(S^1 \times S^1 \times S^2\right)\quad
	&= \qquad \left<\tikzbo{\pic[scale=0.7,transform shape] {S1S1S2};}\right> \\
	&= \quad \sum_{\mathclap{\substack{X,Y \in \Lambda_\mathcal{C} \\ \iota_i, \iota_j\colon Y^* \otimes Y \to I}}} \qdim{X} \qdim{Y} \qquad
	\tikzbo{
		\pic {S1S1S2eval};
	}\\
	&= \quad \sum_{\mathclap{X,Y \in \Lambda_\mathcal{C}}} \qdim{X} \qdim{Y}^{-1} \qquad
	\tikzbo{
		\pic {s1s1s2insertev};
	}\\
	&= \quad \sum_{\mathclap{X,Y \in \Lambda_\mathcal{C}}} \qdim{X} \qdim{Y}^{-1} \qquad
	\tikzbo{
		\pic {hopflink};
	}\\
	&= \quad \sum_{\mathclap{\substack{X \in \Lambda_\mathcal{C} \\ Y \in \Lambda_\mathcal{C'}}}} \qdim{X} \qdim{Y}^{-1} \qquad
	\tikzbo{
		\pic {unlinkwithtransparent};
	}\\
	&= \quad \sum_{\mathclap{\substack{X \in \Lambda_\mathcal{C} \\ Y \in \Lambda_\mathcal{C'}}}} \qdim{X}^2\quad\\
	&= \quad \lvert\Lambda_{\mathcal{C}'}\rvert \,\qdim{\Omega_\mathcal{C}}.
	\numberthis \label{dichromatic is stronger}
\end{align*}
If $\mathcal{C}$ is not modular, that is, if $\Lambda_{\mathcal{C}'}$ has more than one element, $I(S^1 \times S^1 \times S^2) \neq I(S^1 \times S^3)$.

\subsubsection{Homology and homotopy}

Since $\widehat{CY}$ is multiplicative under connected sum,
one can easily calculate the invariant on a manifold as the following:
\begin{align}
	\widehat{CY}_\mathcal{C}\left(S^1 \times S^3 \connsum S^1 \times S^3 \connsum S^2 \times S^2\right) &= \qdim{\Omega_\mathcal{C}} \qdim{\Omega_{\mathcal{C}'}}
\end{align}
This example is of interest since the latter manifold has the same first homology and signature as $S^1 \times S^1 \times S^2$,
but a different fundamental group.
The Crane-Yetter invariant is sensitive to this difference exactly iff the symmetric centre $\mathcal{C}'$ contains a simple object of dimension greater than 1.
This situation occurs when $\mathcal{C}'$ is equivalent to the representations of a noncommutative finite group.
An overview is given in Table \ref{tab:CY examples}.

\begin{table}
	\begin{center}
		\begin{tabular}{llll}
			\toprule
			Manifold $M$
				& $\widehat{CY}_\mathcal{C}(M)$
					& $H^1(M)$
						& $\pi_1(M)$
			\\\midrule
			$S^1 \times S^3$
				& $\qdim{\Omega_\mathcal{C}}$
					& $\Z$
						& $\Z$
			\\
			$S^1 \times S^1 \times S^2$
				& $\qdim{\Omega_\mathcal{C}} \cdot \lvert\Lambda_{\mathcal{C}'}\rvert$
					& $\Z \oplus \Z$
						& $\Z \oplus \Z$
			\\
			$S^1 \times S^3 \connsum S^1 \times S^3 \connsum S^2 \times S^2$
				& $\qdim{\Omega_\mathcal{C}} \cdot \qdim{\Omega_{\mathcal{C}'}}$
					& $\Z \oplus \Z$
						& $\Z * \Z$
			\\\bottomrule
		\end{tabular}
	\end{center}
	\caption{The Crane-Yetter invariant for three non-simply-connected manifolds with zero Euler characteristic and signature,
	compared to their first homologies and their fundamental group.
	The notation $\Z * \Z$ stands for the free group product of $\Z$ with itself,
	i.e. the free group on two generators.}
	\label{tab:CY examples}
\end{table}

\subsubsection{Refined Broda invariant}

\label{sssec:refined Broda invariant}
An example of the Crane-Yetter invariant is the refined Broda invariant described in Section \ref{petit broda},
where $\mathcal{C}$ is the subcategory of integer spins in a suitable quotient category $\mathcal{D}$ of tilting modules of $U_qsl(2)$,
at an appropriate root of unity.
According to \cite{Roberts1997:RefinedInvariants}, the invariant for any manifold of the form $S^1\times M^3$, with our normalisation, is:
\begin{equation}
	\widehat{CY}_\mathcal{C} = 2^{b_1-1} \qdim{\Omega_\mathcal{C}}
\end{equation}
$b_1$ is the first $\Z_2$-coefficient Betti number of the four-manifold.
A good example occurs for the root $q = \e^{{\im \pi}/{4}}$ (level 2),
when the simple objects are the half-integer spin representations $\Lambda_\mathcal{D} = \left\{0, \frac{1}{2}, 1\right\}$ and $\Lambda_\mathcal{C} = \{0, 1\}$.
In this example, $\mathcal{C} = \mathcal{C}' \simeq \Rep\left(\Z_2\right)$ is symmetric monoidal.
If one takes a different non-trivial root of unity, $\mathcal{C}$ will not be symmetric monoidal any more, but it still has exactly two transparent objects.

Note that our results differ from those reported in \cite{CraneYetterKauffman:1993ClassicalBroda}, where the authors implicitly assumed that $\mathcal{C}$ is modular, which it isn't.

\subsection{Dijkgraaf-Witten models}
\label{dijkgraaf-witten}

The purpose of this section is to show how Dijkgraaf-Witten models are a special case of the Crane-Yetter model, and therefore of Petit's dichromatic invariant.
The construction uses the representations of a finite group.
The same symbol is used for a representation and its underlying vector space.
If $\rho_1$ and $\rho_2$ are representations, then the trivial braiding is the map $c_{\rho_1,\rho_2}(x \otimes y) = y\otimes x$.

\begin{defn}
	Let $F\colon \Rep(G) \hookrightarrow \mathcal{D}$ be a full ribbon inclusion of the representations of a finite group $G$,  with the trivial braiding and trivial twist, into a modular category.
	Then the invariant $I_F$ is called the \textbf{Dijkgraaf-Witten invariant} associated to $G$.
\end{defn}
\begin{rema}
	This choice of name will be justified subsequently.
	Since $F$ is full, $I_F$ can be denoted as $\widehat{CY}_{\Rep(G)}$ and only depends on $G$, as argued in Section \ref{CYSSM}.
	A suitable modular category to embed $\Rep(G)$ is simply the Drinfeld centre.
	Further comments on Dijkgraaf-Witten invariants as Crane-Yetter or Walker-Wang TQFTs are found in Section \ref{TQFT}.
\end{rema}

\begin{defn}
	The \textbf{regular representation} of a finite group $G$ is denoted as $\Co[G]$ and defined as follows:
	The underlying vector space is the free vector space over the set $G$.
	The action of $G$ is defined on the generators by left multiplication.
	
	It is known that $\Co[G] \cong \Omega_{\Rep(G)} \cong \oplus_{\rho} \rho \otimes \Co^{\qdim{\rho}}$ where $\rho$ ranges over the irreducible representations of $G$.
\end{defn}
\begin{defn}
	Every group element $g \in G$ gives rise to a natural transformation of the fibre functor, $\mu(g)_\rho \colon \rho \to \rho$, given by $\mu(g)_\rho(v) = gv$.
	In fact, $\mu$ is a homomorphism.
\end{defn}
The following two lemmas are basic facts of finite group representation theory.
\begin{lem}
	\label{1handlegenerator}
	For any representation $\rho$, there is a projection on the invariant subspace:
	\begin{align}
		\operatorname{inv}_\rho \coloneqq{}& \sum_i \rho \xrightarrow{\iota_i} I \xrightarrow{\iota^i} \rho\\
		={}& \frac{1}{\lvert G \rvert} \sum_g \mu(g)_\rho
	\end{align}
	The $\iota_i$ and $\iota^j$ range over bases with $\iota_i \circ \iota^j = \delta_{i,j}1_I$.
\end{lem}
\begin{lem}
	\label{2handletrace}
	The categorical trace over left multiplication on the regular representation, $\mu(g)_{\Co[G]}$, is proportional to the delta function:
	\begin{equation}
		\tr \left(\mu(g)_{\Co[G]}\right) = \lvert G \rvert \delta (g)
	\end{equation}
\end{lem}

\begin{defn}
	For a finite group $G$, a \textbf{flat $G$-connection} on a topological space $M$ is a homomorphism $\pi_1(M) \to G$.
\end{defn}
\begin{rema}
	\label{Connections-from-handles}
	Only connections on four-manifolds will be considered here.
	Recall from Section \ref{Kirby-calculus-fundamental-group} that the generators of the fundamental group $\pi_1(M)$ are given by the 1-handles, while each 2-handle is a relation word.
	Then a homomorphism $\pi_1(X) \to G$ is a choice of a group element for each 1-handle such that for every 2-handle, the group elements according to its relation word compose to the trivial element.
\end{rema}
The following result shows that this invariant depends only on $\pi_1(M)$.
\begin{thm}
	Let $\Rep(G)$ be the representations of a finite group $G$ with the symmetric braiding and trivial twist.
	Then $\widehat{CY}_{\Rep(G)}(M)$ is the number of flat $G$-connections on $M$.
\end{thm}
\begin{proof}
	The proof is graphical.
	Since $\mathcal D$ is modular, the simplified definition of the invariant from Proposition \ref{kirby direct definition} can be used.
	Since $F$ is full, the invariant can be calculated using objects and morphisms from $\mathcal C$, as in Proposition \ref{CY-case}.
	The morphism $K(F)$ can be manipulated using the coherence axioms of ribbon categories as isotopies of the link in the plane.
	An example is given in Figure \ref{subfig:dw a}, though one should bear in mind that in general there may be more than two 2-handle attaching curves passing along each 1-handle.
	There may also be crossings that cannot be removed by an isotopy alone.
	\begin{figure}
		\centering
		\begin{subfigure}{0.48\textwidth}
			\[
				\sum_{\mathclap{X,Y, \iota_i, \iota_j}} \qdim{X} \qdim{Y} \cdot
				\tikzbo{\pic[encircstyle/.style={thick},scale=0.95] {S1S1S2eval};}
			\]
			\caption{Evaluation of a handle picture of a non-simply-connected manifold.
			(In this example, $S^1 \times S^1 \times S^2$).}
			\label{subfig:dw a}
		\end{subfigure}
		\quad
		\begin{subfigure}{0.48\textwidth}
			\centering
			\[
				\sum_{\mathclap{Y, \iota_i, \iota_j}} \qdim{Y}\lvert G \rvert  \cdot
				\tikzbo{\pic[s1s1s2noencircle,scale=0.95] {S1S1S2eval};}
			\]
			\caption{Remove the 2-handles not attached to any \mbox{1-handles} to give a global factor.}
			\label{subfig:dw b}
		\end{subfigure}
		\medskip
		\begin{subfigure}{0.48\textwidth}
			\[
				\sum_{\mathclap{Y, \iota_i, \iota_j}} \qdim{Y}\lvert G \rvert  \cdot
				\tikzbo{\pic {dw1};}
			\]
			\caption{Rearrange the 1-handles to recognise the projection morphisms.}
			\label{subfig:dw c}
		\end{subfigure}
		\quad
		\begin{subfigure}{0.48\textwidth}
			\centering
			\[
				\lvert G \rvert^{-1} \cdot
				\tikzbo{\pic {dw2};}
			\]
			\caption{1-handles are generators of the fundamental group.
				Trace with the relation words.}
			\label{subfig:dw d}
		\end{subfigure}
		\caption{For the representations of a finite group,
			$\widehat{CY}$ evaluates to the Dijkgraaf-Witten invariant.}
		\label{dw}
	\end{figure}

	Consider any 2-handle in the link picture that is not linked to a 1-handle.
	Since $\Rep(G)$ is symmetric with trivial twist and $F$ is ribbon,
	the knot on the 2-handle, its framing and links to other 2-handles can be undone,
	and then the morphism can be isotoped away.
	All such 2-handles then give a global numerical factor which cancels parts of the normalisation,
	arriving at a diagram that has only 2-handles which start or end in morphisms coming from 1-handles,
	while evaluating to the same value (Figure \ref{subfig:dw b}).
	
	The morphisms on the 1-handles are lined up horizontally and,
	after an isotopy, recognised as the projection morphisms $\operatorname{inv} = \frac{1}{\lvert G \rvert} \sum_g \mu(g)$ defined in Lemma \ref{1handlegenerator}.
	This is shown in Figures \ref{subfig:dw c} and \ref{subfig:dw d}.
	All of the 1-handles then give a morphism $\frac{1}{\lvert G \rvert^{h_1}} \sum_{g_1} \mu\left(g_1\right) \otimes \sum_{g_2} \mu\left(g_2\right) \otimes \cdots \otimes \sum_{g_{h_1}} \mu\left(g_{h_1}\right)$,
	which are traced over with the 2-handles.
	The factor $\frac{1}{\lvert G \rvert^{h_1}}$ is cancelled by the normalisation as well since $\Omega_\mathcal{C} = \lvert G \rvert$.
	
	To perform the trace for each 2-handle, consider Lemma \ref{2handletrace}.
	If the relation word for the 2-handle $k$ is denoted by $r_1r_2\dots{}r_{m_k}$, the trace for $k$ is $\delta\left(g_{r_1}g_{r_2}\cdots{}g_{r_{m_k}}\right)$.
	Again the remaining normalisation is cancelled with the factor $\lvert G \rvert$.
	After tracing out with all 2-handles, the invariant is then
	\begin{align}
		&& \widehat{CY}_{\Rep(G)}(M) &= \sum_{g_1 \in G} \sum_{g_2 \in G} \cdots \sum_{g_{h_1} \in G} \prod_{\text{2-handles }k} \delta\left(g_{r_1}g_{r_2}\cdots{}g_{r_{m_k}}\right)\nonumber\\
		&& \qquad &= \left\lvert\left\{ \phi\colon \pi_1(M) \to G \right\}\right\rvert
	\end{align}
	using Remark \ref{Connections-from-handles}.
\end{proof}
This result shows that $\widehat{CY}_{\Rep(G)}$ is the partition function of a Dijkgraaf-Witten model, described for example in \cite{Yetter:1992rz}.
In the more common normalisation in the literature, one would divide $\widehat{CY}_{\Rep(G)}$ by $\lvert G \rvert = \qdim{\Omega_\mathcal{C}}$, though.

\begin{rema}
	One would expect a four-dimensional Dijkgraaf-Witten model to depend not only on a finite group $G$, but also on a 4-cocycle on $G$.
	The cocycle in the present model is trivial, though.
	A natural way for a 4-cocycle to arise is as a pentagonator in a tricategory.
	But braided categories are a special case of a tricategory with one 1-morphism, and these have a trivial pentagonator, see e.g. \cite{GurskiCheng:PeriodicTableII}.
	Hence, there seems little hope to introduce the data of a 4-cocycle into the representation category of $G$.
	The model would have to be generalised to fully weak monoidal bicategories, for example, following e.g. \cite{Mackaay:Sph2cats}.
\end{rema}
\begin{rema}
	Due to the Doplicher-Roberts reconstruction (see \cite[Paragraph 2.12]{OnBraidedFusionCats} for a categorical approach),
	it is known that symmetric fusion categories with trivial twist are essentially representation categories of finite supergroups.
	If  the dimensions of all objects are required to be positive, the supergroup is in fact a group.
	So the case studied here is not much more restrictive than demanding that $\mathcal{C}$ be a symmetric fusion category.
\end{rema}

\subsection{Invariants from group homomorphisms}
\label{sec:homomorphism}
It is natural to consider generalising the Dijkgraaf-Witten examples by replacing the group $G$ with a homomorphism $\phi\colon P\to G$.
Any homomorphism can be factored into a surjective homomorphism followed by an inclusion, as $P\to \Ima\phi\to G$.
Taking the categories of unitary finite-dimensional representations leads to a functor 
\[ \phi^*\colon\Rep(G)\to\Rep(P)
\]
given by composition with $\phi$.
It factors into functors $A\colon\Rep(G)\to\Rep(\Ima\phi)$ followed by $B\colon\Rep(\Ima\phi)\to\Rep (P)$.
The first functor $A$ is a restriction functor, which is a dominant functor.
This follows from the fact that any $\Ima\phi$-representation $\rho$
is a subobject of $A(\Ind\rho)$,
where $\Ind$ is the induction functor to $P$-representations.
The second functor $B$ is a full inclusion.

\subsubsection{Trivial braiding}
\label{sec:trivial-braiding}
The first case to consider is when $\Rep(P)$ is augmented with the trivial braiding and trivial twist to make it a ribbon category, as in the Dijkgraaf-Witten invariant.
Let $F\colon \Rep(P) \hookrightarrow \mathcal{D}$ be a full ribbon inclusion of $\Rep(P)$ with the trivial ribbon structure into a modular category.

Then the invariant $I_{F\circ\phi^*}$ generalises the Dijkgraaf-Witten invariant in principle but its evaluation is the same as a Dijkgraaf-Witten invariant.
Indeed $F\circ\phi^*=F\circ B\circ A$.
But $A$ is dominant unitary and can be cancelled using Proposition \ref{invariant reduces}, while $F\circ B$ is a full ribbon inclusion of $\Rep(\Ima\phi)$ in $\mathcal D$, and so defines a Dijkgraaf-Witten invariant. 

Despite the fact that the invariant is not new, the construction is still interesting because it may be a starting point for physical models.
Just as in Proposition \ref{CY-case}, the invariant can be calculated in the category $\mathcal \Rep(P)$.
The object labels are simple objects  $X_i\in\Rep(G)$ and the morphism labels are a basis in $\Rep(P)\left(\phi^*X_1\otimes\ldots\otimes \phi^*X_N,I\right)$, or its dual space.
The invariant is evaluated using the representation $p\mapsto\phi^*\mu_{\Co[G]}(p)=\mu_{\Co[G]}(\phi(p))$ with trace
\[ \tr\mu_{\Co[G]}\left(\phi(p)\right) = \lvert G \rvert\delta\left(\phi(p)\right) \]
 using the delta-function in $G$.
 The projection morphisms are 
\[ \frac{1}{\lvert P \rvert}\sum_p\mu\left(\phi(p)\right).\]
Since the functor $F$ is a full inclusion, the multiplicity $n$ is just the multiplicity of $I$ in $\phi^*\Co[G]$.
This can be calculated as $n=\frac{\lvert G \rvert}{\lvert \Ima\phi\rvert}$.
The formula for the invariant is thus
\begin{align}
	\label{eq:teleparallel}
	&&I_{F\circ\phi^*}(M) &=\frac{1}{\lvert \Ker\phi\rvert^{h_1}} \sum_{p_1 \in P} \sum_{p_2 \in P} \cdots \sum_{p_{h_1} \in P} \prod_{\text{2-handles }k} \delta\left(\phi(p_{r_1}p_{r_2}\cdots{}p_{r_{m_k}})\right)
\end{align}
Immediately, one can see that one can replace the $\delta$-function in $G$ by the one in $\Ima\phi$ without changing the value of the invariant.
Also each group element $\phi(p)$ appears exactly $\lvert \Ker\phi \rvert$ times, cancelling the normalisation.
Thus one sees explicitly that the manifold invariant is the Dijkgraaf-Witten invariant of the subgroup $\Ima\phi\subset G$.

\subsubsection{Non-trivial braiding}

A different construction from a group homomorphism is to consider cases where $\Rep(P)$ is augmented with  a non-trivial braiding.
Then one can consider the invariant $I_{\phi^*}$ directly,
without needing the inclusion into a modular category.
(Of course this also works with the trivial braiding,
but then $I_{\phi^*}$ can be postcomposed with the fibre functor to vector spaces,
Proposition \ref{invariant reduces} can be applied,
and the invariant is equal to $1$.)

\begin{exam}
	If $\phi\colon P\to G$ is injective, then $I_{\phi^*}=I_{1_{\Rep(P)}}$, which is a Broda invariant for the category $\Rep P$ and depends only on the Euler number and signature of the four-manifold.
\end{exam}

\begin{exam}
	\label{ex:surjective}
	If $\phi\colon P\to G$ is surjective, then $I_{\phi^*}$ is a Petit dichromatic invariant.
\end{exam}

Simple examples arise from $P=\Z_n$, the cyclic group of order $n$ with the anyonic braiding \cite[Example 2.1.6]{Majid:book} and the pivotal structure from $\Vect$.
The irreducible representations are one-dimensional and also labelled by $\Z_n$.
The braiding on two irreducibles $k,k'$ is
\[ x\otimes y\mapsto e^{\frac{2\pi i}n kk'} y\otimes x \]
and so the transparent objects are $k=0$, and also $k=n/2$ if $n$ is even.
In the case that $n$ is odd, $\Rep(\Z_n)$ is modular and so the invariant of Example \ref{ex:surjective} only depends on $\Rep(G)$ with its induced ribbon structure.
It is a Crane-Yetter invariant.

There are many more possible braidings \cite{Davydov:1997QuasitriangularCocommutativeHopfAlgs} and it seems an interesting project to explore the corresponding constructions of the invariant and Crane-Yetter models, which is left for future work.

\section{Relations to TQFTs and physical models}
\label{other-models}

This discussion section is written in a more informal style.

The invariants defined in this paper are related to various physical models.
It is not just the value of the invariant that is important but also its construction in terms of data on simplices or handles.
This is because in a physical model one is interested in features that are localised to lower-dimensional subsets, such as boundaries, corners or defects associated to embedded graphs, surfaces or other strata.
In some cases it is possible to identify this data as the discrete version of a field in quantum field theory.
In summary, the same invariant can extend to lower dimensions in different ways.

\subsection{TQFTs from state sum models}

Whenever there is a state sum formula for $I_F$, that is, when $\mathcal{D}$ is modular, it is possible to cast it in the form of a Topological Quantum Field Theory (TQFT) $\mathcal{Z}$, following a standard recipe \cite{TuraevViro:1992865}.
\begin{itemize}
	\item For a boundary manifold $M^3$ with a given triangulation $\Delta$, define the set of labellings $L(M, \Delta)$ exactly like for the state sum model in Definition \ref{labelling}.
	Then define the free complex Hilbert space $Y(M, \Delta) \coloneqq \Co[L(M, \Delta)]$.
	\item For a cobordism $\Sigma^4\colon M_1 \to M_2$ with triangulation $\Delta$,
	the transition amplitude $\braopket{l_1}{U(\Sigma, \Delta)}{l_2}$ is defined for the basis vectors coming from $l_{1,2} \in L\left(M_{1,2}, \Delta\rvert_{1,2}\right)$ via the state sum:
	Sum over all labellings of $\Sigma$ that have $l_1$ and $l_2$ as boundary conditions.
	This gives a linear map $U(\Sigma, \Delta)\colon L\left(M_1, \Delta\rvert_1\right) \to L\left(M_2, \Delta\rvert_2\right)$.
	It is independent of the triangulation in the interior.
	\item $\mathcal{Z}$ assigns to an object $M^3$ the image of $U(I \times M)$.
	These spaces can be identified for different triangulations in a coherent way,
	again using cylinders.
	The resulting vector space is then independent of the triangulation of $M$.
	\item $\mathcal{Z}$ on morphisms $\Sigma$ is defined by the restriction of $U$ to the aforementioned spaces.
	Since a cylinder can always be glued to a cobordism without changing its isomorphism class, this is well-defined.
\end{itemize}

\begin{toricyes}
\subsection{Walker-Wang models and the toric code}
\end{toricyes}
\begin{toricno}
\subsection{Walker-Wang models}
\end{toricno}
\label{TQFT}

By the previous subsection,
Petit's dichromatic invariant $I_F$ for a full inclusion $F\colon \mathcal{C} \hookrightarrow \mathcal{D}$ into a modular category extends to a Topological Quantum Field Theory $\mathcal{Z}$.
More precisely, for a closed cobordism $\Sigma^4$,
\begin{align}
	\mathcal{Z}(\Sigma) &= \frac{\widehat{CY}_\mathcal{C}(\Sigma)}{\qdim{\Omega_\mathcal{C}}^{1-\chi(\Sigma)}}
\end{align}
The denominator $\qdim{\Omega_\mathcal{C}}^{1-\chi(\Sigma)}$ is provided by comparison to the Crane-Yetter state sum \eqref{proof crane yetter}.

It is believed that Walker-Wang TQFTs \cite{WalkerWang} are the Hamiltonian formulation of Crane-Yetter TQFTs.

This would imply that the dimensions of these state spaces for boundary manifolds $M^3$ can be calculated:
\begin{align*}
	\dim \mathcal{Z}(M) &= \tr \eins_{\mathcal{Z}(M)}\\
	&= \mathcal{Z}\left(S^1 \times M\right)\\
	&= \frac{I_\mathcal{C} \left(S^1 \times M\right)}{\qdim{\Omega_\mathcal{C}}}\numberthis
\end{align*}
Non-trivial values of the invariant for manifolds of the form $S^1 \times M^3$ can then be interpreted as dimensions of state spaces of the corresponding TQFT.
Comparing with Section \ref{non-simply-connected} shows that these dimensions can indeed be greater than 1, as in the example of Broda's refined invariant.

As an example, for $M = S^1 \times S^2$, one arrives at $\dim \mathcal{Z}(S^1 \times S^2) = \lvert\Lambda_{\mathcal{C}'}\rvert$.
This result is in excellent agreement with the analysis of Walker-Wang ground state degeneracies in \cite{KeyserlingkBurnellSimon}.
The state space of a TQFT corresponds to the space of ground states of the Hamiltonian.

\begin{toricyes}
A known special case is again the refined Broda invariant, mentioned in Section \ref{non-simply-connected}.
Choosing the deformation parameter $q=\e^{\frac{\im\pi}{4}}$ results in $\mathcal{C} \simeq \Z_2$ and 
the resulting model is the 3+1-dimensional generalisation of the toric code \cite{LevinWen2005} studied in \cite{KeyserlingkBurnellSimon}.
\end{toricyes}

If $\mathcal{C} \simeq \Rep(G)$ for $G$ a finite group, the dimensions can be calculated explicitly, recalling Section \ref{dijkgraaf-witten}:
\begin{align*}
	\frac{\widehat{CY}_{\Rep(G)} \left(S^1 \times M\right)}{\qdim{\Omega_\mathcal{C}}} &= \frac{\left\lvert \left\{ \phi\colon \pi_1\left(S^1 \times M\right) \to G \right\} \right\rvert}{\lvert G \rvert}\\
	&= \frac{\left\lvert \left\{ \phi\colon \Z \times \pi_1(M) \to G \right\} \right\rvert}{\lvert G \rvert}\\
	&= \frac{\left\lvert \left\{ \left(\phi\colon \pi_1(M) \to G, g \in G\right) \lvert \phi = g\phi g^{-1} \right\} \right\rvert}{\lvert G \rvert}\\
	\text{(By Burnside's lemma)} \qquad &= \left\lvert \left\{ \phi\colon \pi_1(M) \to G \right\} / \phi \sim g\phi g^{-1}\right\rvert\numberthis
\end{align*}
The state spaces are thus spanned by conjugacy classes of connections on the boundary manifolds,
as one would expect if $\widehat{CY}_{\Rep(G)}$ extends as a Dijkgraaf-Witten TQFT.

\subsection{Quantum gravity models}
\label{invariants cartangeo}

General relativity can be formulated in terms of connections and so it is natural to construct state sum models, or more generally quantum invariants of manifolds, that are modelled on connections.
Usually the groups are Lie groups,
but their representation categories are not fusion since the number of irreducibles is not finite.
As a toy model therefore one can replace the Lie groups by finite groups to get an easy comparison with some of the invariants constructed above.
A more sophisticated resolution of this problem is to use instead representations of quantum groups at a root of unity, which are indeed fusion categories.
Finite groups are discussed here first and then some comments on the obstruction to using quantum groups in a similar way are made below.

Cartan connections can be thought of as principal $G$-connections that allow only gauge transformations of a subgroup $P \hookrightarrow G$.
One of the motivations for the development of the generalised dichromatic invariant was the hope of arriving at a state sum model that could be interpreted as quantum Cartan geometry.
Since there are formulations of general relativity in terms of Cartan geometry (see e.g. \cite{derek:MacDowell-Mansouri}), this would give an interesting new approach to quantum gravity.
However the constructions in Section \ref{sec:homomorphism} based on an inclusion $P\hookrightarrow G$ do not appear to lead to interesting new models. 

A closely related construction is teleparallel gravity.
This is based on a surjective homomorphism $P\to G$ with kernel $N$.
According to Baez and Wise \cite[theorem 32]{Baezderek:teleparallel} the data for teleparallel gravity is a flat $G$-connection and a 1-form with values in the Lie algebra of $N$.
For them, $P$ is the Poincar\'e group and $N$ the translation subgroup, but here the groups are allowed to be more general.

A flat $G$ connection is easily described as an assignment of an element $g\in G$ to each 1-handle with a relation on each 2-handle, as in the Dijkgraaf-Witten model.
The discrete analogue of the 1-form is the assignment of an element $n\in N$ to each 1-handle, with no relations on this data.
For finite groups, this is exactly the data that is summed over in \eqref{eq:teleparallel}, the invariant associated to the homomorphism $\phi\colon P\to G$ that has kernel $N$.
Two elements $p,p'\in P$ such that $\phi(p)=\phi(p')$ differ by an element $p^{-1}p'\in N$.
This is the discrete analogue of the fact that the difference of two connection forms on a manifold is a 1-form.
Thus the construction in \eqref{eq:teleparallel} is a plausible finite group analogue of a sum over configurations of teleparallel gravity.

\subsubsection{Quantum groups}

Classical geometry works with Lie groups, which have an infinite number of irreducible representations.
One hope would be to use quantum groups at a root of unity as a regularisation.
However, few Lie group homomorphisms carry over to quantum groups.
There are many examples of subgroups of Lie groups, but fewer sub-quantum groups of quantum groups are known.
This is because most Lie group homomorphisms do not preserve the root system of the Lie algebras and thus neither the deformation.
And even for Hopf algebra homomorphisms, the restriction functor is not necessarily pivotal:
\begin{exam}
	As an example of a restriction functor that isn't pivotal, consider the category of tilting modules of $U_qsl(2)$ at an $n$-th root of unity.
	Its simple objects are spins $j \in \{0,\frac{1}{2},\dots{}\}$.
	Recall that $\Co[\Z_n]$ is a sub-Hopf algebra of $U_qsl(2)$.
	Recalling that $S^2 = SU(2)/U(1)$, one would hope that this Hopf algebra inclusion serves as Cartan geometry with a quantum 2-sphere.
	
	The irreducible representations of $\Co[\Z_n]$ are Fourier modes $\dots{}, -1, 0, 1, \dots{}$.
	Consider the restriction functor of representations, $\Res$.
	It is obviously monoidal.
	Then $\Res\left(\frac{1}{2}\right) = -1 \oplus 1$.
	Both summands are invertible and thus have dimensions 1, whereas the quantum dimension of $\frac{1}{2}$ is generally not even an integer.
	Thus $\Res$ does not preserve quantum dimensions and can't be pivotal.
\end{exam}
The crucial problem here is that the inclusion does not map the spherical element of $\Co[\Z_n]$, which is $1$, onto the spherical element of $U_qsl(2)$.
A quantum group homomorphism of spherical quantum groups that preserves the spherical elements always gives rise to a pivotal functor on the representation categories \cite[Example 8.5]{Schaumann:GrayCats}.
However, no such homomorphism that gives rise to an invariant that is not a combination of the previously studied cases is known to the authors.

\subsubsection{Spin foam models}

Spin foam models are state sum models for quantum gravity constructed using representations of a quantum group, originally the ``spins'' of $U_qsl(2)$, hence the name.
Starting with a Crane-Yetter state sum, a popular strategy in spin foam models is to impose constraints on the labels on the triangles and tetrahedra to mimick approaches to gravity as a constrained $BF$-theory \cite{Baez:SFMsBF}.
The unconstrained theory corresponds to the Crane-Yetter state sum, and different quantisation strategies of the classical constraints lead to different constraints, like in the Barrett-Crane \cite{BarrettCrane:19983296} or the EPRL-model \cite{EPRL}.
However, in these models the constraints on objects and morphisms typically spoil the monoidal product and so are not examples of the constructions presented here.
An interesting question is whether it is possible to construct spin foam models of the type considered here, for example a spin foam model for teleparallel gravity.
Such a model would involve studying the question of whether there are interesting quantum group analogues of a surjective homomorphism of groups.

\subsection{Nonunitary theories}
There are two possibilities to arrive at a theory which might be more general than the Crane-Yetter model.
The first is to drop the assumption of the target category being modularisable; however this is a mild assumption which only specialises from supergroups to groups.
Alternatively, when dropping the assumption that the categories are unitary,
Lemma \ref{FPdimensions reduce} is not applicable any more.
To the knowledge of the authors, it is not known whether for a dominant pivotal functor will always satisfy $F\Omega_\mathcal{C} = n \cdot \Omega_\mathcal{D}$, so a counterexample might lead to an invariant that can't be reduced to a Crane-Yetter model.

\subsection{Extended TQFTs}

It is a common assumption that the Crane-Yetter model for modular $\mathcal{C}$ is an invertible four-dimensional extended TQFT.
According to the cobordism hypothesis, it should correspond to an invertible (and therefore fully dualisable) object in a 4-category.
The 4-category in question has as objects braided monoidal categories, as 1-morphisms monoidal bimodule categories (with an isomorphism between left and right action compatible with the braiding), as 2-morphisms linear bimodule categories, and furthermore bimodule functors and natural transformations.

A ribbon fusion category $\mathcal{C}$ acting on itself as a mere fusion category $\mathcal{M}$ from left and right should be an example for a fully dualisable, potentially noninvertible object.
The object is $\mathcal{C}$ itself, while its dualisation data on the 1-morphism level is the bimodule data of $\mathcal{M}$.
Being a fusion category, $\mathcal{M}$ is a bimodule over itself, giving the 2-morphism level of dualisation.
The higher levels of dualisation should correspond to finite semisimplicity.

As has been suggested recently \cite[Section 3.2]{HenriquesPenneysTener:2015CatTracesAndModuleTensorCatsOverBraidedCats},
a good notion of monoidal module structure on a monoidal category $\mathcal{M}$ over a braided category $\mathcal{C}$ is a braided central functor from $\mathcal{C}$ to $\mathcal{M}$,
i.e. a braided functor $F\colon \mathcal{C} \to \mathcal{Z}(\mathcal{M})$.
One would expect that the extended TQFT corresponding to such a bimodule is an extension of our (properly normalised) invariant for $F$, whenever it is also pivotal.
And indeed, the inclusion $\mathcal{C} \to \mathcal{Z}(\mathcal{C})$ yields the Crane-Yetter model for $\mathcal{C}$.

\section{Outlook}
\label{outlook}
The generalised dichromatic invariant is a very versatile invariant in that it contains many previously studied theories as special cases.
Table \ref{table:overview} gives an overview which functors give rise to several special cases.
\begin{table}[!ht]
\begin{tabular}{p{0.35\textwidth}p{0.41\textwidth}p{0.14\textwidth}}
	\toprule
	Model & Pivotal functor $F$ & Discussion\\
	\midrule
	$U_qsl(2)$-Crane-Yetter state sum, Broda invariant & $1_\mathcal{C}\colon \mathcal{C} \to \mathcal{C}$ for $\mathcal{C}$ the tilting modules (spins) of $U_qsl(2)$ & Example \ref{identity}\\\midrule
	\begin{toricno}
	Refined Broda invariant with $q=\e^{\im\pi/4}$
	\end{toricno}
	\begin{toricyes}
	Refined Broda invariant with $q=\e^{\im\pi/4}$, toric code
	\end{toricyes}
	 & Canonical inclusion $\mathcal{C} \hookrightarrow \mathcal{D}$ for $\mathcal{C} \simeq \Rep \Z_2$ generated by spins $\{0,1\}$ and $\mathcal{D}$ all spins $\left\{0,\frac{1}{2},1\right\}$ & Sections \ref{non-simply-connected}, \ref{petit broda} and \ref{TQFT}\\\midrule
	Refined Broda invariant, Crane-Yetter model for integer spins & Canonical inclusion $\mathcal{C} \hookrightarrow \mathcal{D}$ for $\mathcal{C}$ integer spins and $\mathcal{D}$ all spins & Sections \ref{non-simply-connected} and \ref{petit broda}\\\midrule
	Dijkgraaf-Witten TQFT for a finite group $G$ & Any full inclusion of $\Rep(G)$ into a modular category, e.g. canonical inclusion $\Rep(G) \hookrightarrow \mathcal{Z}(\Rep(G))$ & Sections \ref{dijkgraaf-witten} and \ref{TQFT}\\\midrule
	General Crane-Yetter state sum, Walker-Wang TQFT for $\mathcal{C}$ any ribbon fusion category & Any full inclusion of $\mathcal{C}$ into a modular category, e.g. canonical inclusion ${\mathcal{C} \hookrightarrow \mathcal{Z(C)}}$ & Sections \ref{CYSSM} and \ref{TQFT}\\\midrule
	Petit's dichromatic invariant & Any full inclusion $F\colon \mathcal{C} \hookrightarrow \mathcal{D}$ for $\mathcal{C}$ and $\mathcal{D}$ ribbon fusion categories & Example \ref{dichromatic}\\\midrule
	``Generalised dichromatic state sum models'' & Any functor into a modular category & Section \ref{general ssm}\\\bottomrule
\end{tabular}

\caption{Overview of the known special cases of the generalised dichromatic invariant,
up to a factor of the Euler characteristic.}
\label{table:overview}
\end{table}
The generalised dichromatic invariant is at least as strong as the Crane-Yetter invariant, which is stronger than Euler characteristic and signature, although it is not known how strong exactly.
If the additional constraints that the pivotal functor is unitary and the target category is modularisable are imposed,
the generalised dichromatic invariant is exactly as strong as $CY$.
In this situation, an upper bound for the strength of the state sum formula is probably given in \cite{2005FreedmanKitaevNajakSlingerlandWalkerWang4dunitaryTQFTs}:
Unitary four-dimensional TQFTs cannot distinguish homotopy equivalent simply-connected manifolds, or in general, s-cobordant manifolds.
It remains to be demonstrated whether it is possible to construct a stronger, nonunitary TQFT with the present framework.

It is indicated in the literature \cite{WalkerWang} that the Walker-Wang model -- and therefore also $CY$ -- for an arbitrary ribbon fusion category should factor into $CY$ of its modularisation and its symmetric centre.
The former reduces to the signature and the latter has been shown here to depend only on the fundamental group in the case of the symmetric centre being just the representations of a finite group.
With the present framework,
the conjecture can be formulated precisely:
\begin{conj}
	Let $\mathcal{C}$ be a modularisable ribbon fusion category with $\mathcal{C}'$ its symmetric centre and $\widetilde{\mathcal{C}}$ its modularisation.
	Then $\widehat{CY}_\mathcal{C} = \widehat{CY}_{\mathcal{C}'} \cdot \widehat{CY}_{\widetilde{\mathcal{C}}}$.
\end{conj}
The case of supergroups has not been treated here, but one would not expect it to differ much,
except possibly a sensitivity to spin structures in the same manner as in the refined Broda invariant (Section \ref{petit broda}).

The question whether the general case of the framework presented here is stronger than the mentioned special cases still remains open.
Either way, motivated from solid state physics and TQFTs it would still be interesting to study how defects behave in the new models.

\section{Acknowledgements}
The authors wish to thank Bruce Bartlett, Steve Simon, Jamie Vicary, Nick Gurski, Claudia Scheimbauer, Alexei Davydov, André Henriques, Ehud Meir, Ingo Runkel, Christoph Schweigert, David Yetter and the participants of the Oxford ``TQFTea seminar'' and the workshop ``Higher TQFT and categorical quantum mechanics'' at the Erwin Schr\"odinger Institute in Vienna for helpful discussions and correspondence.

\printbibliography

\end{document}